\documentclass[journal]{IEEEtran}
\IEEEoverridecommandlockouts

\newcommand{\IEEEversion}{false}  
\usepackage[normalem]{ulem}      
\makeatletter
\font\redwavefont=lasyb10 scaled 850 
\newcommand{\revise}{%
  \bgroup
  \markoverwith{\lower3.5\p@\hbox{\sixly\textcolor{red}{\redwavefont\char58}}}%
  \ULon
}
\makeatother

\usepackage{tikz}
\usepackage{tkz-euclide}
\usetikzlibrary{arrows.meta}

\usepackage{setspace}
\usepackage{bm}
\usepackage{mathtools}
\usepackage{graphicx}
\usepackage{relsize}
\usepackage{cite}
\usepackage{amsmath}
\usepackage[dvipsnames]{xcolor}
\usepackage{amssymb}
\usepackage{dblfloatfix}
\usepackage{tabularx}
\usepackage{makecell}
\usepackage{lipsum}
\usepackage{float}
\usepackage{subfiles}
\usepackage{epstopdf}
\usepackage{verbatim}
\usepackage{amsmath}
\usepackage{ifthen}
\usepackage{pgfplots}
\usepackage{pgfplotstable}
\usepackage{tikzscale}
\usepackage{acronym}
\usepackage{amsfonts}
\usepackage{tikz-3dplot}
\usepackage{cuted}

\usepackage{capt-of}

\usetikzlibrary{intersections}  
\usepgfplotslibrary{fillbetween}

\usetikzlibrary{spy,backgrounds}
\usetikzlibrary{fit}
\usetikzlibrary{calc}

\usepackage{placeins}
\usepackage{tikzscale}		
\usepackage{tcolorbox}

\usepackage{ifthen}

\usepackage[framemethod=TikZ]{mdframed}	

\usepackage{pgfplots}
  \tcbset{shield externalize} 
  \pgfplotsset{compat=newest}
  \usetikzlibrary{plotmarks}
  \usetikzlibrary{arrows.meta}
  
  \usetikzlibrary{arrows}
  \usetikzlibrary{mindmap,trees}
  \usetikzlibrary{positioning,spy}
  \usetikzlibrary{calc,fadings,decorations.pathreplacing}
  \usetikzlibrary{decorations.pathmorphing,decorations.text}
  \usetikzlibrary{datavisualization.polar}      
  \usepackage{pgfplotstable}
  \usepgfplotslibrary{polar}
  \usetikzlibrary{patterns}

  \usetikzlibrary{plotmarks}
  \usepgfplotslibrary{patchplots}
  \usepackage{grffile}
  \usepackage{amsmath}
  \usepackage{tikzscale}		
  \usepackage{tikz-3dplot} 
  \usetikzlibrary{positioning}
  \usetikzlibrary{arrows}
  \usetikzlibrary{fit,backgrounds}  

\pgfdeclarelayer{behindBackground}
\pgfdeclarelayer{background}
\pgfdeclarelayer{foreground}
\pgfsetlayers{behindBackground,background,main,foreground}   

\tikzset{%
  >=latex,
  inner sep=0pt,%
  outer sep=2pt,%
  mark coordinate/.style={inner sep=0pt,outer sep=0pt,minimum size=3pt,
  fill=black,circle}%
}

\DeclareMathOperator*{\argmax}{arg\,max}

\renewcommand{\baselinestretch}{1} 

\colorlet{veccol}{green!50!black}
\colorlet{projcol}{blue!70!black}
\colorlet{myblue}{blue!80!black}
\colorlet{myred}{red!90!black}
\colorlet{mydarkblue}{blue!50!black}
\tikzset{>=latex} 
\tikzstyle{proj}=[projcol!80,line width=0.08] 
\tikzstyle{area}=[draw=veccol,fill=veccol!80,fill opacity=0.6]
\tikzstyle{vector}=[-stealth,myblue,thick,line cap=round]
\tikzstyle{unit vector}=[->,veccol,thick,line cap=round]
\tikzstyle{dark unit vector}=[unit vector,veccol!70!black]
\usetikzlibrary{angles,quotes} 

\definecolor{RDlightgreen}{RGB}{141 192 69}
\definecolor{RDgreen}{rgb}{0.3647, 0.4275, 0.2667}
\definecolor{RDdarkgreen}{rgb}{0.2196, 0.2196, 0.2196}
\definecolor{RDmaroon}{rgb}{.522,.22,.353} %

\definecolor{IEEEblue}{RGB}{0 98 155}
\definecolor{IEEElightblue}{RGB}{0 181 226}
\definecolor{IEEEturquoise}{RGB}{0 156 166}
\definecolor{IEEEred}{RGB}{186 12 47}
\definecolor{IEEEgreen}{RGB}{0 132 61}
\definecolor{IEEElightgreen}{RGB}{120 190 32}
\definecolor{IEEEorange}{RGB}{225 163 0}
\definecolor{IEEEyellow}{RGB}{255 209 0}
\definecolor{IEEEviolet}{RGB}{152 29 151}
\definecolor{IEEEdarkmaroon}{RGB}{134 31 65}
\definecolor{IEEEdarkorange}{RGB}{232 119 34}

\colorlet{red}{IEEEred}
\colorlet{blue}{IEEEblue}
\colorlet{orange}{IEEEorange}
\colorlet{green}{IEEEgreen}
\definecolor{violet}{RGB}{152 29 151}

\definecolor{sand}{rgb}{0.88235,0.63922,0.00000}
\colorlet{IEEEyellow}{sand}
\colorlet{SFVcolor}{IEEElightblue}

\newcommand{\ist}{\hspace*{.3mm}}

\newcommand{\iist}{\hspace*{1mm}}

\newcommand{\nn}{\nonumber}

\DeclareMathAlphabet{\mathsfbr}{OT1}{cmss}{m}{n}
\SetMathAlphabet{\mathsfbr}{bold}{OT1}{cmss}{bx}{n}
\DeclareRobustCommand{\msf}[1]{%
	\ifcat\noexpand#1\relax\msfgreek{#1}\else\mathsfbr{#1}\fi
}

\makeatletter
\newcommand{\msfgreek}[1]{\csname s\expandafter\@gobble\string#1\endcsname}
\makeatother

\DeclareRobustCommand{\mcal}[1]{%
	\ifcat\noexpand#1\relax\mathnormal{#1}\else\cal{#1}\fi
}
\DeclareRobustCommand{\BM}[1]{%
	\ifcat\noexpand#1\relax\bm{\boldUppercaseItalicGreek{#1}}\else\bm{#1}\fi
}
\makeatletter
\newcommand{\boldUppercaseItalicGreek}[1]{\csname var\expandafter\@gobble\string#1\endcsname}
\makeatother

\newcommand{\rv}[1]{\msf{#1}}
\newcommand{\RV}[1]{\bm{\msf{#1}}}

\newcommand{\V}[1]{\bm{#1}}
\newcommand{\M}[1]{\BM{#1}}

\let\geq\geqslant

\let\leq\leqslant

\let\succeq\succcurlyeq

\providecommand{\IEEEQED}{\IEEEQEDopen}

\newcommand{\E}[0]{\mathbb{E}}

\DeclareMathVersion{timesmath}
\SetSymbolFont{letters}{timesmath}{OML}{ntxmi}{m}{it}
\DeclareMathVersion{timesmathbold}
\SetSymbolFont{letters}{timesmathbold}{OML}{ntxmi}{b}{it}

\makeatletter
\newif\ifAC@uppercase@first%
\def\Aclp#1{\AC@uppercase@firsttrue\aclp{#1}\AC@uppercase@firstfalse}%
\def\AC@aclp#1{%
	\ifcsname fn@#1@PL\endcsname%
	\ifAC@uppercase@first%
	\expandafter\expandafter\expandafter\MakeUppercase\csname fn@#1@PL\endcsname%
	\else%
	\csname fn@#1@PL\endcsname%
	\fi%
	\else%
	\AC@acl{#1}s%
	\fi%
}%
\def\Acp#1{\AC@uppercase@firsttrue\acp{#1}\AC@uppercase@firstfalse}%
\def\AC@acp#1{%
	\ifcsname fn@#1@PL\endcsname%
	\ifAC@uppercase@first%
	\expandafter\expandafter\expandafter\MakeUppercase\csname fn@#1@PL\endcsname%
	\else%
	\csname fn@#1@PL\endcsname%
	\fi%
	\else%
	\AC@ac{#1}s%
	\fi%
}%
\def\Acfp#1{\AC@uppercase@firsttrue\acfp{#1}\AC@uppercase@firstfalse}%
\def\AC@acfp#1{%
	\ifcsname fn@#1@PL\endcsname%
	\ifAC@uppercase@first%
	\expandafter\expandafter\expandafter\MakeUppercase\csname fn@#1@PL\endcsname%
	\else%
	\csname fn@#1@PL\endcsname%
	\fi%
	\else%
	\AC@acf{#1}s%
	\fi%
}%
\def\Acsp#1{\AC@uppercase@firsttrue\acsp{#1}\AC@uppercase@firstfalse}%
\def\AC@acsp#1{%
	\ifcsname fn@#1@PL\endcsname%
	\ifAC@uppercase@first%
	\expandafter\expandafter\expandafter\MakeUppercase\csname fn@#1@PL\endcsname%
	\else%
	\csname fn@#1@PL\endcsname%
	\fi%
	\else%
	\AC@acs{#1}s%
	\fi%
}%
\edef\AC@uppercase@write{\string\ifAC@uppercase@first\string\expandafter\string\MakeUppercase\string\fi\space}%
\def\AC@acrodef#1[#2]#3{%
	\@bsphack%
	\protected@write\@auxout{}{%
		\string\newacro{#1}[#2]{\AC@uppercase@write #3}%
	}\@esphack%
}%
\def\Acl#1{\AC@uppercase@firsttrue\acl{#1}\AC@uppercase@firstfalse}
\def\Acf#1{\AC@uppercase@firsttrue\acf{#1}\AC@uppercase@firstfalse}
\def\Ac#1{\AC@uppercase@firsttrue\ac{#1}\AC@uppercase@firstfalse}
\def\Acs#1{\AC@uppercase@firsttrue\acs{#1}\AC@uppercase@firstfalse}
\robustify\Aclp
\robustify\Acfp
\robustify\Acp
\robustify\Acsp
\robustify\Acl
\robustify\Acf
\robustify\Ac
\robustify\Acs
\makeatother

\makeatletter
\def\underbracex#1#2{\mathop{\vtop{\m@th\ialign{##\crcr
				$\hfil\displaystyle{#2}\hfil$\crcr
				\noalign{\kern3\p@\nointerlineskip}%
				#1\crcr\noalign{\kern3\p@}}}}\limits}

\def\upbracefilla{$\m@th \setbox\z@\hbox{$\braceld$}%
	\bracelu\leaders\vrule \@height\ht\z@ \@depth\z@\hfill 
	\leaders\vrule \@height\ht\z@ \@depth\z@\hfill\bracerd
	\braceld\leaders\vrule \@height\ht\z@ \@depth\z@\hfill
	\kern\p@\vrule \@width\p@\kern\p@\vrule \@width\p@\kern\p@\vrule \@width\p@ 
	$}

\def\upbracefilld{$\m@th \setbox\z@\hbox{$\braceld$}%
	\vrule \@width\p@\kern\p@\vrule \@width\p@\kern\p@\vrule \@width\p@\kern\p@
	\leaders\vrule \@height\ht\z@ \@depth\z@\hfill\braceru$}

\makeatother

\newcommand{\DEBUG}{false}       

\usepackage[final]{pdfpages}    

\newcommand{\wallrefA}[0]{%
  \tikzexternaldisable%
  \tikz[baseline=(char.base)]{%
    \node[rectangle,inner sep=0.5pt,rounded corners=0.65mm,minimum size=1.0em,
          fill=IEEEblue,
          text=white] (char) {\adjustbox{max width=0.65em}{$1$}};
  }%
  \tikzexternalenable%
}

\newcommand{\wallrefD}[0]{%
  \tikzexternaldisable%
  \tikz[baseline=(char.base)]{%
    \node[rectangle,inner sep=0.5pt,rounded corners=0.65mm,minimum size=1.0em,
          fill=IEEElightblue,
          text=white] (char) {\adjustbox{max width=0.65em}{$4$}};
  }%
  \tikzexternalenable%
}

\usepackage[hidelinks]{hyperref}
\usepackage{cite}
\usepackage{amsmath,amssymb,amsfonts}
\usepackage{algorithmic}
\usepackage{graphicx}
\usepackage{textcomp}
\usepackage{xcolor}
\usepackage{cite}
\usepackage{amsmath,amssymb,amsfonts}
\usepackage{algorithmic}
\usepackage{graphicx}
\usepackage{textcomp}
\usepackage{xcolor}
\usepackage{orcidlink}
\usepackage{cuted}          
\usepackage{mathtools}      
\usepackage{enumitem}       
\usepackage{comment}

\usepackage{array, booktabs, xltabular}
\usepackage{colortbl}         
\usepackage{multirow}	      
\usepackage{longtable}	      
\usepackage{nicefrac}	      
\usepackage{adjustbox}        
\usepackage{upgreek}          
\usepackage[per-mode=symbol]{siunitx}    		
\DeclareSIUnit{\dBm}{dBm}	  
\usepackage{balance}          
\usepackage{graphicx}         
\usepackage{placeins}         
\usepackage{comment}

\usepackage[
    indexonlyfirst, 
    section=section,
    nonumberlist,numberedsection
]{glossaries}
\makenoidxglossaries%
\newacronym{2d}{2D}{two-dimensional}
\newacronym{3d}{3D}{three-dimensional}
\newacronym{aoa}{AoA}{angle of arrival}
\newacronym{aod}{AoD}{angle of departure}
\newacronym{awgn}{AWGN}{additive white Gaussian noise}
\newacronym{bp}{BP}{belief propagation}
\newacronym{cdf}{CDF}{cumulative distribution function}
\newacronym{csi}{CSI}{channel state information}
\newacronym{crlb}{CRLB}{Cram\'er--Rao lower bound}
\newacronym{pcrlb}{PCRLB}{posterior CRLB}
\newacronym{dm}{DM}{diffuse multipath}
\newacronym{dmc}{DMC}{diffuse multipath component}
\newacronym{efim}{EFIM}{equivalent FIM}
\newacronym{fg}{FG}{factor graph}
\newacronym[longplural={Fisher information matrices}]{fim}{FIM}{Fisher information matrix}
\newacronym{gpu}{GPU}{graphics processing unit}
\newacronym{iot}{IoT}{Internet of Things}
\newacronym{isac}{ISAC}{integrated sensing and communications}
\newacronym{los}{LoS}{line-of-sight}
\newacronym{mc}{MC}{Monte Carlo}
\newacronym{meb}{MEB}{mapping error bound}
\newacronym{mimo}{MIMO}{multiple-input multiple-output}
\newacronym{miso}{MISO}{multiple-input single-output}
\newacronym{ml}{ML}{maximum likelihood}
\newacronym{dmimo}{D-MIMO}{distributed MIMO}
\newacronym{xlmimo}{XL-MIMO}{extremely large-scale multiple-input multiple-output}
\newacronym{mmse}{MMSE}{minimum mean square error}
\newacronym{mmwave}{mmWave}{millimeter wave}
\newacronym{mpc}{MPC}{multipath component}
\newacronym{mpslam}{MP-SLAM}{multipath-based simultaneous localization and mapping}
\newacronym{mrt}{MRT}{maximum ratio transmission}
\newacronym{mse}{MSE}{mean squared error}
\newacronym{mt}{MT}{mobile terminal}
\newacronym{mva}{MVA}{master virtual anchor}
\newacronym{nlos}{NLoS}{non-LoS}
\newacronym{nc}{NC}{noncoherent}
\newacronym{ncv}{NCV}{nearly constant velocity}
\newacronym{oeb}{OEB}{orientation error bound}
\newacronym{olos}{OLoS}{obstructed LoS}
\newacronym{ota}{OTA}{over-the-air}
\newacronym{pa}{PA}{physical anchor}
\newacronym{bs}{BS}{base station}
\newacronym{peb}{PEB}{position error bound}
\newacronym{pdf}{PDF}{probability density function}
\newacronym{pm}{PM}{physical mobile}
\newacronym{pmf}{PMF}{probability mass function}
\newacronym{pf}{PF}{potential feature}
\newacronym{pg}{PG}{path gain}
\newacronym{pl}{PL}{path loss}
\newacronym{pr}{PBR}{particle-based representation}
\newacronym{ppr}{PR}{potential ray}
\newacronym{prop}{PROP}{proposed algorithm}
\newacronym{rcs}{RCS}{radar cross section}
\newacronym{ref}{REF}{reference algorithm}
\newacronym{rf}{RF}{radio frequency}
\newacronym{rfid}{RFID}{radio frequency identification}
\newacronym{rmse}{RMSE}{root mean square error}
\newacronym{roi}{ROI}{region of interest}
\newacronym{rssi}{RSSI}{received signal strength indicator}
\newacronym{rv}{RV}{random variable}
\newacronym{s-parameter}{S-parameter}{scattering parameter}
\newacronym{sfv}{SFV}{surface feature vector}
\newacronym{siso}{SISO}{single-input single-output}
\newacronym{sir}{SIR}{sequential importance resampling}
\newacronym{sis}{SIS}{sequential importance sampling}
\newacronym{slam}{SLAM}{simultaneous localization and mapping}
\newacronym{smc}{SMC}{specular multipath component}
\newacronym{snr}{SNR}{signal-to-noise ratio}
\newacronym{spa}{SPA}{sum-product algorithm}
\newacronym{tdoa}{TDOA}{time-difference-of-arrival}
\newacronym{toa}{TOA}{time-of-arrival}
\newacronym{tosm}{TOSM}{through – open – short – match}
\newacronym{ue}{UE}{user equipment}
\newacronym{ura}{URA}{uniform rectangular array}
\newacronym{uwb}{UWB}{ultrawideband}
\newacronym{va}{VA}{virtual anchor}
\newacronym{vm}{VM}{virtual mobile}
\newacronym{vna}{VNA}{vector network analyzer}
\newacronym{wb}{WB}{wideband}
\newacronym{wpt}{WPT}{wireless power transfer}
\newacronym{xets}{XETS}{cross exponentially tapered slot}

\newacronym{nzm}{NZM}{nonzero-mean}
\newacronym{zm}{ZM}{zero-mean}

\usepackage{ifthen}
\newcommand{\exportFigures}{true}       

\usepackage{etoolbox}
\usepackage{bibunits}
\defaultbibliographystyle{IEEEtran}

\newcommand{\herm}{\mathsf{H}}

\newcommand{\trp}{\mathsf{T}}

\newcommand{\binsetone}[1]{ \mathbb{B}^{#1}  }
\newcommand{\realset}[2]{ \mathbb{R}^{#1 \times #2}  }
\newcommand{\realsetone}[1]{ \mathbb{R}^{#1}  }
\newcommand{\complexset}[2]{ \mathbb{C}^{#1 \times #2}  }
\newcommand{\complexsetone}[1]{ \mathbb{C}^{#1}  }

\DeclareMathOperator{\tr}{tr}
\DeclareMathOperator{\arctantwo}{arctan2}
\newcommand{\condIndep}{\mathrel{\perp\mspace{-10mu}\perp}} 
\newcommand{\eye}[1]{\mathbf{I}_{\scriptscriptstyle#1}}       

\newcommand{\CN}{{\mathcal{CN}}}

\newcommand{\Nz}[0]{\mathrm{N}_{\scriptscriptstyle \mathrm{z}}}             
\newcommand{\Nfrequency}[0]{\mathrm{N}_{\scriptscriptstyle \mathrm{f}}}     
\newcommand{\Nantennas}[0]{\mathrm{M}}                                      
\newcommand{\Nantennasy}[0]{\mathrm{M}_{\scriptscriptstyle \mathrm{y}}}     
\newcommand{\Nantennasz}[0]{\mathrm{M}_{\scriptscriptstyle \mathrm{z}}}     
\newcommand{\RVNfeatures}[1]{\rv{S}_{\scriptscriptstyle #1}}                
\newcommand{\Nfeatures}[1]{S_{\scriptscriptstyle #1}}                       
\newcommand{\Nfeaturest}[1]{\widetilde{S}_{\scriptscriptstyle #1}}                       
\newcommand{\Ncomponents}[0]{ K }                          
\newcommand{\setAnchors}[0]{\mathcal{J}}                                            
\newcommand{\setFeatures}[1]{\mathcal{S}_{\scriptscriptstyle #1}}                   
\newcommand{\setFeaturest}[1]{\widetilde{\mathcal{S}}_{\scriptscriptstyle #1}}      
\newcommand{\setFeaturesLegacy}[1]{\underline{\mathcal{S}}_{\scriptscriptstyle #1}} 
\newcommand{\setFeaturesLegacyt}[1]{\widetilde{\underline{\mathcal{S}}}_{\scriptscriptstyle #1}} 
\newcommand{\setFeaturesNew}[1]{\overline{\mathcal{S}}_{\scriptscriptstyle #1}}     
\newcommand{\RVstate}[1]{\RV{x}_{\scriptscriptstyle #1}}        
\newcommand{\state}[1]{\V{x}_{\scriptscriptstyle #1}}           
\newcommand{\stateHat}[1]{\widehat{\V{x}}_{\scriptscriptstyle #1}}  
\newcommand{\RVpos}[1]{\RV{p}_{\!\scriptscriptstyle #1}}        
\newcommand{\pos}[1]{\V{p}_{\!\scriptscriptstyle #1}}           
\newcommand{\posHat}[1]{\hat{\V{p}}_{\!\scriptscriptstyle #1}}           
\newcommand{\RVvel}[1]{\RV{v}_{\!\scriptscriptstyle #1}}        
\newcommand{\vel}[1]{\V{v}_{\!\scriptscriptstyle #1}}        
\newcommand{\RVPFy}[2]{\RV{y}_{\!\scriptscriptstyle #1\!,#2}}    
\newcommand{\PFy}[3]{\V{y}_{\!\scriptscriptstyle #1\!,#2}}       
\newcommand{\PFyn}[1]{\V{y}_{\scriptscriptstyle #1}}                                    
\newcommand{\RVPFyn}[1]{\RV{y}_{\!\scriptscriptstyle #1}}                               
\newcommand{\RVPFynj}[2]{\RV{y}_{\!\scriptscriptstyle #1}}     
\newcommand{\PFynj}[2]{\V{y}_{\!\scriptscriptstyle #1}}        
\newcommand{\RVPFr}[3]{\rv{r}_{\!\scriptscriptstyle #1,#2}}    
\newcommand{\PFr}[3]{{r}_{\!\scriptscriptstyle #1,#2}}         
\newcommand{\PFrBar}[2]{{r}_{\!\scriptscriptstyle #1,#2}}                               
\newcommand{\existenceProb}[2]{p_{\!\scriptscriptstyle #1,#2}}                          
\newcommand{\existenceProbPRs}[2]{p_{\!\scriptscriptstyle #1,#2}^{\scriptscriptstyle\mathrm{PR}}}                          
\newcommand{\existenceProbPA}[3]{p_{\!\scriptscriptstyle #1,#2}^{\scriptscriptstyle(#3)}}        
\newcommand{\RVPFphi}[2]{\unslant[-.25]{\V{\phi}}_{\!\scriptscriptstyle #1,#2}}        
\newcommand{\PFphiHat}[3]{\widehat{\V{\phi}}_{\!\scriptscriptstyle #1,#2}}                  
\newcommand{\PFphi}[3]{\V{\phi}_{\!\scriptscriptstyle #1,#2}}  
\newcommand{\RVpsfv}[3]{\RV{p}_{\!\scriptscriptstyle #1,#2}^{\text{\tiny sfv}}}      
\newcommand{\psfv}[3]{\V{p}_{\!\scriptscriptstyle #1,#2}^{\text{\tiny sfv}}}         
\newcommand{\psfvHat}[3]{\widehat{\V{p}}_{\!\scriptscriptstyle #1,#2}^{\text{\tiny \,sfv}}}         
\newcommand{\RVPRr}[3]{\rv{r}_{\!\scriptscriptstyle #1,#2}^{\scriptscriptstyle(#3)}}    
\newcommand{\RVPRrj}[2]{\RV{r}_{\!\scriptscriptstyle #1}^{\scriptscriptstyle(#2)}}    
\newcommand{\PRr}[3]{r_{\!\scriptscriptstyle #1,#2}^{\scriptscriptstyle(#3)}}    
\newcommand{\PRrj}[2]{\V{r}_{\!\scriptscriptstyle #1}^{\scriptscriptstyle(#2)}}    
\newcommand{\RVPRrInfrastructure}[2]{\rv{r}_{\!\scriptscriptstyle #1,#2}^{\scriptscriptstyle\mathrm{PR}}}    
\newcommand{\PRrInfrastructure}[2]{r_{\!\scriptscriptstyle #1,#2}^{\scriptscriptstyle\mathrm{PR}}}    
\newcommand{\RVPRrBar}[1]{\RV{r}_{\!\scriptscriptstyle #1}}    
\newcommand{\PRrBar}[1]{\V{r}_{\!\scriptscriptstyle #1}}    

\newcommand{\amplitude}[3]{\rho_{\!\scriptscriptstyle #1,#2}^{\scriptscriptstyle (#3)}}         
\newcommand{\RVamplitude}[3]{\unslant[-.25]{\rho}_{\!\scriptscriptstyle #1,#2}^{\scriptscriptstyle (#3)}}         
\newcommand{\PFmu}[3]{\mu_{\!\scriptscriptstyle #1,#2}}                              
\newcommand{\RVPFmu}[3]{\unslant[-.25]{\mu}_{\!\scriptscriptstyle #1,#2}}           
\newcommand{\PFgamma}[3]{\gamma_{\!\scriptscriptstyle #1,#2}}                       
\newcommand{\RVPFgamma}[3]{\unslant[-.25]{\gamma}_{\!\scriptscriptstyle #1,#2}}     
\newcommand{\observation}[2]{\V{z}_{\scriptscriptstyle #1}^{\scriptscriptstyle(#2)}}         
\newcommand{\RVobservation}[2]{\RV{z}_{\scriptscriptstyle #1}^{\scriptscriptstyle(#2)}}      
\newcommand{\observationn}[1]{\V{z}_{\scriptscriptstyle #1}}                    
\newcommand{\RVobservationn}[1]{\RV{z}_{\scriptscriptstyle #1}}                 
\newcommand{\RVetan}[1]{\unslant[-.25]{\eta}_{\scriptscriptstyle #1}}           
\newcommand{\etan}[1]{{\eta}_{\scriptscriptstyle #1}}                           
\newcommand{\RVetann}[1]{\unslant[-.25]{\bm{\eta}}_{\scriptscriptstyle #1}}                           
\newcommand{\etann}[1]{\bm{\eta}_{\scriptscriptstyle #1}}                           
\newcommand{\etanHat}[1]{\widehat{\eta}_{\scriptscriptstyle #1}}                
\newcommand{\etaXi}[1]{{\eta}_{\scriptscriptstyle\xi,#1}^{\scriptscriptstyle (j)}}                       
\newcommand{\RVnoise}[2]{\RV{n}_{\scriptscriptstyle #1}^{\scriptscriptstyle(j)}}   
\newcommand{\noise}[2]{\V{n}_{\scriptscriptstyle #1}^{\scriptscriptstyle(j)}}   
\newcommand{\elevation}{\theta}                                     
\newcommand{\azimuth}{\vartheta}                                    
\newcommand{\delay}{\tau}     

\newcommand{\steerVec}[1]{\V{\psi}^{\scriptscriptstyle (#1)}}      
\newcommand{\steerVecx}[3]{\V{\psi}_{\!\scriptscriptstyle #1,#2}^{\scriptscriptstyle (#3)}} 

\newcommand{\RVelVecx}[2]{\V{\unslant[-.25]{\elevation}}_{\scriptscriptstyle #1}^{\!\scriptscriptstyle (#2)}}           
\newcommand{\RVazVecx}[2]{\V{\unslant[-.25]{\azimuth}}_{\scriptscriptstyle #1}^{\!\scriptscriptstyle (#2)}}             
\newcommand{\RVdelayVecx}[2]{\V{\unslant[-.25]{\tau}}_{\scriptscriptstyle #1}^{\!\scriptscriptstyle (#2)}}                               

\newcommand{\range}[3]{\bm{r}_{\scriptscriptstyle#1,#2}^{\scriptscriptstyle (#3)}} 
\newcommand{\ranget}[3]{\widetilde{\bm{r}}_{\scriptscriptstyle#1,#2}^{\scriptscriptstyle (#3)}} 
\newcommand{\rangep}[3]{\acute{\bm{r}}_{\scriptscriptstyle#1,#2}^{\scriptscriptstyle (#3)}} 
\newcommand{\rx}[0]{\acute{r}_{\scriptscriptstyle x}}
\newcommand{\ry}[0]{\acute{r}_{\scriptscriptstyle y}}
\newcommand{\rz}[0]{\acute{r}_{\scriptscriptstyle z}}
\newcommand{\houseDet}[1]{ \mathbf{H}_{\scriptscriptstyle #1} }                 
\newcommand{\pwk}[0]{ \bm{p}^\text{\tiny{w}}_{\scriptscriptstyle k} }           
\newcommand{\nw}[0]{ {\bm{n}^\text{\tiny{w}}_{\scriptscriptstyle k}} }          
\newcommand{\fc}{\mathrm{f}_{\text{\tiny c}}}                       
\newcommand{\lightspeed}{\mathrm{c}}                                
\newcommand{\posPA}[1]{ \mathbf{p}_{\scriptscriptstyle\mathrm{pa}}^{\scriptscriptstyle(#1)}  }           
\newcommand{\rotM}[1]{ \mathbf{M}_{\scriptscriptstyle#1}  }         
\newcommand{\origin}[0]{ \mathbf{0}  }                              

\newcommand{\posVAdet}[2]{ \mathbf{p}_{\scriptscriptstyle\mathrm{va},\scriptscriptstyle #1}^{\scriptscriptstyle (#2)} }    
\newcommand{\psfvDet}[1]{\mathbf{p}_{\scriptscriptstyle\mathrm{sfv},\scriptscriptstyle#1} }           
\newcommand{\Rknj}[3]{\bm{R}_{\scriptscriptstyle #1,#2}^{\scriptscriptstyle (#3)}}    

\newcommand{\psurvival}[0]{p_{\scriptscriptstyle\mathrm{S}}}    
\newcommand{\psurvivalPR}[0]{p_{\scriptscriptstyle\mathrm{S}}^{\scriptscriptstyle\mathrm{PR}}}    
\newcommand{\precoveryPR}[0]{p_{\scriptscriptstyle\mathrm{R}}^{\scriptscriptstyle\mathrm{PR}}}    
\newcommand{\transitionmatrix}{{\mathbf{\Phi}}}                         
\newcommand{\processNoiseCov}{{\mathbf{Q}}}                             

\newcommand{\belief}[0]{\widetilde{f}}                                          
\newcommand{\beliefp}[0]{\widetilde{p}}                                          
\newcommand{\Mbeta}[0]{\beta_{\scriptscriptstyle }}                             
\newcommand{\Miota}[1]{\iota_{\scriptscriptstyle }^{\scriptscriptstyle (#1)}}   
\newcommand{\Malpha}[1]{\alpha_{\scriptscriptstyle #1}}                         
\newcommand{\Mkappa}[2]{\kappa_{\scriptscriptstyle #1}^{\scriptscriptstyle (#2)}} 
\newcommand{\Mzeta}[2]{\zeta_{\scriptscriptstyle #1}^{\scriptscriptstyle (#2)}}  
\newcommand{\Momega}[2]{\omega_{\scriptscriptstyle #1}^{\scriptscriptstyle (#2)}}
\newcommand{\Mxi}[0]{\xi_{\scriptscriptstyle }}                                 
\newcommand{\Mnu}[1]{\nu_{\scriptscriptstyle }^{\scriptscriptstyle (#1)}}       
\newcommand{\muiota}[0]{\bm{\mu}_{\scriptscriptstyle n}^{\scriptscriptstyle \iota(j)}}      
\newcommand{\Kiota}[0]{\bm{K}_{\scriptscriptstyle n}^{\scriptscriptstyle \iota(j)}}      
\newcommand{\Ciota}[0]{\bm{C}_{\scriptscriptstyle n}^{\scriptscriptstyle \iota(j)}}         
\newcommand{\mukappa}[0]{\bm{\mu}_{\scriptscriptstyle s,n}^{\scriptscriptstyle \kappa(j)}}      
\newcommand{\Kkappa}[0]{\bm{K}_{\scriptscriptstyle s,n}^{\scriptscriptstyle \kappa(j)}}      
\newcommand{\Ckappa}[0]{\bm{C}_{\scriptscriptstyle\!\!s,n}^{\scriptscriptstyle \kappa(j)}}         
\newcommand{\muomega}[0]{\bm{\mu}_{\scriptscriptstyle s,n}^{\scriptscriptstyle \omega(j)}}      
\newcommand{\Komega}[0]{\bm{K}_{\scriptscriptstyle s,n}^{\scriptscriptstyle \omega(j)}}      
\newcommand{\Comega}[0]{\bm{C}_{\scriptscriptstyle\!\!s,n}^{\scriptscriptstyle \omega(j)}}         
\newcommand{\munu}[0]{\bm{\mu}_{\scriptscriptstyle n}^{\scriptscriptstyle \nu(j)}}          
\newcommand{\Knu}[0]{\bm{K}_{\scriptscriptstyle n}^{\scriptscriptstyle \nu(j)}}      
\newcommand{\Cnu}[0]{\bm{C}_{\scriptscriptstyle n}^{\scriptscriptstyle \nu(j)}}             
\newcommand{\musnj}[1]{\bm{\mu}_{\scriptscriptstyle #1,s,n}^{\scriptscriptstyle (j)}}       
\newcommand{\musnjp}[1]{\bm{\mu}_{\scriptscriptstyle #1,s'\!,n}^{\scriptscriptstyle (j)}}   
\newcommand{\Csnj}[1]{\bm{C}_{\scriptscriptstyle #1,s,n}^{\scriptscriptstyle (j)}}        
\newcommand{\Csnjp}[1]{\bm{C}_{\scriptscriptstyle #1,s'\!,n}^{\scriptscriptstyle (j)}}        
\newcommand{\weight}[2]{w_{\scriptscriptstyle#1}^{\scriptscriptstyle(#2)}}
\newcommand{\weightt}[2]{\widetilde{w}_{\scriptscriptstyle#1}^{\scriptscriptstyle(#2)}}
\newcommand{\particle}[2]{#1^{\scriptscriptstyle(#2)}}
\newcommand{\normConst}[2]{C_{\scriptscriptstyle#1}^{\scriptscriptstyle#2}} 
\usepackage{dsfont} 

\newcommand{\surface}[0]{sfv}	
\newcommand{\dimLocal}[0]{ \mathrm{D}_\text{\tiny{ch}} }            
\newcommand{\dimGlobal}[0]{ \mathrm{D}_\text{\tiny{g}} }            
\newcommand{\dimPhase}[0]{ \mathrm{D}_{\scriptscriptstyle \varphi} }            
\newcommand{\parvec}{\theta}                                     
\newcommand{\etaglobal}[1]{\V{\parvec}^{\text{\tiny{g}}}_{\!\scriptscriptstyle#1}}            
\newcommand{\RVetaglobal}[1]{\unslant[-.25]{\V{\parvec}}^{\text{\!\tiny{g}}}_{\!\scriptscriptstyle#1}}            
\newcommand{\RVetaglobalSmall}[1]{\unslant[-.25]{\scriptstyle\V{\parvec}}^{\text{\!\tiny{g}}}_{\!\scriptscriptstyle#1}}            
\newcommand{\RVetaglobalHat}[1]{\hat{\unslant[-.25]{\bm{\parvec}}}^{\text{\!\tiny{g}}}_{\!\scriptscriptstyle#1}}            
\newcommand{\etach}[2]{\V{\parvec}_{\!\text{\tiny{ch}}\scriptscriptstyle,#1}^{\scriptscriptstyle(#2)}}            
\newcommand{\RVetach}[2]{\unslant[-.25]{\V{\parvec}}_{\text{\tiny{ch}}\scriptscriptstyle,#1}^{\!\scriptscriptstyle(#2)}}           
\newcommand{\RVphasevec}[2]{\unslant[-.25]{\bm{\varphi}}_{\!\scriptscriptstyle#1}^{\!\scriptscriptstyle#2}}   
\newcommand{\RVmodulivec}[2]{\RV{a}_{\scriptscriptstyle#1}^{\scriptscriptstyle#2}}   

\newcommand{\RVpMVAposStacked}[0]{\overline{\RV{p}}^{\text{\tiny\surface}}_{\scriptscriptstyle n}}                   

\newcommand{\PCRLB}[0]{\CRLB_{\scriptscriptstyle n|n}}                 
\newcommand{\FIMstep}[2]{\FIM_{\scriptscriptstyle #1|#2}}     
\newcommand{\FIM}[0]{\V{I}}                              
\newcommand{\CRLB}[0]{\bm{P}}                      
\newcommand{\FIMch}[2]{\FIM_{\scriptscriptstyle \text{\tiny ch},#1}^{\scriptscriptstyle(#2)}}   
\newcommand{\FIMglobal}[1]{\FIM_{\scriptscriptstyle #1}^{\text{\tiny g}}}   
\newcommand{\FIMclassic}[1]{\FIM_{\scriptscriptstyle #1}^{\text{\tiny F}}}   

\newcommand{\jacobian}{ \boldsymbol{J} }
\newcommand{\jacobgn}[1]{ \jacobian_{\!\scriptscriptstyle n}^{\scriptscriptstyle(#1)}}             
\newcommand{\jacobMVAblock}[2]{ \jacobian_{\scriptscriptstyle k,n}^{\scriptscriptstyle \text{\tiny sb},j}}   
\newcommand{\jacobMVAsb}[1]{ \jacobian_{\scriptscriptstyle n,j}^{\scriptscriptstyle \text{\tiny sb},s}}   

\newcommand{\PEB}[0]{\sigma_{\bm{p}_n}}                                      
\newcommand{\MEB}[1]{\sigma_{\bm{p}_{#1}^{\text{\tiny sfv}}}}          

\newcommand{\sA}{\grave{s}}
\newcommand{\sB}{\acute{s}}

\newcommand{\interval}[1]{\hat{I}_{\scriptscriptstyle#1}}				

\ifthenelse{\equal{\exportFigures}{false}}
{
   \newcommand{\tikzexternaldisable}{}  
   \newcommand{\tikzexternalenable}{}   
  }
{
  \usepgfplotslibrary{external}
  \tikzexternalize[prefix=tikz/]
  \usepackage{shellesc}
}

\newcounter{MYtempeqncnt}
\usepackage{amsthm}
\newtheorem{proposition}{Proposition}

\newlength\figureheight
\newlength\figurewidth

\newsavebox{\foobox}
\newcommand{\slantbox}[2][0]{\mbox{%
        \sbox{\foobox}{#2}%
        \hskip\wd\foobox
        \pdfsave
        \pdfsetmatrix{1 0 #1 1}%
        \llap{\usebox{\foobox}}%
        \pdfrestore
}}
\newcommand\unslant[2][-.25]{\slantbox[#1]{$#2$}}

\newcolumntype{R}[1]{>{\raggedleft\arraybackslash}p{#1}}
\newcolumntype{L}[1]{>{\raggedright\arraybackslash}p{#1}}

\usepackage{algorithm2e}		
\RestyleAlgo{ruled}			
\SetKw{KwBy}{by}			%
\SetKwComment{Comment}{\% }{}	

\usepackage[hang,flushmargin]{footmisc}
\makeatletter
\newcommand{\algorithmfootnote}[2][\footnotesize]{%
  \let\old@algocf@finish\@algocf@finish
  \def\@algocf@finish{\old@algocf@finish
    \leavevmode\rlap{\begin{minipage}{\linewidth}
    #1#2
    \end{minipage}}%
    \vspace{-0.3cm} 
  }%
}

\newcolumntype{C}{@{\hskip 0.075cm}c@{\hskip 0.075cm}}

\def\BibTeX{{\rm B\kern-.05em{\sc i\kern-.025em b}\kern-.08em
    T\kern-.1667em\lower.7ex\hbox{E}\kern-.125emX}}

\usepackage{pgfplots}
\usetikzlibrary{arrows.meta}
\pgfplotsset{compat=1.18}
    
\begin{document}
\def\colorNoncohVA{IEEEred}         
\def\colorNoncoh{IEEEblue}          
\def\colorCoh{IEEEgreen}            
\tikzstyle{LineNoncohVA} = [        
            color=\colorNoncohVA, 
            line cap = round, 
            line join=round, 
            line width=\LWestimates,
            mark=*, 
            mark repeat=\markRepeat, 
            mark size=0.7\markSize,
            mark options={solid, line width=0.3pt,fill=\colorNoncohVA}]

\tikzstyle{LineNoncoh} = [color=\colorNoncoh, 
            line cap = round, 
            line join=round, 
            line width=\LWestimates,
            mark=triangle*, 
            mark repeat=\markRepeat, 
            mark size=0.7\markSize,
            mark options={solid, line width=0.3pt,fill=\colorNoncoh}]

\tikzstyle{LineCoh} = [color=\colorCoh, 
            line cap = round, 
            line join=round, 
            line width=\LWestimates,
            mark=diamond*, 
            mark repeat=\markRepeat, 
            mark size=\markSize,
            mark options={solid, line width=0.3pt,fill=\colorCoh}]

\tikzstyle{LineCohNC} = [color=gray, 
            line cap = round, 
            line join=round, 
            line width=\LWestimates,
            mark=Mercedes star,
            mark repeat=\markRepeat, 
            mark size=\markSize,
            mark options={solid, line width=0.3pt,fill=gray}]



\title{\Huge Soft-Coherent Direct Multipath SLAM\\
\ifthenelse{\equal{\IEEEversion}{true}}%
{%
}%
{%
\thanks{
\IEEEauthorrefmark{1}The AMBIENT-6G project has received funding from the Smart Networks and Services Joint Undertaking 
under the European Union's Horizon Europe research and innovation programme under Grant Agreement No. 101192113.}
}
}


\allowdisplaybreaks
\frenchspacing
	\author{\IEEEauthorblockN{Benjamin J.\,B. Deutschmann, Klaus Witrisal, Erik Leitinger}
	
    \IEEEauthorblockA{
	Institute of Comm. Networks and Satellite Comms., Graz University of Technology, Austria \\ 
    }
	\vspace*{-6mm}}

\maketitle

\begin{abstract}
Challenging indoor and urban environments with severe multipath propagation and obstructed line-of-sight degrade classical radio positioning. Multipath-based simultaneous localization and mapping (MP-SLAM) addresses this by building and exploiting propagation maps for robust localization. Emerging distributed multiple-input multiple-output (D-MIMO)/extremely large-scale MIMO (XL-MIMO) infrastructures provide large spatial apertures and high-resolution sensing, especially when phase coherence is maintained across base stations, subarrays, or distributed arrays.

We propose a scalable Bayesian direct MP-SLAM method for coherent data fusion in D-MIMO/XL-MIMO systems that jointly infers the environment while performing robust, high-accuracy localization directly from raw radio signals. While commonly used zero-mean Type-II likelihood functions inherently lead to noncoherent processing across distributed arrays and thus to aperture loss, the proposed phase-preserving nonzero-mean Type-II likelihood shares a complex mean across distributed arrays. This enables coherent fusion and preserves the distributed aperture gain, while the variance captures noncoherent signal power. The method is combined with a surface model that enables map-feature fusion across the distributed infrastructure and supports near-field propagation and visibility effects. Bayesian inference is performed using belief propagation by means of the sum-product algorithm on a factor graph with particle-based messages. Parallelizing over particles and arrays, the GPU-accelerated implementation achieves millisecond-level runtimes even in large or distributed infrastructures. Simulation results show that the proposed method achieves performance gains over existing noncoherent methods and approaches the corresponding posterior CRLB, highlighting the potential of coherent processing for high-resolution sensing and localization.
\end{abstract}


\glsresetall        

\section{Introduction}\label{sec:intro}

Accurate localization and mapping using radio signals is becoming a key component for future wireless systems, supporting applications such as autonomous navigation, land robotics, indoor localization, and 6G \gls{isac}~\cite{GonFurKalValDarSheSheBayWymProcIEEE2024}. In challenging indoor and urban scenarios, classical positioning methods suffer from severe performance degradation due to \gls{olos}, i.e., blockage of \gls{los} paths, and rich multipath propagation. Rather than treating \glspl{mpc} as an impairment, \gls{mpslam} methods exploit reflected and scattered signal components to infer both the \gls{mt} state and geometric properties of the environment \cite{GentnerTWC2016, LeitMeyHlaWitTufWin:TWC2019}. 
This paradigm is particularly attractive for emerging \gls{xlmimo} infrastructures \cite{XuLiuZhaMinCai:TSP2024, XuLiuZhaCaiWu:Arxiv2025,Hua:EUSIPCO2025, WuQuiSunWeiZhaEld:TSP2026, Deutschmann26WCM} and \gls{dmimo} \cite{Fascista23RadioStripesICC, Fascista25RadioStripes, TenWymKesDeySveJCS2026, Deutschmann24SPAWC}, where many spatially distributed antenna panels can jointly form an extremely large synthetic aperture.

\subsection{State of the Art}
Recent progress in radio-based \gls{slam} methods can broadly be divided into two methodological families. 
\emph{Two-step approaches} first extract parametric channel estimates---such as delays, angles, and complex amplitudes of \glspl{mpc}---from the received signals \cite{richter2005estimation, ChuJSTSP2019,GreLeiWitFle:TWC2024}, and subsequently perform Bayesian mapping and tracking \cite{LeitMeyHlaWitTufWin:TWC2019, KimGranSveKimWym:TVT2022, Leitinger23mvaSLAM, GeKalXiaGarAngKimTalValWymSev:TSP2025}. 
Early radio \gls{slam} methods typically represented individual \glspl{mpc} as \glspl{va} or point scatterers and estimated them using random finite sets \cite{GeKalKimJiaTalVakSvenKimWym:JSAC2022, KimGranSveKimWym:TVT2022, GeKalXiaGarAngKimTalValWymSev:TSP2025} or factor graph-based inference \cite{LeitMeyHlaWitTufWin:TWC2019, Leitinger23mvaSLAM,LiCaiLeiTuf:ICC2024}. 
While the two-step paradigm simplifies computation, the intermediate parametric channel estimation step inevitably discards information contained in the complex baseband signals---particularly phase relations across anchors---and introduces measurement-origin uncertainty. 
In contrast, \emph{direct approaches} \cite{ZhaStaJosWanGenDamWymHoeTAES2020, LiaLeiMey:Asilomar2023,LiaMey:Asilomar2024, Deutschmann24SPAWC, LiaLeiMey:TSP2025, Fascista25RadioStripes} operate directly on the raw received signals and embed the physical channel model within the statistical inference engine. 
By jointly performing detection, estimation, and data association in a unified probabilistic model, these methods avoid information loss and enable principled exploitation of low \gls{snr} signal components. 
Direct \gls{mpslam} based on \gls{bp} message passing \cite{LiaLeiMey:Asilomar2023,LiaLeiMey:TSP2025} has recently demonstrated significant robustness improvements in challenging geometries. 
Moreover, this method provides a natural foundation for Bayesian fusion directly at the raw-signal level suitable for coherent processing.

The emergence of \gls{dmimo} \cite{TenWymKesDeySveJCS2026, Fascista25RadioStripes} and \gls{xlmimo} systems \cite{RusekSPM2013, HuRusEdf:TWC2018, GueGuiDarDju:TSP2021} opens the possibility of transforming spatially distributed panels into a synthetic aperture comparable to a large array, provided that phase coherence across anchors can be preserved. 
Exploiting this potential requires statistical models capable of handling several realistic propagation effects, including near-field spherical wavefronts, spatial nonstationarity, partial visibility of \glspl{mpc}, and heterogeneous noise levels across panels. 
In practice, most existing \gls{slam} methods avoid these challenges by performing noncoherent fusion of delay- or angle-based features, thereby losing the potential aperture gain.

From a statistical perspective, coherent fusion requires likelihood models that preserve the complex mean values of \gls{los} paths and \glspl{mpc} across \glspl{bs}, subarrays, or distributed arrays, hereafter collectively termed \glspl{pa}. 
Classical Type-I concentrated likelihood functions \cite{KrimViberg96ASP} retain phase information but require explicit estimation of complex path amplitudes, while zero-mean Type-II likelihood functions obtained by analytical marginalization over the complex amplitudes assuming a hierarchical prior \cite{Tipping2003, GreLeiWitFle:TWC2024,MoeWesVenLei:Fusion2025} have proven to provide robust estimation results, but render the phase information not directly accessible. 
Consequently, commonly used likelihood models either discard phase information a priori or rely on simple stacking of observations, which is incompatible with partial coherence and nonstationary channel effects. 
Designing a robust Type-II likelihood model that preserves the phase information for \gls{dmimo}/\gls{xlmimo} systems therefore remains an open problem. 
Key unresolved challenges include
(i) joint modeling of coherent and noncoherent \glspl{pa} within a unified inference framework,
(ii) incorporation of spherical wavefront effects and nonstationary array responses (partial obstructions of \glspl{mpc}), and
(iii) scalable inference algorithms capable of processing large antenna arrays in real time.

\subsection{Contributions and Paper Organization}
We propose a scalable direct \gls{mpslam} method for coherent data fusion in \gls{dmimo}/\gls{xlmimo} systems that jointly estimates a geometric environment map via surface features while performing robust, high-accuracy sensing and localization for $3$D scenarios, termed \emph{soft-coherent} direct \gls{mpslam}.
The proposed algorithm builds on the direct \gls{mpslam} method in \cite{LiaLeiMey:TSP2025}. 
Our key idea behind this soft-coherent processing is a \emph{nonzero-mean Type-II likelihood} whose variance accounts for ``noncoherent'' signal power, while a complex mean shared among synchronized distributed \glspl{pa} 
enables coherent fusion. 
This construction preserves the global phase structure and retains aperture gain whenever coherence is available across the entire array---making spherical wavefront processing possible and significantly increasing the effective aperture---yet gracefully degrades to noncoherent operation otherwise. 
At the same time, computational complexity is significantly reduced by processing the raw radio signals separately and in parallel for each \gls{pa}.
The method is combined with the recent \gls{sfv} surface model \cite{Leitinger23mvaSLAM, Li25adaptiveDMIMOslam, LiLeiTufMey:Arxiv2026}, which enables the fusion of map features across \glspl{pa} and propagation paths while accounting for partial feature visibility over \glspl{pa}. 
Bayesian inference is performed via particle-based \gls{bp} message passing by means of the \gls{spa} rules on a factor graph. A \gls{gpu}-accelerated implementation enables parallel processing across \glspl{pa} and particles, making real-time operation feasible. The main contributions are summarized as follows.
\begin{itemize}[leftmargin=5mm]
    \item We introduce a phase-preserving nonzero-mean Type-II likelihood for coherent distributed processing that retains aperture gain and supports spherical wavefront processing.
    \item We extend direct \gls{mpslam}~\cite{LiaLeiMey:TSP2025} with the \gls{sfv} model to enable map-feature fusion across distributed \glspl{pa} and propagation paths.
    This accounts for partial obstructions of \glspl{mpc} using \gls{mpc} existences at the \gls{pa}-level (partial obstructions) as well as global \gls{sfv} existences \cite{Li25adaptiveDMIMOslam}.
    \item We develop a new proposal density for potential feature births 
    that 
    improves early detection of map features.    
    \item We compare the proposed coherent algorithm to a noncoherent version, the \gls{ref} from~\cite{LiaLeiMey:TSP2025}, and to the \gls{pcrlb}.
    \item We present a fully \gls{gpu}-parallel implementation achieving an order-of-magnitude speedup over~\gls{ref}.
\end{itemize}

The remainder of the paper is organized as follows.  
Sections~\ref{sec:state-vectors} and~\ref{sec:system-model} introduce the geometric, signal, and statistical models underlying the proposed multipath-based SLAM formulation.
Section~\ref{sec:BP-method} describes the proposed \gls{bp}-based inference algorithm and its particle-based implementation.  
Section~\ref{sec:PCRLB} derives the \gls{pcrlb} used as a performance benchmark.  
Section~\ref{sec:Results} presents numerical experiments and performance evaluations.  
Finally, Section~\ref{sec:conclusion} concludes the paper.

\textit{Notation:} 
Scalars are denoted by lowercase letters $\mathrm{x}$, 
column vectors by bold lowercase letters $\mathbf{x}$, and matrices by bold uppercase letters $\mathbf{X}$. 
Regular upright font is used for deterministic constants.
\Glspl{rv} are typeset in sans serif, upright font, e.g., $\rv{x}$ and $\RV{x}$, and their realizations in serif, italic font, e.g., $x$ and $\V{x}$. 
$f(\V{x})$, shorthand for $f_{\RV{x}}(\V{x})$, denotes the \gls{pdf} of continuous \gls{rv} $\RV{x}$, while $p(\V{r})$, shorthand for $p_{\RV{r}}(\V{r})$, denotes a \gls{pmf} of discrete \gls{rv} $\RV{r}$.
$f(\V{x}|\V{y})$ is the conditional \gls{pdf} of $\RV{x}$ given $\RV{y}$, a shorthand for $f_{\RV{x}|\RV{y}}(\V{x}|\V{y})$.
Calligraphic uppercase letters denote sets, except $\mathcal{N}$, $\mathcal{CN}$, $\mathcal{G}$, and $\mathcal{U}$.
The cardinality of a set $\mathcal{X}$ is $\vert \mathcal{X} \vert$.
We use $\bm{x}^\trp$ and $\bm{x}^\herm$ to denote the transpose and Hermitian transpose of $\bm{x}$, respectively.
The Hadamard product is denoted by $\odot$ and the Kronecker product is denoted by $\otimes$. 
The Dirac delta function is $\delta(\cdot)$ and 
$1_{\mathcal{X}}(x)$ denotes the indicator function, i.e., $1_{\mathcal{X}}(x)\!=\!1$ for $x\!\in\!\mathcal{X}$ and $0$ otherwise.
The $N\!\times\!N$ identity matrix is denoted by $\eye{N}$.
With $\bm{X}$ being an $(M\!\times\!N)$ matrix, $\bm{X}^\dag$ denotes its $(N\!\times\!M)$ Moore-Penrose pseudoinverse.
$\binsetone{}\!\coloneqq\!\{0,1\}$ is the binary set.

\section{Signal Model and Geometrical Model}\label{sec:state-vectors}  

In this work, we consider a multipath channel given by the sum of the \gls{los} component and up to $\Ncomponents$ specular \glspl{mpc} originating from single-bounce paths, which often dominate multipath in indoor channel measurements~\cite{Deutschmann26WCM}.%
\footnote{Following \cite{Li25adaptiveDMIMOslam}, the model can be extended to double-bounce paths.}
Specular \glspl{mpc} are reflections of the \gls{mt}'s signal at large flat surfaces with equal incidence and reflection angles. 
Due to channel reciprocity, a reflection can either be modeled (i) as if it virtually impinges at a \gls{va} that is the mirror image of \gls{pa} $j$ at position $\posPA{j}\!\in\!\realsetone{3}$ mirrored across surface $k$, or (ii) as if it is virtually transmitted from a \gls{vm} position that is a mirror image of the \gls{mt} at position $\mathbf{p}_{\scriptscriptstyle n}\!\in\!\realsetone{3}$ at time $n$ mirrored across surface $k$. 
A \gls{va} is a popular type of point-map feature in \gls{mpslam}~\cite{LeitMeyHlaWitTufWin:TWC2019,GentnerTWC2016, KimGranSveKimWym:TVT2022, LiaLeiMey:TSP2025}. 
It is a \textit{\gls{pa}-local} feature, meaning that a single surface $k$ gives rise to a separate \gls{va} position $\posVAdet{k}{j}\!\in\!\realsetone{3}$ for each \gls{pa} $j$. 
In this work, however, we use the \gls{sfv} model~\cite{Leitinger23mvaSLAM} described by its feature positions.

\subsection{Geometric Model}\label{sec:geometry} 

An \gls{sfv} position $\psfvDet{k}\!\in\!\realsetone{3}\setminus\{\origin\}$ is computed by mirroring the origin $\origin$ of the global Cartesian coordinate system across a specular surface $k$.
It is a \textit{global} point-map feature, meaning that it is equal across the entire infrastructure of \glspl{pa}, allowing data fusion across \glspl{pa} and across multi-bounce paths~\cite{Deutschmann25Asilomar,Li25adaptiveDMIMOslam}.
For $k>0$, the vector
\begin{align}\label{eq:rangep}
    \rangep{k}{n}{j}(\mathbf{p}_{\scriptscriptstyle n},\psfvDet{k},\posPA{j},\rotM{j}) 
    \coloneqq \rotM{j}^{-1} \houseDet{k}(\mathbf{p}_{\scriptscriptstyle n} - \posVAdet{k}{j}) \in \realsetone{3}
\end{align}
points from \gls{pa} $j$ to the \gls{vm} in \textit{local} Cartesian \gls{pa} coordinates, which is the signal source position as perceived by the \gls{pa} taking into account its orientation $\rotM{j}$. 
Here, the vector $\mathbf{p}_{\scriptscriptstyle n}\!-\!\posVAdet{k}{j}$ points from the \gls{va} to the \gls{mt} in global coordinates, while 
the Householder matrix
\begin{align}\label{eq:house}
    \houseDet{k} \coloneqq \eye{3} - 2 \, \frac{{\psfvDet{k}} {\psfvDet{k}}^\trp}{\lVert \psfvDet{k} \rVert^2} \,.
\end{align}
mirrors it across surface $k$, and $\rotM{j}^{-1}$ rotates it into local Cartesian \gls{pa} coordinates.
The \gls{va} position $\posVAdet{k}{j}$ is found through the transformation~\cite{Leitinger23mvaSLAM}
\begin{align}\label{eq:posPA-SB}
\posVAdet{k}{j}  \!\coloneqq\! \posPA{j} \!-\!\left(\! 
        \frac{2  {\posPA{j}}^\trp \psfvDet{k}}{\lVert \psfvDet{k} \rVert^2} \!-\! 1\!
    \right)  \psfvDet{k} \,.
\end{align}
For $k=0$, i.e., the \gls{los} component, we define, for notational compatibility, $\houseDet{0}\!\coloneqq\!\eye{3}$ and $\posVAdet{0}{j}\!\coloneqq\!\posPA{j}$, which yields $\rangep{0}{n}{j}(\mathbf{p}_{\scriptscriptstyle n},\psfvDet{0},\posPA{j},\rotM{j}) \!=\!\rotM{j}^{-1}\big( \mathbf{p}_{\scriptscriptstyle n}\!-\!\posPA{j} \big)$.
See Supplementary Material, Sec.\,S-I for details of this geometric model.

\subsection{Signal Model}\label{sec:signal-model}
We consider a single-antenna \gls{mt}\footnote{The signal and system models, as well as the direct \gls{mpslam} method, extend readily to a MIMO setup with a multi-antenna \gls{mt} (cf.\,\cite{DeutschmannCISA2025}).} 
moving on a trajectory of unknown positions $\mathbf{p}_{\scriptscriptstyle n}$. 
At each time step $n$, an \gls{mt} transmits an uplink pilot signal with spectrum\footnote{We assume a unit-modulus $S(f)$ over the sampled baseband, i.e.,\,$|S(f)|\!=\!1$.} $S(f)$ and bandwidth $B$ at carrier frequency $\fc$, which is received by a set $\setAnchors\!\coloneqq\!\{1\hdots J\}$ of identical \glspl{pa} at known positions $\posPA{j}$ and with orientations modeled by rotation matrices\footnote{Rotation matrices are from the special orthogonal group $SO(3)\!=\!\{\bm{M}\!\in\!\realset{3}{3}|\bm{M} \bm{M}^\trp =\bm{M}^\trp \bm{M}\!=\!\eye{3}, \det(\bm{M})\!=\!1 \}$.}
$\rotM{j}\!\in\!SO(3)$,
each equipped with an $\Nantennas$-antenna \gls{ura}.
After Nyquist filtering, $\Nfrequency$ samples of the received signal waveform are recorded by
\gls{pa} $j$. In the frequency domain, the resulting $\Nfrequency\!=\! B/\Delta_\text{f}+1$
samples have a frequency spacing of $\Delta_\text{f}$. 
The stacked vector $\observation{n}{j}\!\in\!\complexsetone{\Nz}$ with $\Nz\!=\!\Nfrequency\Nantennas$, which collects all signal samples of all antenna elements, reads
\begin{align}\label{eq:observation-generative}
    \observation{n}{j} =    \sum_{k = 0}^{\Ncomponents} 
                            \varrho_{\scriptscriptstyle k,n}^{\scriptscriptstyle(j)}         %
                            \steerVec{j}(\mathbf{p}_{\scriptscriptstyle n},\psfvDet{k})
                            +
                            \noise{n}{j}\,,
\end{align}
where the 
the indices $1\!\leq\! k\!\leq\!\Ncomponents$  indicate the \glspl{mpc} originating from specular surfaces modeled by \glspl{sfv}.
Here, $\varrho_{\scriptscriptstyle k,n}^{\scriptscriptstyle(j)}\!\in\!\complexsetone{}$ denotes an amplitude\footnote{We assume that antenna gain patterns, polarization losses, and reflection losses of \glspl{mpc} are lumped together into the amplitudes $\varrho_{\scriptscriptstyle k,n}^{\scriptscriptstyle(j)}$.}, where components $k$ that are not visible at \gls{pa} $j$ at time $n$ are modeled through $\varrho_{\!\scriptscriptstyle k,n}^{\scriptscriptstyle(j)}\!=\!0$, hence $\Ncomponents$ is fixed.
The noise $\noise{n}{j}$ represents different noise sources potentially including a \gls{dmc} \cite{Fascista25RadioStripes}. 
The steering vectors $\steerVec{j}(\mathbf{p}_{\scriptscriptstyle n},\psfvDet{k})\!\in\!\complexsetone{\Nz}$ are elements of the array manifold obtained by using the \gls{mt} positions $\mathbf{p}_{\scriptscriptstyle n}$ and \gls{sfv} positions $\psfvDet{k}$ to parameterize the array response, which we introduce next. 

\paragraph*{Array Response}
The key to phase-coherent distributed processing with a Type-II likelihood model is the coherent fusion of complex amplitude means across \glspl{pa}, thereby forming an effective array aperture spanning the fused \glspl{pa}. 
This enables near-field (spherical wavefront) processing across the entire array aperture. 

Let $\acute{\V{r}}=[\rx\iist\ry\iist\rz]^\trp$ denote one arbitrary vector in \textit{local} Cartesian coordinates, i.e., the local frame of reference, of one specific \gls{pa}.
For data fusion about the amplitude moduli, we define a path-loss compensated array response
\begin{align}\label{eq:array-response}
    \V{\psi}(\acute{\V{r}})
    &\coloneqq 
    \frac{\lambda}{\sqrt{4 \pi}}
    \frac{1}{\sqrt{4 \pi} \lVert \acute{\V{r}} \rVert}
    \widetilde{\V{\psi}}(\acute{\V{r}})
\end{align}
with wavelength $\lambda=\frac{\lightspeed}{\fc}$ and
propagation velocity $\lightspeed$, and for data fusion about amplitude phases, the carrier phase-based unit-modulus array response is given by
\begin{align}\label{eq:unit-modulus-array-response}
    \widetilde{\V{\psi}}(\acute{\V{r}})
    &\coloneqq \big(\bm{b}(\delay(\acute{\V{r}}))  
        \otimes 
        \bm{a}(\elevation(\acute{\V{r}}),\azimuth(\acute{\V{r}}))\big)
        \exp \!\Big(\!\!-\!\mathrm{j}\frac{2\pi}{\lightspeed}\fc 
        \lVert \acute{\V{r}}\rVert 
        \Big)
\end{align}
which lies on the $\Nz$-torus $\mathbb{T}^{\Nz}\!\coloneqq\!(\mathbb{S}^1)^{\Nz}$ with $\mathbb{S}^1\!\coloneqq\!\{x\!\in\!\mathbb{C}\!:\!|x|\!=\!1\}$ denoting the 
unit circle. 
The array response in \eqref{eq:unit-modulus-array-response} is given by the Kronecker product of the frequency array response in delay $\delay$ 
\begin{align}\label{eq:delay-array-response}
    \bm{b}(\delay) \coloneqq S(\mathbf{f})\odot\exp\big(\!-\!\mathrm{j} 2 \pi \mathbf{f} \delay\big) \quad \in \mathbb{T}^{\Nfrequency}
\end{align}
and the spatial array response in \gls{aoa}~\cite{richter2005estimation,DeutschmannCISA2025}, $\bm{a}(\elevation,\azimuth)\!\coloneqq\! \bm{a}_y(\elevation,\azimuth)\!\otimes\!\bm{a}_z(\elevation)$. 
Due to the \gls{ura} layout, the latter is itself modeled as a Kronecker product of the ``horizontal'' array response
\begin{align}\label{eq:spatial-response-y}
    \bm{a}_y(\elevation,\azimuth) \coloneqq \exp\Big(\mathrm{j} \frac{2 \pi}{\lambda} \mathbf{p}_y \sin(\elevation) \sin(\azimuth) \Big) \quad \in \mathbb{T}^{\Nantennasy}
\end{align}
and the ``vertical'' array response
\begin{align}\label{eq:spatial-response-z}
    \bm{a}_z(\elevation) \coloneqq \exp\Big(\mathrm{j} \frac{2 \pi}{\lambda} \mathbf{p}_z \cos(\elevation) \Big) \quad \in \mathbb{T}^{\Nantennasz}\,.
\end{align}
We use $\mathbf{f}\!\in\!\realsetone{\Nfrequency}$ to denote a vector of $\Nfrequency$ baseband frequencies, a vector $\mathbf{p}_y\!\in\! \realsetone{\Nantennasy}$ of $\Nantennasy$ horizontal positions, and a vector $\mathbf{p}_z\!\in\!\realsetone{\Nantennasz}$ of $\Nantennasz$ vertical positions of a ``template'' \gls{ura}. 
We require each of the support vectors $\{\mathbf{f}, \mathbf{p}_y, \mathbf{p}_z\}$ to be symmetric around $0$.
With this model, we have $\Nz\!=\!\Nfrequency\Nantennasy\Nantennasz$. 
The array responses are parameterized by the \textit{local} spherical \gls{pa} coordinates in propagation delay $\delay(\acute{\V{r}}) = \lVert \acute{\V{r}} \rVert / \lightspeed$, elevation angle $\elevation(\acute{\V{r}}) = \arccos (\rz/ \lVert \acute{\V{r}} \rVert)$ and azimuth angle $\azimuth(\acute{\V{r}}) = \arctantwo (\ry,\rx)$. 
The Kronecker-factorization in~\eqref{eq:unit-modulus-array-response} implies a per-\gls{pa} plane wave assumption, yet spherical wavefront propagation is modeled on the infrastructure level.

We define steering vectors $\steerVec{j}(\mathbf{p}_{\scriptscriptstyle n},\psfvDet{k}) \!\coloneqq \V{\psi}(\rangep{k}{n}{j})$ as elements of the array manifold $\mathcal{M}\!\coloneqq\!\big\{\V{\psi}(\acute{\V{r}})|\acute{\V{r}}\!\in\! \mathbb{R}^3\setminus\{\mathbf{0}\} \big\}\!\subset\!\complexsetone{\Nz}$, obtained by parameterizing the array response in~\eqref{eq:array-response} with $\rangep{k}{n}{j}(\mathbf{p}_{\scriptscriptstyle n},\psfvDet{k},\posPA{j},\rotM{j})$ as defined in~\eqref{eq:rangep}. 
This vector is expressed in the \textit{local} Cartesian coordinate system of \gls{pa} $j$ and points from the \gls{pa} to a (possibly virtual) signal source associated with \gls{mpc} $k$. By absorbing the \gls{pa} position $\posPA{j}$ and orientation $\rotM{j}$, the \gls{pa}-dependent steering vectors $\steerVec{j}(\mathbf{p}_{\scriptscriptstyle n},\psfvDet{k})$ depend only on the \gls{mt} position $\mathbf{p}_{\scriptscriptstyle n}$ and the \gls{sfv} position $\psfvDet{k}$.

\section{System Model and Problem Formulation}\label{sec:system-model}  
In what follows, we introduce the proposed system model and discuss the 
coherent \gls{mpslam} problem formulation.

\subsection{State Vectors and Measurement Model} \label{sec:measurementmodel}
In the \gls{mpslam} problem, the \gls{mt} position and \gls{sfv} positions are unknown and modeled as \glspl{rv} $\RVpos{n}$ and $\RVpsfv{s}{n}{j}$, respectively.
At each time $n$, the \gls{mt} state is defined as $\RVstate{n} \!\coloneqq\! [\RVpos{n}^\trp \iist \RVvel{n}^\trp]^\trp \!\in\! \realsetone{6}$, comprising the \gls{mt} position $\RVpos{n}$ and velocity $\RVvel{n}$. 
The sequence of \gls{mt} states up to time $n$ is denoted by $\RVstate{0:n} \!\coloneqq\! [\RVstate{0}^\trp \iist \cdots \iist \RVstate{n}^\trp]^\trp$. 
Since the number of visible \glspl{mpc}, $\Ncomponents$, is generally unknown, time-varying, and 
\gls{pa}-dependent, we adopt a two-level hierarchical existence model. 
Following~\cite{MeyKroWilLauHlaBraWin:IEEEProc2018, LeitMeyHlaWitTufWin:TWC2019, Li25adaptiveDMIMOslam}, the first layer accounts for the unknown number of \glspl{sfv} at time $n$ by introducing \glspl{pf} indexed by $s \!\in\! \setFeatures{n} \!\coloneqq\! \{1 \ist \hdots \ist \RVNfeatures{n}\}$. 
The number of \glspl{pf}, $\RVNfeatures{n}$, is jointly inferred by the \gls{dmimo} infrastructure and is a global (i.e., \gls{pa}-independent) variable. The complete set is given by $\setFeaturest{n} \!\coloneqq\! \setFeatures{n} \cup \{0\}$, where $s\!=\!0$ corresponds to the \gls{los} path. 

Following~\cite{LiaLeiMey:TSP2025}, the \gls{pf} state is defined as $\RVPFy{s}{n} \!\coloneqq\! [\RVPFphi{s}{n}^{\!\trp} \iist \RVPFr{s}{n}{j}]^\trp$, where $\RVPFphi{s}{n} \!\coloneqq\! [{\RVpsfv{s}{n}{j}}^{\!\trp} \iist \RVPFgamma{s}{n}{j} \iist \RVPFmu{s}{n}{j}]^\trp \!\in\! \realsetone{4} \!\times\! \complexsetone{}$ denotes the continuous state. In contrast to~\cite{LiaLeiMey:TSP2025}, the amplitude is parameterized by a prior variance $\RVPFgamma{s}{n} \!\in\! \realsetone{}_{\scriptscriptstyle\geq 0}$ and a complex prior mean $\RVPFmu{s}{n}{j} \!\in\! \complexsetone{}$, shared across distributed \glspl{pa} $j$ to enable phase-coherent processing. 
The existence of the \gls{pf} $s$ is modeled by a binary random variable $\RVPFr{s}{n}{j} \!\in\! \{0,1\}$, where $\RVPFr{s}{n}{j}\!=\!1$ indicates presence.

In line with \cite{Li25adaptiveDMIMOslam}, the second layer accounts for the unknown and time-varying number of \gls{los}- and \gls{pf}-related propagation paths (e.g., due to partial obstructions or geometric constraints) by introducing, for each \gls{pa}, a \gls{ppr} with binary existence variable $\RVPRr{s}{n}{j} \!\in\! \{0,1\}$. 
For later use, we define the joint \gls{pf} state $\RVPFynj{n}{j} \!\coloneqq\! [\RVPFy{0}{n}^{\trp} \ist \cdots \ist \RVPFy{\Nfeatures{n}}{n}^{\trp}]^\trp$,
the joint \gls{ppr} state $\RVPRrj{n}{j} \!\coloneqq\! [\RVPRr{0}{n}{j} \iist \cdots \iist \RVPRr{\Nfeatures{n}}{n}{j}]^\trp$ and further $\RVPRrBar{n} \!\coloneqq\! [{\RVPRrj{n}{1}}^{\trp} \!\!\iist \cdots \iist {\RVPRrj{n}{J}}^{\trp}]^\trp$.
The vector of stacked observations is $\RVobservationn{n} \!\coloneqq\! [{\RVobservation{n}{1}}^\trp \ist \hdots \ist {\RVobservation{n}{J}}^\trp]^\trp$.

Based on this model, the measurement model in~\eqref{eq:observation-generative} for \gls{pa} $j$ can be reformulated as
\begin{align}\label{eq:observation}
    \observation{n}{j} = 
    \sum_{s\in \setFeaturest{n}} 
    \PFr{s}{n}{j} 
    \PRr{s}{n}{j} 
    \amplitude{s}{n}{j} 
    \steerVecx{s}{n}{j}
    +
    \noise{n}{j}
    \quad
    \in \complexsetone{\Nz}
\end{align}
where $\steerVecx{s}{n}{j}\!\coloneqq\!\steerVec{j}(\state{n},\psfv{s}{n}{j})$, and the complex amplitudes $\amplitude{s}{n}{j}$ are modeled as complex Gaussian random variables with \gls{pf}-specific variance $\RVPFgamma{s}{n}{j}$ and mean $\RVPFmu{s}{n}{j}$, i.e., $\RVamplitude{s}{n}{j}|\RVPFmu{s}{n}{j},\RVPFgamma{s}{n}{j} \!\sim\! \CN(\RVPFmu{s}{n}{j},\RVPFgamma{s}{n}{j})$. 
Conditioned on their prior parameters, the amplitudes $\RVamplitude{s}{n}{j}$ are independent across $s$ and $n$ and i.i.d. across \glspl{pa} $j$.
For simplicity, we assume temporally and spatially uncorrelated circularly-symmetric complex AWGN, i.e., $\RVnoise{n}{j}|\RVetan{n}^{\scriptscriptstyle(j)} \!\sim\! \CN(\mathbf{0},\RVetan{n}^{\scriptscriptstyle(j)} \eye{\Nz})$ with noise variance $\RVetan{n}^{\scriptscriptstyle(j)}$.

Conditioned on $\RVstate{n}$, the joint \gls{pf} state $\RVPFynj{n}{j}$, the joint \gls{ppr} state $\RVPRrj{n}{j}$, and the noise variance $\RVetan{n}^{\scriptscriptstyle(j)}\!$, the measurement $\RVobservation{n}{j}$ is complex Gaussian-distributed, leading to the Type-II likelihood
\begin{align}\label{eq:likelihood-Type-II}
    f(\observation{n}{j}|\state{n},\PFynj{n}{j}, \PRrj{n}{j}\!\!,\etan{n}^{\scriptscriptstyle(j)}) 
    &= \CN(\observation{n}{j};\bm{\mu}_{\scriptscriptstyle n}^{\scriptscriptstyle (j)},
                            \bm{C}_{\scriptscriptstyle n}^{\scriptscriptstyle (j)})
\end{align}
with nonzero mean $\bm{\mu}_{\scriptscriptstyle n}^{\scriptscriptstyle (j)} \!\coloneqq\! \sum_{s\in \setFeaturest{n}} \PFr{s}{n}{j}\PRr{s}{n}{j}  \PFmu{s}{n}{j} \steerVecx{s}{n}{j}$ and covariance matrix $\bm{C}_{\scriptscriptstyle n}^{\scriptscriptstyle (j)} \!\coloneqq\! \etan{n}^{\scriptscriptstyle(j)} \eye{\Nz} + \sum_{s\in \setFeaturest{n}}  \PFr{s}{n}{j} \PRr{s}{n}{j} \PFgamma{s}{n}{j} \steerVecx{s}{n}{j} {\steerVecx{s}{n}{j}}^\herm$. 

\subsection{State-Transition Factors}\label{sec:state-transition-model}  
All states evolve independently according to first-order Markov models, i.e., the \gls{mt} state evolves with transition \gls{pdf} $f(\state{n}|\state{n-1})$, the noise variances with $f(\etan{n}^{\scriptscriptstyle(j)}|\etan{n-1}^{\scriptscriptstyle(j)})$, the \glspl{pf} with $f(\PFy{s}{n}{1}|\PFy{s}{n-1}{J})$, and the \glspl{ppr} with $p(\PRr{s}{n}{j}|\PRr{s}{n\!-1}{j})$.

\subsubsection{Legacy \glspl{pf}} 
The state-transition \gls{pdf} of a legacy \gls{pf} $s$ from time $n\!-\!1$ to $n$, when conditioned on its nonexistence at the previous time step $n\!-\!1$, i.e., $\PFr{s}{n\!-\!1}{J}\!=\!0$, is
\begin{align}\label{eq:nonexistence-legacy-time}
    f(\PFphi{s}{n}{1},\PFr{s}{n}{1}|\PFphi{s}{n\!-\!1}{J},
    0) = 
    \begin{cases}
        0,                              & \PFr{s}{n}{1}\!=\!1 \\
        f_\mathrm{d}(\PFphi{s}{n}{1}),  & \PFr{s}{n}{1}\!=0\! 
    \end{cases}
\end{align}
where  
$f_\mathrm{d}(\PFphi{s}{n}{1})$ is a ``dummy'' \gls{pdf}~\cite{MeyKroWilLauHlaBraWin:IEEEProc2018} ensuring that 
$\sum_{\PFr{s}{n}{1} \in \{0,1\}} \int f(\PFphi{s}{n}{1},\PFr{s}{n}{1}|\PFphi{s}{n-1}{J},\PFr{s}{n-1}{J})\mathrm{d}\PFphi{s}{n}{1}\!=\!1$.
Equation~\eqref{eq:nonexistence-legacy-time} implies that a \gls{pf} that has not existed at time step $n\!-\!1$ cannot exist as legacy \gls{pf} at time $n$.

However, if it existed at time step $n\!-\!1$, it continues to exist at time $n$ with a survival probability $\psurvival$.
The state-transition \gls{pdf} from time step $n\! -\!1$ to time $n$, when conditioned on the existence of the feature $s$ at the previous time step $n\!-\!1$, i.e., $\PFr{s}{n\!-\!1}{J}\!=\!1$, is~\cite[eq.\,(8)]{Leitinger23mvaSLAM}
\begin{align}\label{eq:existence-legacy-time}
    f(\PFphi{s}{n}{1},\PFr{s}{n}{1}|\PFphi{s}{n\!-\!1}{J},
    1) = 
    \begin{cases}
        \psurvival f(\PFphi{s}{n}{1}|\PFphi{s}{n\!-\!1}{J}),                          & \PFr{s}{n}{1}\!=\!1  
        \\
        (1\!-\!\psurvival)f_\mathrm{d}(\PFphi{s}{n}{1}),  & \PFr{s}{n}{1}\!=0\! 
    \end{cases}\,.
\end{align}
\vspace*{-2mm}
\subsubsection{New \glspl{pf}}\label{sec:new-PFs} 
\newcommand{\muBq}{{\mu}_q^{\scriptscriptstyle \mathrm{B}}}
Following~\cite{LiaMey:Asilomar2024,LiaLeiMey:TSP2025}, the births of new \glspl{pf} in a considered \gls{roi} $\mathcal{P} \!\subset\! \realsetone{3}$ are modeled according to a Poisson point process.
That is, at each time step $n$, the number of newly appearing \glspl{pf} is Poisson-distributed.
Partition
$\mathcal{P} = \bigcup_{q=1}^{Q} \mathcal{P}_q$ 
into disjoint sets $\mathcal{P}_q$ with volumes $V_q \!\coloneqq\! \mathrm{Vol}(\mathcal{P}_q)$ in \SI{}{\cubic\metre}.
The number of newly appearing features $\rv{t}_q$ per partition $\mathcal{P}_q$ is likewise Poisson-distributed~\cite{Kingman1993Poisson} according to a Poisson \gls{pmf} $\mathbb{P}(\rv{t}_q\!=\!t_q)\!=\!\nicefrac{(\varsigma V_q)^{\scriptscriptstyle t_q}\exp(-\varsigma V_q)}{t_q!}$ with mean $\muBq\!\coloneqq\!\varsigma V_q$ and $\varsigma$ a spatial intensity in \SI{}{\per\cubic\metre}. We assume $\varsigma$ spatially uniform. For a nonuniform intensity see~\cite{LiaMey:Asilomar2024,LiaLeiMey:TSP2025}.
Let $Q$ be large enough to have approximately at most $1$ new \gls{pf} per partition $\mathcal{P}_q$, we can approximate the Poisson birth model with a Bernoulli birth model with \gls{pmf}
\begin{align}
    \mathbb{P}(\rv{t}_q = t_q) &= 
    \begin{cases}
        p_{\scriptscriptstyle \mathrm{B}}^{\scriptscriptstyle(q)}, & t_q=1 \\
        1 - p_{\scriptscriptstyle \mathrm{B}}^{\scriptscriptstyle(q)}, & t_q=0 
    \end{cases} 
\end{align}
by equating ratios of both \glspl{pmf}
\begin{align}
    \frac{\mathbb{P}(\rv{t}_q = 1)}{\mathbb{P}(\rv{t}_q = 0)} 
    &= 
    \frac{
        \muBq
        \exp(-\muBq)
    }{
        \exp(-\muBq)
    } 
    = 
    \muBq
    =
    \frac{
        p_{\scriptscriptstyle \mathrm{B}}^{\scriptscriptstyle(q)}
    }{
        1 - p_{\scriptscriptstyle \mathrm{B}}^{\scriptscriptstyle(q)}
    } \, ,
\end{align}
hence $\mathsf{t}_q \,\dot{\sim}\, \mathrm{Bernoulli}(p_{\scriptscriptstyle \mathrm{B}}^{\scriptscriptstyle(q)})$ with
$p_{\scriptscriptstyle \mathrm{B}}^{\scriptscriptstyle(q)}\!=\!
\nicefrac{\muBq}{(1+\muBq)}\!=\!
\nicefrac{\varsigma V_q}{(1+\varsigma V_q)}$.
The mean number of newly appearing \glspl{pf} in the \gls{roi} is $\mu_{\scriptscriptstyle \mathrm{B}}\!\coloneqq\!\varsigma \,V$, where $V\!\coloneqq\!\mathrm{Vol}(\mathcal{P})$.
We draw exactly one \gls{pf} per partition $\mathcal{P}_q$.
Let the set of new \glspl{pf} be $\setFeaturesNew{n}\!\coloneqq\!\{|\setFeaturesLegacy{n}|\!+\!1 \hdots |\setFeaturesLegacy{n}|\!+\!Q\}$ 
and define an appropriate mapping from new \glspl{pf} to partitions $q(s) \!\coloneqq\! s-\!|\setFeaturesLegacy{n}|$. 
The resulting \gls{pdf} for newly appearing \glspl{pf} 
$s\in\setFeaturesNew{n}$ is 
$f(\PFy{s}{n}{1})=f(\PFphi{s}{n}{1},\PFr{s}{n}{1})= f(\PFphi{s}{n}{1}|\PFr{s}{n}{1})p(\PFr{s}{n}{1})$ 
with 
$f(\PFphi{s}{n}{1}|1) = f_{\scriptscriptstyle \mathrm{B}}(\PFphi{s}{n}{1})$
and
$f(\PFphi{s}{n}{1}|0) = f_\mathrm{d}(\PFphi{s}{n}{1})$
and
$p(\PFr{s}{n}{1}) = \mathrm{Bernoulli}(p_{\scriptscriptstyle \mathrm{B}}^{\scriptscriptstyle(q)})$
from which follows
\begin{align}\label{eq:birth-PDF}
    f(\PFphi{s}{n}{1},\PFr{s}{n}{1}) = 
    \begin{cases}
        p_{\scriptscriptstyle \mathrm{B}}^{\scriptscriptstyle(q)}
        f_{\scriptscriptstyle \mathrm{B}}(\PFphi{s}{n}{1}), & \PFr{s}{n}{1}\!=1\! 
        \\
        (1-p_{\scriptscriptstyle \mathrm{B}}^{\scriptscriptstyle(q)}) f_\mathrm{d}(\PFphi{s}{n}{1}),  & \PFr{s}{n}{1}\!=0\!
    \end{cases}\,.
\end{align}
We assume the birth \gls{pdf} 
$f_{\scriptscriptstyle \mathrm{B}}(\PFphi{s}{n}{1})\!\coloneqq\!f(\psfv{s}{n}{1})f(\PFgamma{s}{n}{1})f(\PFmu{s}{n}{1})$ to factorize into a spatial birth \gls{pdf} $f(\psfv{s}{n}{1})\!\coloneqq\!\mathcal{U}(\psfv{s}{n}{1};\mathcal{P}_q)$, and hyperpriors 
$f(\PFgamma{s}{n}{1})$ and $f(\PFmu{s}{n}{1})$.

\subsubsection{\Glspl{ppr}}
The birth of a new \gls{pf} introduces $J$ new \glspl{ppr} according to a Bernoulli birth \gls{pmf}
\begin{align}\label{eq:PR-birth-PMF}
    p(\PRr{s}{n}{j})=
    \begin{cases}
        p_{\scriptscriptstyle\mathrm{B}}^{\scriptscriptstyle\mathrm{PR}}
        , & \PRr{s}{n}{j}=1 \\
        1-p_{\scriptscriptstyle\mathrm{B}}^{\scriptscriptstyle\mathrm{PR}}, & \PRr{s}{n}{j}=0
    \end{cases}\,
\end{align}
with \gls{ppr} birth probability $p_\mathrm{B}^{\mathrm{PR}}$.
Conditioned on $\PRr{s}{n-1}{j}\!=\!1$ the state transition \gls{pmf} of legacy \glspl{ppr} follows the Bernoulli \gls{pmf}
\begin{align}\label{eq:PRr-transition-existence}
    p(\PRr{s}{n}{j}|\PRr{s}{n-1}{j}\!=\!1)
    = 
    \begin{cases}
        \psurvivalPR,    & \PRr{s}{n}{j} = 1
        \\
        1-\psurvivalPR, & \PRr{s}{n}{j} = 0
    \end{cases}
\end{align}
meaning that a \gls{ppr} that existed at time $n\!-\!1$ continues to exist at time $n$ with \gls{ppr} survival probability $\psurvivalPR$.
Conditioned on $\PRr{s}{n-1}{j}\!=\!0$ the state transition \gls{pmf} of legacy \glspl{ppr} is
\begin{align}\label{eq:PRr-transition-nonexistence}
    p(\PRr{s}{n}{j}|\PRr{s}{n-1}{j}\!=\!0)
    = 
    \begin{cases}
        \precoveryPR,    & \PRr{s}{n}{j} = 1
        \\
        1-\precoveryPR , & \PRr{s}{n}{j} = 0
    \end{cases}
\end{align}
meaning that a \gls{ppr} that did not exist (e.g., due to obstruction or limited surface extent) at time $n\!-\!1$ can exist again at time $n$ with \gls{ppr} revival probability $\precoveryPR$.

\subsection{Estimation and Declaration}\label{sec:estimation-declaration}  

The goal of multipath-based \gls{slam} is to jointly estimate the \gls{mt} state $\RVstate{n}$, the map captured by \gls{pf} states $\RVPFy{s}{n}$ of all \glspl{pf} $s\in\setFeatures{n}$, and the \gls{ppr} state $\RVPRrj{n}{j}$ for all \glspl{pa} $j \in \{1\iist\dots\iist J\}$ based on the observed (thus fixed) measurements $\observationn{1:n} \!\coloneqq\! [\observationn{1}^\trp \ist \hdots \ist \observationn{n}^\trp]^\trp$. In Bayesian inference, the problem is solved by determining the marginal posterior \glspl{pdf} of the \gls{mt} state
$f(\state{n}|\observationn{1:n})$, 
of the noise variance $f(\etan{n}^{\scriptscriptstyle(j)}|\observationn{1:n})$ for $j \in \mathcal{J}$,
of all \gls{pf} states
$f(\PFphi{s}{n}{J}|\PFr{s}{n}{J}\!=\!1,\observationn{1:n})=
f(\PFphi{s}{n}{J},\PFr{s}{n}{J}\!=\!1|\observationn{1:n})
/p(\PFr{s}{n}{J}\!=\!1|\observationn{1:n})$ with
marginal posterior \gls{pf} existence probabilities
$\existenceProb{s}{n}\!\coloneqq\!p(\PFr{s}{n}{J}\!=\!1|\observationn{1:n})$ for $s \!\in\! \setFeaturest{n}$, and the marginal posterior \gls{pmf} of all \gls{ppr} states $p(\PRr{s}{n}{j}\!=\!1|\observationn{1:n})$ for $s \!\in\! \setFeaturest{n}$ and $j \in \mathcal{J}$.
The \gls{mt} state and noise variances are estimated by means of the \gls{mmse} estimator \cite{kay1993estimation}, i.e.,
\begin{align}
    \stateHat{n}^{\text{\tiny MMSE}} &= \E({\RVstate{n}|\RVobservationn{1:n}\!=\!\observationn{1:n}})\! = 
    \!\int\!\! \state{n} f(\state{n}|\observationn{1:n}) \mathrm{d}\state{n}\,, \ist
    \label{eq:stateHat}\\
    \etanHat{n}^{\,\text{\tiny MMSE}\ist \scriptscriptstyle(j)} &= 
    \E({\RVetan{n}^{\scriptscriptstyle(j)}|\RVobservationn{1:n}\!=\!\observationn{1:n}})\! = \!\int\!\! \etan{n}^{\scriptscriptstyle(j)} f(\etan{n}^{\scriptscriptstyle(j)}|\observationn{1:n}) \mathrm{d}\etan{n}^{\scriptscriptstyle(j)} \ist.
    \label{eq:etanHat} 
\end{align}
A \gls{pf} is declared to exist, i.e., detected, if $\existenceProb{s}{n}= p({\PFrBar{s}{n}\!=\!1|\observationn{1:n}})  \!>\! T_{\text{\tiny dec}}$, where $T_{\text{\tiny dec}}$ is a declaration threshold. 
Similarly, the visibility of a \gls{ppr} between \gls{pa} $j$ and \gls{pf} $s$ can be confirmed by $p(\PRr{s}{n}{j}\!=\!1 \mid \observationn{1:n}) \!>\! T_{\text{\tiny dec}}$.
The \gls{mmse} estimator for the states of existing \glspl{pf} is
\begin{align}
    \PFphiHat{s}{n}{J}^{\text{\tiny MMSE}}&= 
    \E({\RVPFphi{s}{n}| \RVPFr{s}{n}{J}\!=\!1,\RVobservationn{1:n}\!=\!\observationn{1:n}})\! \nonumber \\
    &= \!\int\!\! \PFphi{s}{n}{J} f(\PFphi{s}{n}{J}|\PFr{s}{n}{J}\!=\!1,\observationn{1:n}) \mathrm{d}\PFphi{s}{n}{J}\ist.\label{eq:PFphiHat}
\end{align}
Note that the use of both \gls{pa}-local \gls{ppr} existence variables $\PRr{s}{n}{j}$ and infrastructure-global \gls{pf} existence variables $\PFr{s}{n}{}$ supports partial obstruction, a propagation phenomenon common in \gls{dmimo}. 
We assume that new \glspl{pf} appear according to a Poisson point process.
To avoid unbounded growth of the set cardinality $|\setFeatures{n}|$ of \glspl{pf}, any features $s$ are removed from the \gls{pf} set $\setFeatures{n}$ for which $\existenceProb{s}{n}\!<\!T_{\text{\tiny pru}}$ 
or $\existenceProbPRs{s}{n}\!<\!T_{\text{\tiny pru}}$.
Here, $T_{\text{\tiny pru}}$ denotes a pruning threshold and $\existenceProbPRs{s}{n}\!\coloneqq\!p(\PRrInfrastructure{s}{n}\!=\!1|\observationn{1:n})\!\approx\!1\!-\!\prod_{\scriptscriptstyle j\in\setAnchors} (1\!-\!\existenceProbPA{s}{n}{j})$
denotes the marginal posterior existence probability of the infrastructure-level \gls{ppr} existence $\RVPRrInfrastructure{s}{n}\!\coloneqq\!\bigvee_{\scriptscriptstyle j=1}^{\scriptscriptstyle J} \RVPRr{s}{n}{j}$.
That is, the set of legacy \glspl{pf} $\setFeaturesLegacy{n}$ at time $n$ is computed as $\setFeaturesLegacy{n}\!=\!\setFeatures{n-1}\setminus \{s \in \setFeatures{n-1}  | \existenceProb{s}{n-1}\!<\!T_{\text{\tiny pru}} \vee \existenceProbPRs{s}{n-1}\!<\!T_{\text{\tiny pru}}\}$.
Let $\setFeaturesNew{n}$ denote the set of new \glspl{pf} appearing at time $n$,
the complete set of \glspl{pf} at time step $n$ is $\setFeatures{n} = \setFeaturesLegacy{n} \cup \setFeaturesNew{n}$.

Following the statistical model and assumptions in Sections~\ref{sec:measurementmodel} and~\ref{sec:state-transition-model}, the joint posterior \gls{pdf} of $\RVstate{0:n}$, $\RVPFyn{0:n}$, $\RVetann{0:n}$, and $\RVPRrBar{0:n}$ conditioned on the measurements $\observationn{1:n}$ can be factorized as shown in~\eqref{eq:joint-post-PDF}\addtocounter{equation}{1}. 
A single time step of the corresponding factor graph~\cite{Kschischang01factorGraphs,Loeliger04IntroFG} is depicted in Fig.\,\ref{fig:factor-graph}. 
This factorization enables the development of an efficient method for computing approximate marginal posterior \glspl{pdf}/\glspl{pmf}, referred to as beliefs, i.e., $\belief(\state{n}) \approx f(\state{n}|\observationn{1:n})$, $\belief(\etan{n}^{\scriptscriptstyle(j)}) \approx f(\etan{n}^{\scriptscriptstyle(j)}|\observationn{1:n})$, $\belief(\PFy{s}{n}{J}) \approx f(\PFy{s}{n}{J}|\observationn{1:n})$, and $\beliefp(\PRr{s}{n}{j}) \approx p(\PRr{s}{n}{j}|\observationn{1:n})$, as described in the following.

\begin{figure*}[t!]
    \centering%
    \includegraphics[width = \linewidth]{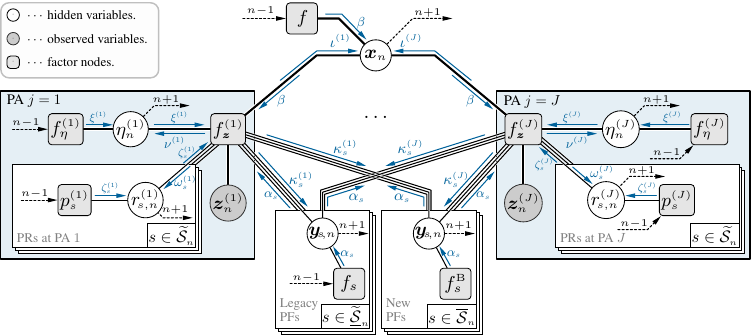}%
    \vspace{-0.1cm}\caption{Factor graph representing the joint posterior \gls{pdf} from \eqref{eq:joint-post-PDF}. 
    Light blue boxes indicate $J$ PAs. 
    The following shorthand notations are used.
    State-transition and birth \glspl{pdf}/\glspl{pmf}:     
    $f_\eta^{\scriptscriptstyle(j)}\coloneqq f(\etan{n}^{\scriptscriptstyle(j)}|\etan{n-1}^{\scriptscriptstyle(j)})$,               
    $f\coloneqq f(\state{n}|\state{n-1})$,                  
    $f_{s}\coloneqq f(\PFy{s}{n}{j}|\PFy{s}{n-1}{j})$  
    and
    $p_{s}^{\scriptscriptstyle (j)}\coloneqq p(\PRr{s}{n}{j}|\PRr{s}{n-1}{j})$
    for $s\in\setFeaturesLegacyt{n}$, 
    $f_{s}^{\mathrm{B}}\coloneqq f(\PFy{s}{n}{j})$ 
    and 
    $p_{s}^{\scriptscriptstyle (j)}\coloneqq p(\PRr{s}{n}{j})$
    for $s\in\setFeaturesNew{n}$.  
    Likelihood: 
    $f_{\bm{z}}^{\scriptscriptstyle(j)}\coloneqq
    f(\observation{n}{j}|\state{n},\PFynj{n}{j}, \PRrj{n}{j}\!\!,\etan{n}^{\scriptscriptstyle(j)})$.
    Prediction messages:   
    $\Mbeta \coloneqq  \beta(\state{n})$,
    $\Malpha{s}\coloneqq \alpha(\PFy{s}{n}{j})$,
    $\Mxi^{\scriptscriptstyle(j)} \coloneqq  \Mxi(\etan{n}^{\scriptscriptstyle(j)})$,
    $\Mzeta{s}{j} \coloneqq  \zeta(\PRr{s}{n}{j})$.
    Update messages:   
    $\Miota{j}\coloneqq \iota(\state{n};\observation{n}{j})$,
    $\Mkappa{s}{j}\coloneqq \kappa(\PFy{s}{n}{j};\observation{n}{j})$,
    $\Mnu{j}\coloneqq  \nu(\etan{n}^{\scriptscriptstyle(j)};\observation{n}{j})$,
    $\Momega{s}{j} \coloneqq  \omega(\PRr{s}{n}{j};\observation{n}{j})$.
    }%
    \label{fig:factor-graph}%
    \vspace{-5mm}%
\end{figure*}

\begin{figure*}[!t]
\normalsize
\setcounter{MYtempeqncnt}{\value{equation}}
\setcounter{equation}{22}       
\begin{align}\label{eq:joint-post-PDF}
    f(\state{0:n},\PFyn{0:n},\etann{0:n},\PRrBar{0:n}|&\observationn{1:n})
    \propto 
    \underbrace{
        f(\state{0})
        \Big(
            \prod_{j \in \setAnchors}
            f(\etan{0}^{\scriptscriptstyle(j)})
        \Big)
        \bigg(\prod_{s=0}^{\Nfeatures{0}}f(\PFy{s}{0}{J})
                \prod_{j \in \setAnchors}
                p(\PRr{s}{0}{j})
        \bigg)
    }_{\text{Initial prior PDFs/PMFs}} 
    \Bigg(
    \underbrace{
    \prod_{n'=1}^{n}
    \bigg(
        \prod_{s\in\setFeaturesLegacyt{n'}}\!\!
        f(\PFy{s}{n'}{1}|\PFy{s}{n'\!-1}{J})
        \prod_{j \in \setAnchors}
        p(\PRr{s}{n'}{j}|\PRr{s}{n'\!-1}{j})
    \bigg)
    }_{\text{Legacy PF and PR transition PDFs/PMFs}}   
    \nonumber \\[-1mm]
    &\times
     \underbrace{
        \Big(\prod_{j \in \setAnchors}f(\etan{n'}^{\scriptscriptstyle(j)}|\etan{n'\!-1}^{\scriptscriptstyle(j)})\Big)
        f(\state{n'}|\state{n'\!-1})
    }_{\text{Noise variance and \gls{mt} state transition PDFs}}
    \underbrace{
     \bigg(
        \prod_{s\in\setFeaturesNew{n'}}\!\!
        f(\PFy{s}{n'}{1}
        )
            \prod_{j \in \setAnchors}
            p(\PRr{s}{n'}{j})
    \bigg)
    }_{\text{New PF and PR PDFs/PMFs}}    
    \underbrace{
        \prod_{j \in \setAnchors}
        f(\observation{n'}{j}|\state{n'}\!,\PFynj{n'}{j}\!, \PRrj{n'}{j}\!,\etan{n'}^{\scriptscriptstyle(j)})
    }_{\text{Likelihood functions}}
    \Bigg)
\end{align}
\setcounter{equation}{\value{MYtempeqncnt}}
\hrulefill
\vspace{-5mm}%
\end{figure*}

\section{The Proposed BP Method}\label{sec:BP-method}  

In this section, we derive the \gls{bp} message passing algorithm for the proposed direct \gls{mpslam} method using the \gls{spa} rules \cite{Kschischang01factorGraphs,Loeliger04IntroFG}. \Gls{bp} message passing is an efficient approach for solving high-dimensional Bayesian inference problems. 
It performs local operations, referred to as ``messages'', along the edges of the factor graph, which represents the statistical model of the Bayesian estimation problem. 
If the factor graph is a tree, the resulting \gls{bp} solutions, called ``beliefs'', are equal to the true marginal posterior \glspl{pdf} required for computing optimum estimates. 
However, the factor graph in Fig.\,\ref{fig:factor-graph} has loops, i.e., cycles.
Iteratively applying the \gls{spa} rules on a factor graph with loops results in loopy \gls{bp}~\cite{Frey97loopyBP}, in which beliefs approximate marginal posterior \glspl{pdf}. 
On loopy factor graphs, \gls{bp} does not prescribe a unique message-passing schedule~\cite[Sec.\,V-A]{Kschischang01factorGraphs}.
Inspired by \cite{LiaLeiMey:TSP2025,Leitinger23mvaSLAM}, we specify the following message schedule \cite{Kschischang01factorGraphs}:
Messages are sent only forward in time from $n\!-\!1$ to $n$, i.e., filtering rather than smoothing, and every message is computed only once.
Within this schedule, we define the following message-passing \textit{phases}:
\begin{enumerate}[label=(\roman*)]
    \item State-transition factor-node-to-variable-node prediction messages ($\Mbeta$, $\Mxi^{\scriptscriptstyle(j)}$, $\Malpha{s}$, $\Mzeta{s}{j}$) from time $n-1$
    \item Variable-node-to-likelihood-factor-node prior messages ($\Mbeta$, $\Mxi^{\scriptscriptstyle(j)}$, $\Malpha{s}$, $\Mzeta{s}{j}$)
    \item Likelihood factor-node-to-variable-node update messages ($\Miota{j}$, $\Mnu{j}$, $\Mkappa{s}{j}$, $\Momega{s}{j}$)
\end{enumerate}
The phases are executed in \textit{series} according to the phase order specified above, i.e., each node waits to send its messages until it has received the required messages from the previous phase. 
Within each phase, we choose \textit{flooding}, i.e., each node that has received a message on one edge from the previous phase transmits pending messages on all other edges in parallel.%
\ifthenelse{\equal{\IEEEversion}{true}}%
{%
}%
{%
\footnote{The proposed parallel processing is adopted to reduce runtime by exploiting parallel computing. 
However, when message computations must be distributed across different \glspl{pa}, the proposed \gls{bp} method can be naturally adapted to a sequential processing scheme, consistent with \cite{Leitinger23mvaSLAM,Li25adaptiveDMIMOslam}.}
}

Once all messages at time step $n$ have been computed, marginal posterior \glspl{pdf}/\glspl{pmf} are approximated by beliefs. 

\subsection{Belief Calculation}\label{sec:belief-calculation}  
Beliefs about variables are computed from the product of all received messages at their respective variable nodes:
\begin{align}
    \belief(\state{n}) 
    &\propto \beta(\state{n}) 
    \prod_{\scriptscriptstyle j\in \setAnchors} 
    \iota(\state{n};\observation{n}{j}) 
    \label{eq:belief-state}
    \\
    \belief(\PFy{s}{n}{j}) 
    &\propto 
    \alpha(\PFy{s}{n}{j}) 
    \prod_{\scriptscriptstyle j\in \setAnchors} 
    \kappa(\PFy{s}{n}{j};\observation{n}{j}) 
    \label{eq:belief-PFy}
    \\
    \belief(\etan{n}^{\scriptscriptstyle(j)}) &\propto \xi(\etan{n}^{\scriptscriptstyle(j)}) 
    \nu(\etan{n}^{\scriptscriptstyle(j)};\observation{n}{j}) 
    \label{eq:belief-eta}
    \\
    \beliefp(\PRr{s}{n}{j}) 
    &\propto 
    \zeta(\PRr{s}{n}{j})
    \omega(\PRr{s}{n}{j};\observation{n}{j})
    \label{eq:belief-PR}
\end{align}
The beliefs above are nonnegative, and we require them to be normalized to integrate and/or sum to one over their support, making them \glspl{pdf}/\glspl{pmf}.
As a consequence, our prediction messages will likewise be \glspl{pdf}.
Approximating marginal posterior \glspl{pdf}, these beliefs can be used for state estimation and declaration according to Section~\ref{sec:estimation-declaration}.
For their computation, the received messages need to be derived by means of the \gls{spa} rules.
First, we derive the prediction and birth messages in Section~\ref{sec:prediction-messages}.
Then, we derive the update messages in Section~\ref{sec:update-messages}.

\subsection{Prediction and Birth Messages}\label{sec:prediction-messages}  

In sequential Bayesian filtering, the \textit{prediction} step is described by the Chapman-Kolmogorov equation.
In \gls{bp} message passing, prediction messages are described by factor-node-to-variable-node messages sent from state-transition (prior) factors to respective state variables~\cite[eq.\,(6)]{Kschischang01factorGraphs}.
Assuming first-order Markovity, i.e., $\RVstate{n}\!\condIndep\!\RVstate{n\!-\!2}|\RVstate{n\!-\!1}$, and the conditional independence $\RVstate{n}\!\condIndep\!\RVobservationn{1:n\!-\!1}|\RVstate{n\!-\!1}$ of the state from past measurements, both of which the factor graph in Fig.\,\ref{fig:factor-graph} makes explicit, the Chapman-Kolmogorov equation becomes
\begin{align}
    \!\!f(\state{n}|\observationn{1:n-1}) 
    &= \int \! f(\state{n}|\state{0:n-1},\observationn{1:n-1}) f(\state{n-1}|\observationn{1:n-1}) \mathrm{d}\state{n-1} 
    \nonumber\\[-2pt]
    &=
    \int \! f(\state{n}|\state{n-1}) f(\state{n-1}|\observationn{1:n-1}) \mathrm{d}\state{n-1}  \,.
\end{align}
Under the same assumptions for the noise variances $\RVetan{n}^{\scriptscriptstyle(j)}$ and the \gls{pf} states $\RVPFy{s}{n}$, as likewise implied by the factor graph in Fig.\,\ref{fig:factor-graph}, prediction messages depend only on the state-transition \glspl{pdf} and the ``old'' marginal posterior \glspl{pdf}, i.e., 
$f(\etan{n}^{\scriptscriptstyle(j)}|\observationn{1:n-1}) \!=\! \int \! f(\etan{n}^{\scriptscriptstyle(j)}|\etan{n-1}^{\scriptscriptstyle(j)}) f(\etan{n-1}^{\scriptscriptstyle(j)}|\observationn{1:n-1}) \mathrm{d}\etan{n-1}^{\scriptscriptstyle(j)}$,
$f(\PFy{s}{n}{j}|\observationn{1:n-1}) \!=\! 
\sum_{\scriptscriptstyle \PFr{s}{n-1}{j}
}
\int \! 
f(\PFy{s}{n}{1}|\PFy{s}{n-1}{j}) f(\PFy{s}{n-1}{j}|\observationn{1:n\!-1}) \mathrm{d}\PFphi{s}{n-1}{j}$
and
$p(\PRr{s}{n}{j}|\observationn{1:n-1}) \!\!= \!\!\sum_{\scriptscriptstyle \PRr{s}{n-1}{j}
} \! p(\PRr{s}{n}{j}|\PRr{s}{n-1}{j}) p(\PRr{s}{n-1}{j}|\observationn{1:n-1})$.
Due to the loops in the factor graph in Fig.\,\ref{fig:factor-graph}, \gls{bp} message passing merely approximates marginal posterior \glspl{pdf}, e.g., $\belief(\state{n-1}) \!\approx\! f(\state{n-1}|\observationn{1:n-1})$.
Hence, in our approximate \gls{bp} method, prediction messages approximate prediction \glspl{pdf}, e.g., $\Mbeta(\state{n}) \!\approx\! f(\state{n}|\observationn{1:n-1})$. 

\subsubsection{\gls{mt} State and Noise Variance} 
Prediction messages are
\begin{align}\label{eq:Mbeta}
    \Mbeta(\state{n}) &= \int f(\state{n}|\state{n-1}) \belief(\state{n-1}) \mathrm{d}\state{n-1}
\\
    \Mxi(\etan{n}^{\scriptscriptstyle(j)}) &= \int f(\etan{n}^{\scriptscriptstyle(j)}|\etan{n-1}^{\scriptscriptstyle(j)}) \belief(\etan{n-1}^{\scriptscriptstyle(j)}) \mathrm{d}\etan{n-1}^{\scriptscriptstyle(j)}
    \label{eq:Mxi}
\end{align}
for the \gls{mt} state and for the noise variance, respectively.

\subsubsection{Legacy \glspl{pf} $s \in \setFeaturesLegacy{n}\!\cup \{0\} \eqqcolon \setFeaturesLegacyt{n}$} 
The \gls{pf} prediction messages 
from time step $n\!-\!1$ to $n$
\begin{align}
    \alpha(\PFy{s}{n}{1}) &= 
    \!\!\sum_{\scriptscriptstyle \PFr{s}{n-1}{J}\!\in\binsetone{}}\!\int\!\! 
    f(\PFy{s}{n}{1}|\PFy{s}{n-1}{J}) \belief(\PFy{s}{n-1}{J}) 
    \mathrm{d}\PFphi{s}{n-1}{J} \,.
\end{align}
Inserting \eqref{eq:nonexistence-legacy-time} and \eqref{eq:existence-legacy-time}, for $\PFr{s}{n}{1}\!=\!1$ we obtain
\begin{align}\label{eq:alphaLegacy}
    \alpha(\PFphi{s}{n}{1},1) &= 
    \psurvival 
    \int\!\! 
    f(\PFphi{s}{n}{1}|\PFphi{s}{n-1}{J})
    \belief(\PFphi{s}{n-1}{J},1)
    \mathrm{d}\PFphi{s}{n-1}{J}
\end{align}
and for $\PFr{s}{n}{1}\!=\!0$ we obtain
\begin{align}
    \alpha(\PFphi{s}{n}{1},0) &= 
    (1\!-\!\psurvival)
    \int\!\! 
    f_\mathrm{d}(\PFphi{s}{n}{1}) 
    \belief(\PFphi{s}{n-1}{J},1)
    \mathrm{d}\PFphi{s}{n-1}{J} 
    \nonumber
    \\ 
    &~+
    \int\!\! 
    f_\mathrm{d}(\PFphi{s}{n}{1})
    \belief(\PFphi{s}{n-1}{J},0)
    \mathrm{d}\PFphi{s}{n-1}{J}
     \,.
\end{align}

\subsubsection{New \glspl{pf} $s \in \setFeaturesNew{n}$} 
The \gls{pf} birth messages at time $n$ are readily given by \eqref{eq:birth-PDF}, i.e., 
$
    \alpha(\PFphi{s}{n}{1},1) 
    = 
    p_{\scriptscriptstyle \mathrm{B}}^{\scriptscriptstyle(q)}
    f_{\scriptscriptstyle \mathrm{B}}(\PFphi{s}{n}{1})
$
and
$
    \alpha(\PFphi{s}{n}{1},0) 
    = 
    (1-p_{\scriptscriptstyle \mathrm{B}}^{\scriptscriptstyle(q)}) f_\mathrm{d}(\PFphi{s}{n}{1}) 
$
.

\subsubsection{\Glspl{ppr} of legacy \glspl{pf}} 
The prediction messages for \glspl{ppr} are
\begin{align}
    \zeta(\PRr{s}{n}{j}) &= 
    \!\!\sum\nolimits_{\scriptstyle \PRr{s}{n-1}{j}\!\in\binsetone{}}
    p(\PRr{s}{n}{j}|\PRr{s}{n-1}{j}) \beliefp(\PRr{s}{n-1}{j})
     \,.
\end{align}
Inserting \eqref{eq:PRr-transition-existence} and \eqref{eq:PRr-transition-nonexistence}, for $\PRr{s}{n}{j}\!=\!1$ we obtain
\begin{align}\label{eq:zeta}
    \zeta(\PRr{s}{n}{j}\!=\!1) 
    &= 
    \psurvivalPR 
    \beliefp(\PRr{s}{n-1}{j}\!=\!1)
    \!+\!
    \precoveryPR 
    \big(1\!-\!\beliefp(\PRr{s}{n-1}{j}\!=\!1)\big)
\end{align}
and for $\PRr{s}{n}{j}\!=\!0$ we obtain
\begin{align}
    \Mzeta{s}{j}\!(0) 
    &= 
    (1\!-\!\psurvivalPR)
    \beliefp(\PRr{s}{n-1}{j}\!=\!1)
    \!+\!
    (1\!-\!\precoveryPR)
    \big(1\!-\!\beliefp(\PRr{s}{n-1}{j}\!=\!1)\big)\,.
    \nonumber
\end{align}
\subsubsection{\Glspl{ppr} of new \glspl{pf}} 
For new \glspl{pf} $s\!\in\!\setFeaturesNew{n}$, \gls{ppr} prediction messages 
$\zeta(1)\!=\!p_{\scriptscriptstyle\mathrm{B}}^{\scriptscriptstyle\mathrm{PR}}$
and
$\zeta(0)\!=\!1\!-\!p_{\scriptscriptstyle\mathrm{B}}^{\scriptscriptstyle\mathrm{PR}}$
follow from~\eqref{eq:PR-birth-PMF}.

\subsection{Update Messages and Moment-Matched 
Approximations}\label{sec:update-messages}  
After observations $\observationn{n}$ are made, beliefs are updated through update messages.
In sequential Bayesian filtering, the \textit{update} step is described by the Bayesian update equation.\footnote{Strictly speaking, the Bayesian update equation directly yields the posterior \gls{pdf} as the normalized product of prediction \gls{pdf} and observation likelihood.}
In \gls{bp} message passing, update messages are described by factor-node-to-variable-node messages sent from likelihood factors to respective state variables. 
Messages sent from a factor node to a neighboring variable node are computed by multiplying the factor with all incoming messages from the \textit{other} neighboring variable nodes and marginalizing over those variables~\cite{Kschischang01factorGraphs}. 

Update messages sent from likelihood factor nodes 
$f(\observation{n}{j}|\state{n},\PFynj{n}{j}, \PRrj{n}{j}\!\!,\etan{n}^{\scriptscriptstyle(j)})$
to the \gls{mt} state variable node $\state{n}$ are
\begin{align}\label{eq:Miota-exact}
    \iota\big(\state{n};\observation{n}{j}\big) 
    &\!= \!\!\!\!\!\!
    \sum
    _{\scalebox{0.7}{
        $\{\PFr{s}{n}{j}\}_{\setFeaturest{n}}\!\!\in\!\!\binsetone{\Nfeaturest{n}}$
    }}
    \!\!\sum
    _{\scalebox{0.7}{
        $\{\PRr{s}{n}{j}\}_{\setFeaturest{n}}\!\!\in\!\!\binsetone{\Nfeaturest{n}}$
    }}
    \!\! \int\!\!\hdots\!\!\int \!
    f(\observation{n}{j}|\state{n},\PFynj{n}{j}, \PRrj{n}{j}\!\!,\etan{n}^{\scriptscriptstyle(j)})
    \nonumber
    \\
    \times 
    \Mxi&(\etan{n}^{\scriptscriptstyle(j)})
    \prod_{s\in\setFeaturest{n}}
    \alpha(\PFphi{s}{n}{j},\PFr{s}{n}{j}) 
    \zeta(\PRr{s}{n}{j})
    \mathrm{d}\PFphi{s}{n}{j}
    \mathrm{d}\etan{n}^{\scriptscriptstyle(j)}
    \,.
\end{align}
Update messages sent 
to the noise variance 
$\etan{n}^{\scriptscriptstyle(j)}$ are
\begin{align}
    \nu\big(\etan{n}^{\scriptscriptstyle(j)};\observation{n}{j}\big) 
    \!= \!\!\!\!\!\!\!\!
    \sum
    _{\scalebox{0.7}{
        $\{\PFr{s}{n}{j}\}_{\setFeaturest{n}}\!\!\in\!\!\binsetone{\Nfeaturest{n}}$
    }}
    \!\!\!\sum
    _{\scalebox{0.7}{
        $\{\PRr{s}{n}{j}\}_{\setFeaturest{n}}\!\!\in\!\!\binsetone{\Nfeaturest{n}}$
    }}
    \!\! \int\!\!\hdots\!\!\int \!\!\!
    f(\observation{n}{j}|\state{n},\PFynj{n}{j}, \PRrj{n}{j}\!\!,\etan{n}^{\scriptscriptstyle(j)})
    \nonumber
    \\
     \times 
    \beta(\state{n})\!
    \prod\nolimits_{s\in\setFeaturest{n}}
    \!\!
    \alpha(\PFphi{s}{n}{j},\PFr{s}{n}{j}) 
    \zeta(\PRr{s}{n}{j})
    \mathrm{d}\PFphi{s}{n}{j}
    \mathrm{d}\state{n}  \,.
\end{align}
Update messages sent 
to the \gls{pf} state variable node $\PFy{s}{n}{j}$ are
\begin{align*}
    \kappa\big(\PFy{s}{n}{j};\!\observation{n}{j}\!\big) 
    \!= \!\!\!\!\!\!\!\!\!
    \sum
    _{\scalebox{0.7}{
        $\{\PFr{s'}{n}{j}\}_{\setFeaturest{n}\setminus \{s\}}\!\!\in\!\!\binsetone{\Nfeatures{n}}$
    }}
    \!\!\!\!\sum
    _{\scalebox{0.7}{
        $\{\PRr{s}{n}{j}\}_{\setFeaturest{n}}\!\!\in\!\!\binsetone{\Nfeaturest{n}}$
    }}
    \!\!
    \!\! \int\!\!\hdots\!\!\int \!\!\!
    f(\observation{n}{j}|\state{n},\PFynj{n}{j}, \PRrj{n}{j}\!\!,\etan{n}^{\scriptscriptstyle(j)})
    \nn
    \\
    \times \Mbeta(\state{n})
    \Mxi(\etan{n}^{\scriptscriptstyle(j)})
    \!\!\!
    \prod_{\sA\in\setFeaturest{n}\setminus \{s\}}
    \!\!\!\!
    \alpha(\PFphi{\sA\!}{n}{j},\PFr{\sA\!}{n}{j}) 
    \mathrm{d}\PFphi{\sA\!}{n}{j}
    \!\!\prod_{\sB\in\setFeaturest{n}}\!\!
    \zeta(\PRr{\sB}{n}{j})
    \mathrm{d}\state{n}
    \mathrm{d}\etan{n}^{\scriptscriptstyle(j)}  .
\end{align*}
Update messages sent to the \gls{ppr} state variable node $\PRr{s}{n}{j}$ are
\begin{align}\label{eq:Momega-exact}
    &\!\!\!\omega\big(\PRr{s}{n}{j};\observation{n}{j}\big) 
    \!= \!\!\!\!\!\!\!\!\!\!\!
    \sum
    _{\scalebox{0.7}{
        $\{\PFr{s}{n}{j}\}_{\setFeaturest{n}}\!\!\in\!\!\binsetone{\Nfeaturest{n}}$
    }}
    \!\!\!\!\sum
    _{\scalebox{0.7}{
        $\{\PRr{s'}{n}{j}\}_{\setFeaturest{n}\setminus \{s\}}\!\!\in\!\!\binsetone{\Nfeatures{n}}$
    }}
    \!\!
    \!\! \int\!\!\!\hdots\!\!\!\int \!
    \!f(\observation{n}{j}|\state{n},\PFynj{n}{j}, \PRrj{n}{j}\!\!,\etan{n}^{\scriptscriptstyle(j)})
    \nn
    \\
    &\times\!\!\Mbeta(\state{n})
    \Mxi(\etan{n}^{\scriptscriptstyle(j)}\!)
    \!\!\!
    \prod_{\sA\in\setFeaturest{n}}
    \!\!
    \alpha(\PFphi{\sA\!}{n}{j},\!\PFr{\sA\!}{n}{j}) 
    \mathrm{d}\PFphi{\sA\!}{n}{j}
    \!\!
    \!\!\!\!
    \prod_{\sB\in\setFeaturest{n}\setminus \{s\}}\!\!\!\!\!\!
    \zeta(\PRr{\sB}{n}{j})
    \mathrm{d}\state{n}
    \mathrm{d}\etan{n}^{\scriptscriptstyle(j)}\!.
\end{align}
Since the likelihood factor \eqref{eq:likelihood-Type-II} is complex Gaussian, \eqref{eq:Miota-exact}--\eqref{eq:Momega-exact} involve marginalizations over a complex Gaussian mixture for all $2^{\Nfeatures{n}+1}$ or $2^{\Nfeatures{n}}$ 
variations
of \gls{pf} existences $\PFr{s}{n}{j}$ and \gls{ppr} existences $\PRr{s}{n}{j}$.
Hence, this exact message computation according to the \gls{spa} rules does \textit{not} lead to a scalable solution in the number of \glspl{pf} $\Nfeatures{n}$.

To arrive at a scalable solution, the authors of~\cite{Mingchao23TBT,LiaMey:Asilomar2024,LiaLeiMey:TSP2025}
were inspired by the \textit{moment-matching} approach of~\cite{Davies22MomentMatching}.
We follow their approach and approximate the complex Gaussian mixtures in \eqref{eq:Miota-exact}--\eqref{eq:Momega-exact} with unimodal, nonzero-mean complex Gaussian messages with mean and covariance matrix moment-matched to the complex Gaussian mixtures.
That is, we formulate the approximate messages 
$\widetilde{\iota}(\state{n};\observation{n}{j}) \!\coloneqq\! \mathcal{CN}\big(\observation{n}{j};\muiota,\Ciota\big)$,~
$\widetilde{\nu}(\etan{n}^{\scriptscriptstyle(j)};\observation{n}{j})\!\coloneqq\! \mathcal{CN}\big(\observation{n}{j};\munu,\Cnu\big)$, 
$\widetilde{\kappa}(\PFy{s}{n}{j};\observation{n}{j})\!\coloneqq\! \mathcal{CN}\big(\observation{n}{j};\mukappa,\Ckappa\big)$, and $\widetilde{\omega}(\PRr{s}{n}{j};\observation{n}{j})\!\coloneqq\! \mathcal{CN}\big(\observation{n}{j};\muomega,\Comega\big)$.

\begin{proposition}
The first raw moments of each message are
\begin{align}
    \muiota(\state{n}) 
    &\coloneqq 
    \E_{\iota} \!
    \left(
        \RVobservation{n}{j}|\RVstate{n}
    \right)
    =\sum\nolimits_{s \in \setFeaturest{n}} \musnj{1}(\state{n})
    \label{eq:muiota}
    \\
    \munu
    &\coloneqq 
    \E_{\nu} \!
    \left(
        \RVobservation{n}{j}|\RVetan{n}^{\scriptscriptstyle(j)}
    \right)
    =\sum\nolimits_{s \in \setFeaturest{n}} \musnj{3}
    \\
    \mukappa(\PFy{s}{n}{j})
    &\coloneqq 
    \E_{\kappa} \!
    \left(
        \RVobservation{n}{j}|\RVPFy{s}{n}
    \right)
    = \PFr{s}{n}{j} \musnj{2}(\PFphi{s}{n}{j}) \nonumber \\
    & \hspace{28mm} + \sum\nolimits_{s' \in \setFeaturest{n} \setminus \{s\}} \musnjp{3} 
    \\
    \muomega(\PRr{s}{n}{j})
    &\coloneqq 
    \E_{\omega} \!
    \left(
        \RVobservation{n}{j}|\RVPRr{s}{n}{j}
    \right)
    = \PRr{s}{n}{j} \musnj{4} \nonumber \\
    & \hspace{28mm} + \sum\nolimits_{s' \in \setFeaturest{n} \setminus \{s\}} \musnjp{3} 
\end{align}
and the second central moments of each message are
\begin{align}
    &\Ciota(\state{n}) 
    \!\coloneqq\!
    \E_{\iota} 
    \!
    \Big(
        \RVobservation{n}{j}{\RVobservation{n}{j}}^\herm|\RVstate{n}
    \Big)
    \!-\!
    \E_{\iota} \!
    \left(
        \RVobservation{n}{j}|\RVstate{n}
    \right) 
    \E_{\iota} \!
    \left(
        \RVobservation{n}{j}|\RVstate{n}
    \right) ^\herm
    \label{eq:Ciota}
    \\
    &=
    \eye{\Nz} \etaXi{n} \!+\!
    \sum\nolimits_{s \in \setFeaturest{n}}\!\! \Csnj{1}(\state{n}) 
    \!+\!
    \Kiota\!(\state{n}) \!-\! \muiota{\muiota}^\herm\!(\state{n}) 
    \nn \\ 
    &\Cnu(\etan{n}^{\scriptscriptstyle(j)}) 
    \!\coloneqq \!
    \E_{\nu} 
    \!
    \Big(
        \!\RVobservation{n}{j}{\RVobservation{n}{j}}^\herm|\RVetan{n}^{\scriptscriptstyle(j)}\!
    \Big)
    \!-\!
    \E_{\nu} \!
    \left(
        \RVobservation{n}{j}|\RVetan{n}^{\scriptscriptstyle(j)}\!
    \right) 
    \!\E_{\nu} \!
    \left(
        \RVobservation{n}{j}|\RVetan{n}^{\scriptscriptstyle(j)}
    \right) ^\herm
    \nn \\
    &=
    \eye{\Nz} \etan{n}^{\scriptscriptstyle(j)} +
    \sum\nolimits_{s \in \setFeaturest{n}} \Csnj{3}
    \!+\!
    \Knu(\etan{n}^{\scriptscriptstyle(j)})  \!-\! \munu{\munu}^\herm
    \\ 
    &\Ckappa(\PFy{s}{n}{j})
    \!\coloneqq\!
    \E_{\kappa} 
    \!
    \Big(
        \! \RVobservation{n}{j}{\RVobservation{n}{j}}^\herm|\RVPFy{s}{n} \!
    \Big)
    \!-\!
    \E_{\kappa} \!
    \left(
        \RVobservation{n}{j}|\RVPFy{s}{n}
    \right) 
    \!
    \E_{\kappa} \!
    \left(
        \RVobservation{n}{j}|\RVPFy{s}{n}
    \right)^{\!\herm}
    \nn \\
    &=
    \PFr{s}{n}{j} \Csnj{2}(\PFphi{s}{n}{j}) 
    + 
    \eye{\Nz} \etaXi{n}
    +
    \sum\nolimits_{s' \in \setFeaturest{n} \setminus \{s\}} \Csnjp{3}
    \label{eq:Ckappa} \\
    &~~+\!
    \Kkappa(\PFy{s}{n}{}) \!-\! \mukappa{\mukappa}^{\!\herm}(\PFy{s}{n}{}) 
    \nn
    \\ 
    &\Comega(\PRr{s}{n}{j})
    \!\coloneqq\!
    \E_{\omega} 
    \!
    \Big(
        \! \RVobservation{n}{j}{\RVobservation{n}{j}}^\herm|\RVPRr{s}{n}{j} \!
    \Big)
    \!-\!
    \E_{\omega} \!
    \left(
        \RVobservation{n}{j}|\RVPRr{s}{n}{j}
    \right) 
    \!
    \E_{\omega} \!
    \left(
        \RVobservation{n}{j}|\RVPRr{s}{n}{j}
    \right)^{\!\herm}
    \nn \\
    &=
    \PRr{s}{n}{j} \Csnj{4}
    + 
    \eye{\Nz} \etaXi{n}
    +
    \sum\nolimits_{s' \in \setFeaturest{n} \setminus \{s\}} \Csnjp{3}
    \label{eq:Comega}\\
    &~~+\!
    \Komega(\PRr{s}{n}{j}) \!-\! \muomega{\muomega}^{\!\herm}(\PRr{s}{n}{j}) 
    \nn
\end{align}
with the noise variance prior mean $\etaXi{n}\!\coloneqq\!\int\etan{n}^{\scriptscriptstyle(j)}\Mxi(\etan{n}^{\scriptscriptstyle(j)})\mathrm{d}\etan{n}^{\scriptscriptstyle(j)}$ and abbreviated mean and second noncentral moment terms~{\normalfont\cite{Mingchao23TBT,LiaLeiMey:TSP2025}}
\begin{align}
    \musnj{1}(\state{n}) 
    &\coloneqq \int \PFmu{s}{n}{j} 
    \steerVecx{s}{n}{j}
    \alpha(\PFphi{s}{n}{j},1) \mathrm{d}\PFphi{s}{n}{j} 
    \Mzeta{s}{j}\!(1)
    \label{eq:musnj1}
    \\
    \musnj{2}(\PFphi{s}{n}{j}) &\coloneqq \PFmu{s}{n}{j} \!\!
    \int \! \steerVec{j}(\state{n},\PFphi{s}{n}{j}) \beta(\state{n}) \mathrm{d}\state{n}
    \Mzeta{s}{j}\!(1)
    \label{eq:musnj2}
    \\
    \musnj{3} \coloneqq& \iint
    \PFmu{s}{n}{j} 
    \steerVecx{s}{n}{j}
    \beta(\state{n})
    \alpha(\PFphi{s}{n}{j},1) 
    \mathrm{d}\PFphi{s}{n}{j}\mathrm{d}\state{n}
    \Mzeta{s}{j}\!(1)
    \nn
    \\
    \musnj{4} \coloneqq& 
    \iint
    \PFmu{s}{n}{j} 
    \steerVecx{s}{n}{j}
    \beta(\state{n})
    \alpha(\PFphi{s}{n}{j},1) 
    \mathrm{d}\PFphi{s}{n}{j}\mathrm{d}\state{n}
    \label{eq:musnj4}
    \\
    \Csnj{1}(\state{n}) &\coloneqq
    \int \PFgamma{s}{n}{j} \steerVecx{s}{n}{j} {\steerVecx{s}{n}{j}}^\herm \alpha(\PFphi{s}{n}{j},1) \mathrm{d}\PFphi{s}{n}{j}
    \Mzeta{s}{j}\!(1)
    \label{eq:Csnj1}
    \\
    \Csnj{2}(\PFphi{s}{n}{j}) &\coloneqq
    \PFgamma{s}{n}{j} \int  \steerVecx{s}{n}{j} {\steerVecx{s}{n}{j}}^\herm \beta(\state{n}) \mathrm{d}\state{n}
    \Mzeta{s}{j}\!(1)
    \label{eq:Csnj2}
    \\
    \Csnj{3} 
    \coloneqq
    \iint&
    \! \PFgamma{s}{n}{j} 
    \steerVecx{s}{n}{j} {\steerVecx{s}{n}{j}}^\herm
    \beta(\state{n})
    \alpha(\PFphi{s}{n}{j},1) 
    \mathrm{d}\PFphi{s}{n}{j}\mathrm{d}\state{n} 
    \Mzeta{s}{j}\!(1)
    \nn
    \\
    \Csnj{4} 
    \coloneqq
    \iint&
    \! \PFgamma{s}{n}{j} 
    \steerVecx{s}{n}{j} {\steerVecx{s}{n}{j}}^\herm
    \beta(\state{n})
    \alpha(\PFphi{s}{n}{j},1) 
    \mathrm{d}\PFphi{s}{n}{j}\mathrm{d}\state{n} 
    \label{eq:Csnj4}
    \\
    \Kiota\big(\state{n}\big) 
    \!\coloneqq \!\!\!\!\!\!
    &\sum
    _{\scalebox{0.7}{
        $\{\PFr{s}{n}{j}\}_{\setFeaturest{n}}\!\!\in\!\!\binsetone{\Nfeaturest{n}}$
    }}
    \!\!\sum
    _{\scalebox{0.7}{
        $\{\PRr{s}{n}{j}\}_{\setFeaturest{n}}\!\!\in\!\!\binsetone{\Nfeaturest{n}}$
    }}
    \!\! \int\!\!\hdots\!\!\int \!
    \bm{\mu}_{\scriptscriptstyle n}^{\scriptscriptstyle (j)}
    {\bm{\mu}_{\scriptscriptstyle n}^{\scriptscriptstyle (j)}}^\herm
    \nn 
    \\
    \times 
    \Mxi&(\etan{n}^{\scriptscriptstyle(j)})
    \prod_{s\in\setFeaturest{n}}
    \alpha(\PFphi{s}{n}{j},\PFr{s}{n}{j}) 
    \zeta(\PRr{s}{n}{j})
    \mathrm{d}\PFphi{s}{n}{j}
    \mathrm{d}\etan{n}^{\scriptscriptstyle(j)}
    \label{eq:Kiota}
    \\
    \Knu\big(\etan{n}^{\scriptscriptstyle(j)}\big) 
    \!&\coloneqq \!\!\!\!\!\!\!\!
    \sum
    _{\scalebox{0.7}{
        $\{\PFr{s}{n}{j}\}_{\setFeaturest{n}}\!\!\in\!\!\binsetone{\Nfeaturest{n}}$
    }}
    \!\!\!\sum
    _{\scalebox{0.7}{
        $\{\PRr{s}{n}{j}\}_{\setFeaturest{n}}\!\!\in\!\!\binsetone{\Nfeaturest{n}}$
    }}
    \!\! \int\!\!\hdots\!\!\int \!\!\!
    \bm{\mu}_{\scriptscriptstyle n}^{\scriptscriptstyle (j)}
    {\bm{\mu}_{\scriptscriptstyle n}^{\scriptscriptstyle (j)}}^\herm
    \nonumber
    \\
     \times 
    \beta&(\state{n})\!
    \prod\nolimits_{s\in\setFeaturest{n}}
    \!\!
    \alpha(\PFphi{s}{n}{j},\PFr{s}{n}{j}) 
    \zeta(\PRr{s}{n}{j})
    \mathrm{d}\PFphi{s}{n}{j}
    \mathrm{d}\state{n}  
\\
    \Kkappa\big(\PFy{s}{n}{j}\big) 
    \!\coloneqq& \!\!\!\!\!\!\!\!
    \sum
    _{\scalebox{0.7}{
        $\{\PFr{s'}{n}{j}\}_{\setFeaturest{n}\setminus \{s\}}\!\!\in\!\!\binsetone{\Nfeatures{n}}$
    }}
    \!\!\!\!\sum
    _{\scalebox{0.7}{
        $\{\PRr{s}{n}{j}\}_{\setFeaturest{n}}\!\!\in\!\!\binsetone{\Nfeaturest{n}}$
    }}
    \!\!
    \!\! \int\!\!\hdots\!\!\int \!
    \bm{\mu}_{\scriptscriptstyle n}^{\scriptscriptstyle (j)}
    {\bm{\mu}_{\scriptscriptstyle n}^{\scriptscriptstyle (j)}}^\herm
    \\
    \times \Mbeta(\state{n})
    \Mxi(\etan{n}^{\scriptscriptstyle(j)}&)
    \!\!
    \prod_{\sA\in\setFeaturest{n}\setminus \{s\}}
    \!\!\!
    \alpha(\PFphi{\sA\!}{n}{j},\PFr{\sA\!}{n}{j}) 
    \mathrm{d}\PFphi{\sA\!}{n}{j}
    \prod_{\sB\in\setFeaturest{n}}
    \zeta(\PRr{\sB}{n}{j})
    \mathrm{d}\state{n}
    \mathrm{d}\etan{n}^{\scriptscriptstyle(j)} 
    \nn
\\
    \!\!\!\Komega\big(\PRr{s}{n}{j}\big) 
    \!\coloneqq& \!\!\!\!\!\!\!\!
    \sum
    _{\scalebox{0.7}{
        $\{\PFr{s}{n}{j}\}_{\setFeaturest{n}}\!\!\in\!\!\binsetone{\Nfeaturest{n}}$
    }}
    \!\!\!\!\sum
    _{\scalebox{0.7}{
        $\{\PRr{s'}{n}{j}\}_{\setFeaturest{n}\setminus \{s\}}\!\!\in\!\!\binsetone{\Nfeatures{n}}$
    }}
    \!\!
    \!\! \int\!\!\hdots\!\!\int \!
    \!\bm{\mu}_{\scriptscriptstyle n}^{\scriptscriptstyle (j)}
    {\bm{\mu}_{\scriptscriptstyle n}^{\scriptscriptstyle (j)}}^\herm
    \label{eq:Komega}
    \\
    \times\Mbeta(\state{n})
    \Mxi(\etan{n}^{\scriptscriptstyle(j)}&)
    \!\!\!
    \prod_{\sA\in\setFeaturest{n}}
    \!\!
    \alpha(\PFphi{\sA\!}{n}{j},\PFr{\sA\!}{n}{j}) 
    \mathrm{d}\PFphi{\sA\!}{n}{j}
    \!\!
    \!\!\!
    \prod_{\sB\in\setFeaturest{n}\setminus \{s\}}\!\!\!\!
    \zeta(\PRr{\sB}{n}{j})
    \mathrm{d}\state{n}
    \mathrm{d}\etan{n}^{\scriptscriptstyle(j)}
    \nn
\end{align}
again using the shorthand notation $\steerVecx{s}{n}{j} = \steerVec{j}(\state{n},\PFphi{s}{n}{j})$.
\end{proposition}
\begin{proof}
See our Extended Derivations in~\cite{ThisPaperDerivations}.
\end{proof} %

At this point, we have defined all prediction messages in Section~\ref{sec:prediction-messages} and
all update messages in Section~\ref{sec:update-messages} together with their moment-matched approximations.
Using the moment-matched approximate messages
$\widetilde{\iota}(\state{n};\observation{n}{j})$,~$\widetilde{\nu}(\etan{n}^{\scriptscriptstyle(j)};\observation{n}{j})$,  
$\widetilde{\kappa}(\PFy{s}{n}{j};\observation{n}{j})$, 
and $\widetilde{\omega}(\PRr{s}{n}{j};\observation{n}{j})$, we can approximate the beliefs in \eqref{eq:belief-state}--\eqref{eq:belief-PR}, which, inserted in \eqref{eq:stateHat}--\eqref{eq:PFphiHat},
allow us to compute the state estimates desired to solve the \gls{slam} problem.
These continuous derivations serve as the blueprint for the particle-based implementation described next.

\subsection{Particle-Based Implementation}\label{sec:particle-based-implementation}  
In general, the beliefs $\belief(\state{n})$, $\belief(\PFy{s}{n}{j})$, and $\belief(\etan{n}^{\scriptscriptstyle(j)})$ will be arbitrary, non-Gaussian \glspl{pdf}, as will the prediction messages in Section~\ref{sec:prediction-messages}, even under linear Gaussian state-transition \glspl{pdf}.
Consequently, the marginalization integrals for computing \gls{bp} messages do not have a closed-form solution in general.
In the following, we introduce a numerically feasible particle-based implementation of our \gls{bp} method, where \glspl{pdf} are approximated by \glspl{pr}~\cite{LeitMeyHlaWitTufWin:TWC2019,LiaLeiMey:TSP2025}, which are sets of $P$ weighted particles.
That is, we approximate the beliefs in~\eqref{eq:belief-state} and \eqref{eq:belief-eta} using their \glspl{pr}
$\belief(\state{n}) \approxcolon 
\sum_{\scriptscriptstyle p=1}^{\scriptscriptstyle P} \weight{\bm{x},n}{p} \delta(\state{n}\!- \particle{\state{n}}{p})$
and
$\belief(\etan{n}^{\scriptscriptstyle(j)}) \approxcolon 
\sum_{\scriptscriptstyle p=1}^{\scriptscriptstyle P}
\weight{\eta,n}{j,p} \delta(\etan{n}^{\scriptscriptstyle(j)}\!- \particle{\etan{n}}{j,p})$.
To represent \glspl{pdf}, we require their weights to sum to one, i.e., 
$\sum_{\scriptscriptstyle p=1}^{\scriptscriptstyle P}\weight{\bm{x},n}{p}\!=\!1$ and 
$\sum_{\scriptscriptstyle p=1}^{\scriptscriptstyle P}\weight{\eta,n}{j,p}\!=\!1$.

For the belief $\belief(\PFy{s}{n}{j})$ in~\eqref{eq:belief-PFy}, we keep a \gls{pr} 
$\belief(\PFphi{s}{n}{j},1) \approxcolon 
\sum_{\scriptscriptstyle p=1}^{\scriptscriptstyle P}
\weight{\bm{y},s,n}{p} \delta(\PFphi{s}{n}{j}\!-\! \particle{\PFphi{s}{n}{j,p}}{p}) $ only under the alternative hypothesis $\mathcal{H}_1\!:\!\PFr{s}{n}{j}\!=\!1$,
while $\belief(\PFphi{s}{n}{j},0)$ is represented by a constant probability density~\cite[Sec.\,VI]{MeyKroWilLauHlaBraWin:IEEEProc2018}.
This is the case because we require the dummy \gls{pdf} to be constant in 
the \gls{roi} $\mathcal{P}$ and zero outside, 
i.e., $f_\mathrm{d}(\PFphi{s}{n}{j}) \!=\! \mathcal{U}(\PFphi{s}{n}{j};\mathcal{P}\!\times\!\mathcal{R}_\gamma\!\times\!\mathcal{R}_\mu)$, 
with $\mathcal{R}_\gamma\!\subset\!\realsetone{}$ and $\mathcal{R}_\mu\!\subset\!\complexsetone{}$ denoting the supports of $\PFgamma{s}{n}{j}$ and $\PFmu{s}{n}{j}$, respectively,
which implies that $f_\mathrm{d}$ carries no Fisher information about $\PFphi{s}{n}{j}$.
Consequently, the weights $\weight{\bm{y},s,n}{p}$ do not sum to $1$ but instead they approximate marginal posterior \gls{pf} existence probabilities 
\begin{align} \label{eq:existenceProb}
    \existenceProb{s}{n} \!=\! 
    p(\PFr{s}{n}{j}\!=\!1|\observationn{1:n})
    \!\approx \!
    \int \!\! \belief(\PFphi{s}{n}{j},1) \mathrm{d}\PFphi{s}{n}{j}
    \!\approx \!
    \sum_{p=1}^P \weight{\bm{y},s,n}{p} 
    .
\end{align}
The beliefs $\widetilde{p}(\PRr{s}{n}{j})$ are \glspl{pmf} of which we only store the $\widetilde{p}(\PRr{s}{n}{j}\!=\!1)$ branches, 
which approximate the marginal posterior \gls{ppr} existence probabilities $\existenceProbPA{s}{n}{j} \! \coloneqq\! p({\PRr{s}{n}{j}\!=\!1|\observationn{1:n}})\!\approx\!\beliefp(\PRr{s}{n}{j}=1) $.

\subsubsection{Prediction and Birth Messages}

By approximating marginal posterior \glspl{pdf} at time $n\!-\!1$ with the \glspl{pr} of their beliefs, and evaluating the marginalization integral
\begin{align}
    f(\state{n}|\observationn{1:n-1}) 
    \!&= \!\!
    \int \!\! f(\state{n}|\state{n-1}) f(\state{n-1}|\observationn{1:n-1}) \mathrm{d}\state{n-1}
    \nn \\[-2pt]
    &\approx\!\!
    \int \!\! f(\state{n}|\state{n-1}) 
    \!
    \sum_{\scriptscriptstyle p=1}^{\scriptscriptstyle P} \weight{\bm{x},n-1}{p} \delta(\state{n-1}\!- \particle{\state{n-1}}{p})
    \mathrm{d}\state{n-1}
    \nn \\[-2pt]
    &=
    \sum\nolimits_{\scriptscriptstyle p=1}^{\scriptscriptstyle P} 
    \weight{\bm{x},n-1}{p} f(\state{n}|\particle{\state{n-1}}{p}) \,,\label{eq:PR-prediction-PDF}
\end{align}
we obtain a prediction \gls{pdf} parameterized by particles. 
As is done in \gls{sir}~\cite{Arulampalam02PFtutorial},
for each particle $\particle{\state{n-1}}{p}$ we now draw one particle 
$\particle{\state{n}}{p}$ 
from the state-transition \gls{pdf} $f(\state{n}|\particle{\state{n-1}}{p})$ used as proposal \gls{pdf} 
to obtain a \gls{pr} for the prediction \gls{pdf}
$f(\state{n}|\observationn{1:n-1}) \!\approx\! \beta(\state{n}) \!\approx \!\sum\nolimits_{\scriptscriptstyle p=1}^{\scriptscriptstyle P}
\weight{\beta,n}{p} \delta\big( \state{n}-\particle{\state{n}}{p} \big)$
where $\weight{\beta,n}{p}\!=\!\weight{\bm{x},n-1}{p}$ if we can directly sample from $f(\state{n}|\particle{\state{n-1}}{p})$.
We likewise compute the noise variance prediction message 
$\xi(\etan{n}^{\scriptscriptstyle(j)})\approx \sum\nolimits_{\scriptscriptstyle p=1}^{\scriptscriptstyle P}
\weight{\xi,n}{j,p} \delta\big( \etan{n}^{\scriptscriptstyle(j)}-\particle{\etan{n}}{j,p} \big)$
with $\weight{\xi,n}{j,p}\!=\!\weight{\eta,n-1}{j,p}$ 
and $\particle{\etan{n}}{j,p}$ sampled from $f(\etan{n}^{\scriptscriptstyle(j)}|\particle{\etan{n-1}}{j,p})$
using~\eqref{eq:Mxi}.
The prediction message $\zeta(\PRr{s}{n}{j}\!=\!1)$ is computed from $\widetilde{p}(\PRr{s}{n-1}{j}\!=\!1)$ according to~\eqref{eq:zeta}.
\Glspl{pr} of prediction messages for legacy \glspl{pf} $s\!\in\!\setFeaturesLegacyt{n}$ 
are 
$\alpha(\PFphi{s}{n}{1},1)\approx \sum\nolimits_{\scriptscriptstyle p=1}^{\scriptscriptstyle P}
\weight{\alpha,s,n}{p} \delta(\PFphi{s}{n}{1}\!-\! \particle{\PFphi{s}{n}{1,p}}{p}) $
with $\weight{\alpha,s,n}{p}\!=\!\psurvival \weight{\bm{y},s,n-1}{p}$ 
and $\particle{\PFphi{s}{n}{1,p}}{p}$ sampled from $f(\PFphi{s}{n}{1}|\particle{\PFphi{s}{n-1}{J,p}}{p})$ using~\eqref{eq:alphaLegacy}.

As described in Section~\ref{sec:new-PFs}, we introduce one new \gls{pf} per partition $\mathcal{P}_q$ per time step $n$.
\Glspl{pr} of birth messages for these new \glspl{pf} $s \in \setFeaturesNew{n}$ are 
$\alpha(\PFphi{s}{n}{1},1)\approx \sum\nolimits_{\scriptscriptstyle p=1}^{\scriptscriptstyle P}
\weight{\alpha,s,n}{p} \delta(\PFphi{s}{n}{1}\!-\! \particle{\PFphi{s}{n}{1,p}}{p}) $
where the weights would be $\weight{\alpha,s,n}{p}\!=\!\nicefrac{p_{\scriptscriptstyle \mathrm{B}}^{\scriptscriptstyle(q)}}{P}$ if particles $\particle{\PFphi{s}{n}{1,p}}{p}$ were sampled from $f_{\scriptscriptstyle \mathrm{B}}(\PFphi{s}{n}{1})$ per~\eqref{eq:birth-PDF}.
However, we choose to sample from a different proposal \gls{pdf} $f_{\scriptscriptstyle \mathrm{B}}^{\scriptscriptstyle \mathrm{p}}(\PFphi{s}{n}{1}|\observationn{n})$ which improves early detection of map features and introduces importance weights 
$\weightt{\alpha,s,n}{p}\!=\!
\frac{
f_{\scriptscriptstyle \mathrm{B}}(\particle{\PFphi{s}{n}{1,p}}{p})
}{
f_{\scriptscriptstyle \mathrm{B}}^{\scriptscriptstyle \mathrm{p}}(\particle{\PFphi{s}{n}{1,p}}{p}|\observationn{n})
}
\frac{p_{\scriptscriptstyle \mathrm{B}}^{\scriptscriptstyle(q)}}{P}$
 normalized as 
$\weight{\alpha,s,n}{p}
\!=\!
\nicefrac{
p_{\scriptscriptstyle \mathrm{B}}^{\scriptscriptstyle(q)} \weightt{\alpha,s,n}{p}
}{
\sum \weightt{\alpha,s,n}{p}
}$.

\subsubsection{An Efficient Birth Proposal PDF}\label{sec:proposal-PDF}  
For each new \gls{pf} $q(s) \!=\! s\!-\!|\setFeaturesLegacy{n}|$ with $s\!\in\!\setFeaturesNew{n}$, 
our choice of proposal \gls{pdf} is $f_{\scriptscriptstyle \mathrm{B}}^{\scriptscriptstyle \mathrm{p}}(\PFphi{s}{n}{1}|\observationn{n})\!\coloneqq\!
f^{\scriptscriptstyle \mathrm{p}}(\psfv{s}{n}{1}|\observationn{n})
f(\PFgamma{s}{n}{j})f(\PFmu{s}{n}{j})$, where 
$f^{\scriptscriptstyle \mathrm{p}}(\psfv{s}{n}{1}|\observationn{n})\!\coloneqq\!\mathcal{N}(\psfv{s}{n}{1};\bm{\mu}_{\scriptscriptstyle q,n}^{\scriptscriptstyle \mathrm{sfv}},\bm{C}_{\scriptscriptstyle q,n}^{\scriptscriptstyle \mathrm{sfv}})$ is a Gaussian spatial proposal \gls{pdf} parameterized by mean $\bm{\mu}_{\scriptscriptstyle q,n}^{\scriptscriptstyle \mathrm{sfv}}\!\in\!\realsetone{3}$ and covariance matrix
$\bm{C}_{\scriptscriptstyle q,n}^{\scriptscriptstyle \mathrm{sfv}}\!\in\!\realset{3}{3}$.
We obtain these parameters as follows:

First, we define an observation ``residual'' 
$\widetilde{\bm{z}}_{\scriptscriptstyle n}^{\scriptscriptstyle (j)}\coloneqq \bm{\Pi}^\perp_{\scriptscriptstyle j} \bm{z}_{\scriptscriptstyle n}^{\scriptscriptstyle (j)}$ 
using the projector 
$\bm{\Pi}_{\scriptscriptstyle j}^\perp \!\coloneqq\! \eye{\Nz}\!-\!\bm{\Psi}_{\scriptscriptstyle j}\bm{\Psi}_{\scriptscriptstyle j}^\dag$ 
onto the residual-subspace orthogonal to the signal-subspace spanned by the column vectors of the steering matrix 
$\bm{\Psi}_{\scriptscriptstyle j}(\stateHat{n|n-1}^{\text{\tiny MMSE}},
\{\PFphiHat{s}{n-1}{J}^{\text{\tiny MMSE}}
\}_{\setFeaturesLegacy{n}}
)
\coloneqq 
\big[ \steerVec{j}\!(\stateHat{n|n-1}^{\text{\tiny MMSE}} 
,
\!\PFphiHat{0}{n-1}{J}^{\text{\tiny MMSE}} 
)
\hdots \steerVec{j}\!(
\stateHat{n|n-1}^{\text{\tiny MMSE}}
,
\!\PFphiHat{|\setFeaturesLegacy{n}|}{n-1}{J}^{\text{\tiny MMSE}} )
\big] \!\!\in\!\! \complexset{\Nz}{(|\setFeaturesLegacy{n}|+1)}$\!, 
which captures the residual of the signal \textit{not} captured by the \gls{los} $s\!=\!0$ or legacy \glspl{pf} 
$\setFeaturesLegacy{n}$.
Here, we used the predicted state estimate $\stateHat{n|n-1}^{\text{\tiny MMSE}}\!\coloneqq\!\int \state{n}\Mbeta(\state{n})\mathrm{d}\state{n}\!\approx\!\sum_{\scriptscriptstyle p=1}^{\scriptscriptstyle P}\weight{\beta,n}{p} \particle{\state{n}}{p}$.
We further stack 
$\widetilde{\bm{z}}_{\scriptscriptstyle n} \!\coloneqq \!\big[{\widetilde{\bm{z}}_{\scriptscriptstyle n}^{\scriptscriptstyle (1)^\trp}}
\!\hdots\ist
{\widetilde{\bm{z}}_{\scriptscriptstyle n}^{\scriptscriptstyle (J)^\trp}}\big]^{\!\trp}
$.
Then we draw $N_{\scriptscriptstyle \mathrm{g}}$ candidate positions $\bm{p}^{\scriptscriptstyle \mathrm{sfv}}_{\scriptscriptstyle i}$ uniformly from the partition $\mathcal{P}_q$, i.e., from $\mathcal{U}(\mathcal{P}_q)$.
Next, we evaluate a cost function $P_{\scriptscriptstyle \mathrm{B}}(\bm{p}^{\scriptscriptstyle \mathrm{sfv}}_{\scriptscriptstyle i},\widetilde{\bm{z}}_{\scriptscriptstyle n})\!\eqqcolon\!\widetilde{w}_{\scriptscriptstyle i}$ to obtain weights $w_{\scriptscriptstyle i} \!=\! \widetilde{w}_{\scriptscriptstyle i}/ \sum_{\ddot{\imath}=1}^{N_{\scriptscriptstyle \mathrm{g}}}\widetilde{w}_{\scriptscriptstyle \ddot{\imath}}$.
For its computational efficiency and robustness, our choice of a cost function is the coherent Bartlett spectrum $P_{\scriptscriptstyle \mathrm{B}}(\bm{p}^{\scriptscriptstyle \mathrm{sfv}}_{\scriptscriptstyle i},\widetilde{\bm{z}}_{\scriptscriptstyle n})\!\coloneqq\!
\big|
\sum_{\scriptscriptstyle j=1}^{\scriptscriptstyle J}\frac{1}{\Nz}
\left.{\widetilde{\bm{z}}_{\scriptscriptstyle n}^{\scriptscriptstyle (j)}}\right.^{\!\herm} \steerVec{j}(\stateHat{n|n-1}^{\text{\tiny MMSE}},\bm{p}^{\scriptscriptstyle \mathrm{sfv}}_{\scriptscriptstyle i})
\big|^2$.
Finally, we mode-match the mean of our Gaussian proposal \gls{pdf} $f^{\scriptscriptstyle \mathrm{p}}(\psfv{s}{n}{1}|\observationn{n})$ with the empirical maximum estimate of the Bartlett spectrum, i.e., 
$\bm{\mu}_{\scriptscriptstyle q,n}^{\scriptscriptstyle \mathrm{sfv}}\!\coloneqq \!\bm{p}^{\scriptscriptstyle \mathrm{sfv}}_{\scriptscriptstyle i^\star}$ with 
$i^\star \!\coloneqq \!\argmax_{i} P_{\scriptscriptstyle \mathrm{B}}(\bm{p}^{\scriptscriptstyle \mathrm{sfv}}_{\scriptscriptstyle i},\widetilde{\bm{z}}_{\scriptscriptstyle n})$, and moment-match the covariance matrix with an empirical second central moment estimate,
i.e., $\bm{C}_{\scriptscriptstyle q,n}^{\scriptscriptstyle \mathrm{sfv}} \!\coloneqq\! \sum_{\scriptscriptstyle i=1}^{\scriptscriptstyle N_{\scriptscriptstyle \mathrm{g}}} w_{\scriptscriptstyle i} 
\big(
    \bm{p}^{\scriptscriptstyle \mathrm{sfv}}_{\scriptscriptstyle i} \!-\! \bm{\mu}_{\scriptscriptstyle q,n}^{\scriptscriptstyle \mathrm{sfv}}
\big)
\big(
    \bm{p}^{\scriptscriptstyle \mathrm{sfv}}_{\scriptscriptstyle i} \!-\! \bm{\mu}_{\scriptscriptstyle q,n}^{\scriptscriptstyle \mathrm{sfv}}
\big)^{\!\trp}
$.
This choice of proposal \gls{pdf} results in particles 
$\V{p}_{\!\scriptscriptstyle s,n}^{\text{\tiny sfv}\scriptscriptstyle(p)}$ 
of new \glspl{pf} $s\!\in\!\setFeaturesNew{n}$ being drawn from 
$f^{\scriptscriptstyle \mathrm{p}}(\psfv{s}{n}{1}|\observationn{n})$
such that most of them are located in a spatial region where the residual $\widetilde{\bm{z}}_{\scriptscriptstyle n}^{\scriptscriptstyle (j)}$ has most power, and weights $\weightt{\alpha,s,n}{p}\!\propto\!
\frac{
1_{\mathcal{P}_q}\big(\V{p}_{\!\scriptscriptstyle s,n}^{\text{\tiny sfv}\scriptscriptstyle(p)}\big)
}{
f^{\scriptscriptstyle \mathrm{p}}(\V{p}_{\!\scriptscriptstyle s,n}^{\text{\tiny sfv}\scriptscriptstyle(p)}|\observationn{n})
}$.

\subsubsection{Update Messages}\label{sec:PR-update-messages}
Now that \glspl{pr} of prediction messages have been introduced, we can compute the approximate update messages.
In particular, the moments in \eqref{eq:muiota}--\eqref{eq:Comega} are approximated using the terms in \eqref{eq:musnj1}--\eqref{eq:Komega},
where we insert \glspl{pr} of prediction messages to evaluate the marginalization integrals therein.
In particular, we evaluate
$\widetilde{\iota}(\particle{\state{n}}{p};\observation{n}{j}) \!=\! \mathcal{CN}\big(\observation{n}{j};\muiota(\particle{\state{n}}{p}),\Ciota(\particle{\state{n}}{p})\big)$,
$\widetilde{\kappa}(\particle{\PFphi{s}{n}{j,p}}{p},\PFr{s}{n}{j};\observation{n}{j})\!=\! \mathcal{CN}\big(\observation{n}{j};\mukappa(\particle{\PFphi{s}{n}{j,p}}{p},\PFr{s}{n}{j}),\Ckappa(\particle{\PFphi{s}{n}{j,p}}{p},\PFr{s}{n}{j})\big)$,
$\widetilde{\nu}(\particle{\etan{n}}{j,p};\observation{n}{j})\!=\! \mathcal{CN}\big(\observation{n}{j};\munu,\Cnu(\particle{\etan{n}}{j,p})\big)$, 
and
$\widetilde{\omega}(\PRr{s}{n}{j};\observation{n}{j})\!=\! \mathcal{CN}\big(\observation{n}{j};\muomega(\PRr{s}{n}{j}),\Comega(\PRr{s}{n}{j})\big)$
with \glspl{pr} of mean and covariance terms
from Supplementary Material, Sec.\,S-V.

\subsubsection{Beliefs}
With \glspl{pr} of both prediction 
and update messages, \glspl{pr} of the beliefs~\eqref{eq:belief-state}--\eqref{eq:belief-eta} are 
$\belief(\state{n}) \!\approx \!
\sum_{\scriptscriptstyle p=1}^{\scriptscriptstyle P} \weight{\bm{x},n}{p} \delta(\state{n}\!- \particle{\state{n}}{p})$
for the \gls{mt} state
$\belief(\PFphi{s}{n}{j},1) \!\approx\!
\sum_{\scriptscriptstyle p=1}^{\scriptscriptstyle P}
\weight{\bm{y},s,n}{p} \delta(\PFphi{s}{n}{j}\!-\! \particle{\PFphi{s}{n}{j,p}}{p})$ for \gls{pf} states,
and 
$\belief(\etan{n}^{\scriptscriptstyle(j)}) \!\approx\!
\sum_{\scriptscriptstyle p=1}^{\scriptscriptstyle P}
\weight{\eta,n}{j,p} \delta(\etan{n}^{\scriptscriptstyle(j)}\!- \particle{\etan{n}}{j,p})$
for the noise variance
with weights 
$\weight{\bm{x},n}{p}\!=\!\frac{\weightt{\bm{x},n}{p}}{\normConst{\bm{x},n}{}}$, 
$\weight{\bm{y},s,n}{p}\!=\! 
\frac{\weightt{\bm{y},s,n}{p}}{\normConst{\bm{y},s,n}{}}$, and
$\weight{\eta,n}{j,p}\!=\!\frac{\weightt{\eta,n}{j,p}}{\normConst{\eta,n}{(j)}}$
with
\begin{align}
    \weightt{\bm{x},n}{p} 
    &\coloneqq 
    \weight{\beta,n}{p}
    \prod_{\scriptscriptstyle j\in \setAnchors}
    \mathcal{CN}\Big(\observation{n}{j};\muiota(\particle{\state{n}}{p}),\Ciota(\particle{\state{n}}{p})\Big)
    \\
    \weightt{\bm{y},s,n}{p}
    &\coloneqq 
    \weight{\alpha,s,n}{p}
    \prod_{\scriptscriptstyle j\in \setAnchors}
    \mathcal{CN}\big(\observation{n}{j};\mukappa(\particle{\PFphi{s}{n}{j,p}}{p},1),\Ckappa(\particle{\PFphi{s}{n}{j,p}}{p},1)\big)
    \nn
    \\
    \weightt{\eta,n}{j,p}
    &\coloneqq 
    \weight{\xi,n}{j,p}
    \mathcal{CN}\Big(\observation{n}{j};\munu,\Cnu(\particle{\etan{n}}{j,p})\Big)
\end{align}
derived 
in Supplementary Material, Sec.\,S-III%
, and normalized with constants
$\normConst{\bm{x},n}{}\!=\!{\sum_{\scriptscriptstyle p=1}^{\scriptscriptstyle P} \weightt{\bm{x},n}{p}}$, 
$\normConst{\eta,n}{(j)}\!=\!{\sum_{\scriptscriptstyle p=1}^{\scriptscriptstyle P}
\weightt{\eta,n}{j,p}}$, and
\begin{align*}
    \normConst{\bm{y},s,n}{} \! = \! 
    \Big(
        \sum_{\scriptscriptstyle p=1}^{\scriptscriptstyle P} 
        \weightt{\bm{y},s,n}{p}
    \Big)
    \!+\!
    \Big(
        \!1\!-\!\!
        \sum_{\scriptscriptstyle p=1}^{\scriptscriptstyle P} 
        \weight{\alpha,s,n}{p}\!
    \Big)
    \prod_{\scriptscriptstyle j\in\setAnchors}
    \widetilde{\kappa}\big(\PFphi{s}{n}{j},0;\observation{n}{j}\big)
\end{align*}
which is derived in Supplementary Material, Sec.\,S-IV.
The \gls{pf} normalization constant $\normConst{\bm{y},s,n}{}$ requires computing
$\widetilde{\kappa}\big(\PFphi{s}{n}{j},0;\observation{n}{j}\big)
\!=\! \mathcal{CN}\big(\observation{n}{j};\mukappa(\PFphi{s}{n}{j},0),\Ckappa(\PFphi{s}{n}{j},0)\big)$ under the null hypothesis $\mathcal{H}_0\!:\!\PFr{s}{n}{j}\!=\!0$ once, which is constant for all particles. Using the weights $\weight{\bm{y},s,n}{p}$, the marginal posterior PF existence probabilities are calculated by evaluating \eqref{eq:existenceProb}.

The moment-matched approximation for the $\PRr{s}{n}{j}\!=\!1$ branch of the \gls{ppr} belief in~\eqref{eq:belief-PR} is normalized as
\begin{align}\label{eq:PR-normalization}
    \widetilde{p}(\PRr{s}{n}{j}\!=\!1)
    =
    \frac{
        \Mzeta{s}{j}\!(1)\widetilde{\omega}(1;\observation{n}{j})
    }{
        \Mzeta{s}{j}\!(1)\widetilde{\omega}(1;\observation{n}{j})
        +
        \big(1\!-\!\Mzeta{s}{j}\!(1)\big)
        \widetilde{\omega}(0;\observation{n}{j})
    }.
\end{align}
After computing the belief of each state, we perform systematic resampling~\cite[Alg.\,2]{Arulampalam02PFtutorial} which reduces particle degeneracy and implies equal weights of resampled particles.
To counteract particle impoverishment, each particle is convolved with a Gaussian regularization kernel 
with covariance matrix 
moment-matched with the second central moment of the belief
and scaled by the optimal kernel bandwidth $\mathrm{h}_{\text{\tiny opt}}$~\cite[p.\,253]{Musso2001PFCH12}.

\subsection{Computational Complexity and GPU Acceleration}\label{sec:GPU-acceleration}  
Let $\Nfeaturest{n}\!\coloneqq\!\Nfeatures{n}\!+\!1$ denote the number of \glspl{pf} including the \gls{los}.
Liang et al.\,\cite{LiaLeiMey:TSP2025} presented a direct \gls{slam} method that scales with $\mathcal{O}\big(P J \Nfeaturest{n} \Nz^3\big)$ and an ``alternative approximation''
that scales with $\mathcal{O}\big(\Nz^2 + P \big(J(\Nz \Nfeaturest{n}^2 + \Nfeaturest{n}^4)\big)\big)$ according to\,\cite{LiaLeiMey:TSP2025-SuppDoc}. 
The dominating term of the latter is $PJ\Nz \Nfeaturest{n}^2$ in practical settings, i.e., $P\!\gg\!\Nz$ and $\Nz\!\gg \!\Nfeaturest{n}^2$,
hence the algorithm scales \textit{linearly} in the number of particles $P$, the number of \glspl{pa} $J$, and the length of observed data $\Nz$ in practice.
Both implementations showed almost identical performance\,\cite{LiaLeiMey:TSP2025-SuppDoc}.

In Supplementary Material, Sec.\,S-V we follow the alternative approximation from~\cite{LiaLeiMey:TSP2025-SuppDoc} which was originally implemented on a CPU.
The arguably largest potential of this algorithm is that it naturally admits parallelization.
Reflecting the current zeitgeist, our \gls{bp} method is implemented on a \gls{gpu}, parallelizing over both $P$ particles and $J$ \glspl{pa}. 
In Section~\ref{sec:Results}, we show that this achieves an order-of-magnitude speedup over the preceding CPU-based implementation 
when parameterized at equal system parameters, hinting at real-time capability.
The computation of update weights in our algorithm scales with 
$\mathcal{O}\big(
J(\Nz\Nfeatures{n}^3 + \Nfeatures{n}^4) +
P \big(J(\Nz \Nfeaturest{n}^2 + \Nfeaturest{n}^3)\big)
\big)$, while our new birth proposal \gls{pdf} has complexity 
$\mathcal{O}
\big( 
J (\Nz \Nfeaturest{n}^2 + \Nfeaturest{n}^3) + 
Q J N_{\scriptscriptstyle \mathrm{g}} \Nz
\big)$, leaving us with a total computational complexity of
$\mathcal{O}
\big(
Q J N_{\scriptscriptstyle \mathrm{g}} \Nz +
J(\Nz\Nfeatures{n}^3 + \Nfeatures{n}^4) +
P \big(J(\Nz \Nfeaturest{n}^2 + \Nfeaturest{n}^3)\big)
\big)$ per time step $n$.
The last term is dominant in practical implementations, where the per-particle, per-\gls{pa} cost is $\mathcal{O}\big(\Nz \Nfeaturest{n}^2 + \Nfeaturest{n}^3\big)$ after parallelization.

\section{Posterior Cram\'er--Rao Lower Bound}\label{sec:PCRLB}  
Following~\cite{Deutschmann25Asilomar}, we derive two \glspl{pcrlb} for the \gls{dmimo} \gls{slam} problem.
Let $\widetilde{K}\!\coloneqq\!K\!+\!1$.
First, we decompose amplitudes%
\footnote{For deriving the \gls{pcrlb}, amplitudes $\widetilde{\unslant[-.25]{\varrho}}_{\!\scriptscriptstyle k,n}^{\scriptscriptstyle(j)}$ absorb the pathloss instead of the steering vectors $\steerVec{j}$ as well as the carrier-phase term $\exp \!(\!-\mathrm{j}\frac{2\pi}{\lightspeed}\fc \lVert \acute{\V{r}}\rVert )$ from~\eqref{eq:unit-modulus-array-response}. For details, see Supplementary Material, Sec.\,S-VII.} 
into moduli $\rv{a}_{\!\scriptscriptstyle k,n}^{\scriptscriptstyle(j)}\!\in\!\realsetone{}$ and phases $\unslant[-.25]{\varphi}_{\!\scriptscriptstyle k,n}^{\scriptscriptstyle(j)}$ and stack them into vectors 
$\RVmodulivec{n}{(j)}\!\coloneqq\!
\big[{\rv{a}_{\!\scriptscriptstyle 0,n}^{\scriptscriptstyle(j)}}
\ist\hdots\ist
{\rv{a}_{\!\scriptscriptstyle K,n}^{\scriptscriptstyle(j)}}
\big]^{\!\trp}$
and
$\RVphasevec{n}{(j)}$ per \gls{pa}, and 
into vectors
$\RVmodulivec{n}{}\!\coloneqq\!
\big[{\RVmodulivec{n}{(1)}}^{\!\trp}\!
\ist\hdots\ist
{\RVmodulivec{n}{(J)}}^{\!\trp}
\big]^{\!\trp}\!\in\!\realsetone{\widetilde{K}J}$
and 
$\RVphasevec{n}{}\!\coloneqq\!
\big[{\RVphasevec{n}{(1)}}^{\!\trp}\!
\ist\hdots\ist
{\RVphasevec{n}{(J)}}^{\!\trp}
\big]^{\!\trp}\!\in\!\realsetone{\dimPhase}$
for all components $k$.
%
The (i) 
\textit{noncoherent} \gls{pcrlb} is obtained by treating each component phase as a separate \gls{rv} per \gls{pa}, i.e., 
$\RVphasevec{n}{(j)}\!=\!
\big[
\unslant[-.25]{\varphi}_{\!\scriptscriptstyle 0,n}^{\scriptscriptstyle(j)}
\ist\hdots\ist
\unslant[-.25]{\varphi}_{\!\scriptscriptstyle K,n}^{\scriptscriptstyle(j)}
\big]^{\!\trp}$ and $\dimPhase\!=\!\widetilde{K}J$,
while the (ii) \textit{coherent} \gls{pcrlb} is obtained by treating the phase 
$\unslant[-.25]{\varphi}_{\!\scriptscriptstyle k,n}^{\scriptscriptstyle(j)}
\!\eqqcolon\!
\unslant[-.25]{\varphi}_{\!\scriptscriptstyle k,n}\,\forall j\!\in\!\setAnchors$
as a single \gls{rv} common among distributed \glspl{pa} $j,j'\!\in\!\setAnchors$, i.e., 
$\RVphasevec{n}{(j)}\!=\!
\big[
\unslant[-.25]{\varphi}_{\!\scriptscriptstyle 0,n}
\! \hdots \ist
\unslant[-.25]{\varphi}_{\!\scriptscriptstyle K,n}
\big]^{\!\trp}
\!= \!
\RVphasevec{n}{}$
and $\dimPhase\!=\!\widetilde{K}$.
We stack all parameters into a joint state vector
$\RVetaglobal{n}\!\coloneqq\!
\big[ \RVstate{n}^\trp \iist {\RVpMVAposStacked}^\trp \iist {\RVphasevec{n}{}}^\trp \iist {\RVmodulivec{n}{}}^{\trp} \iist \RVetan{n}\big]^\trp \! \in \!\realsetone{\dimGlobal}$
of dimension $\dimGlobal\!=\!6+3K+\dimPhase\!+\!J\widetilde{K}\!+\!1$,
where the map is captured by a vector
$\RVpMVAposStacked\!\coloneqq\![{\RVpsfv{1}{n}{J}}^{\!\!\!\trp} \ist\dots\ist {\RVpsfv{K}{n}{J}}^{\!\!\!\trp}]^\trp \!\in\! \realset{3K}{1}$ 
of stacked \gls{sfv} positions.
For this \gls{pcrlb} derivation and the following experiments, we assume a common noise variance across all \glspl{pa}, i.e., $\etan{n}^{\scriptscriptstyle(j)}\!\eqqcolon\!\etan{n}\, \forall j \!\in\! \setAnchors$.

We are ultimately interested in obtaining the global~\gls{pcrlb}
\begin{align}\label{eq:posteriorCRLB}
	\PCRLB = 
	\big(
		\FIMglobal{n} + \FIMstep{n}{n\!-\!1}
	\big)^{-1} %
	\quad \in  \realset{\dimGlobal}{\dimGlobal}
\end{align}
that is a lower bound on the \gls{mse} matrix~\cite[eq.\,(29)]{VanTrees2007PCRLB} 
of any estimator%
\footnote{The expectation is to be taken under the joint \gls{pdf} $f(\etaglobal{n}, \observationn{n}|\observationn{1:n-1})$.
The notation $\M{X} \succeq \M{0}$ is to be interpreted as $\M{X}$ being positive semidefinite~\cite{kay1993estimation}.} 
$\mathbb{E}
\!\big((\RVetaglobalHat{n}\!-\RVetaglobal{n} ) (\RVetaglobalHat{n}\!-\RVetaglobal{n} )^\trp\big) \succeq \PCRLB$.
The \gls{pcrlb} matrix $\PCRLB$ is the inverse of the posterior information matrix 
$\FIMstep{n}{n}\!\coloneqq\!\FIMglobal{n} \!+\! \FIMstep{n}{n\!-\!1}$
that is computed through the information fusion of the information matrix $\FIMglobal{n}$ about the global parameters of interest $\RVetaglobal{n}$ obtained from a snapshot of observations $\observationn{n}$ at the current time step $n$ with the predicted information matrix 
$\FIMstep{n}{n\!-\!1}$. 
Under a linear Gaussian state-transition \gls{pdf} 
$f(\etaglobal{n}|\etaglobal{n-1})\!=\!\mathcal{N}(\etaglobal{n};\transitionmatrix \etaglobal{n-1},\processNoiseCov)$
the predicted 
information matrix is~\cite[eq.\,(16)]{Hernandez02PCRLB}
\begin{align}\label{eq:priorFIM}
	\FIMstep{n}{n\!-\!1} = 
	\left(
		\transitionmatrix \, \FIMstep{n\!-\!1}{n\!-\!1}^{-1} \, \transitionmatrix^\trp + \processNoiseCov
	\right)^{-1} \,,
\end{align}
with 
state-transition matrix $\transitionmatrix$ and process noise covariance matrix $\processNoiseCov$.
%
Using the \gls{pcrlb} matrix in~\eqref{eq:posteriorCRLB}, we define the \gls{peb} as $\PEB\!\coloneqq\!\sqrt{\tr\big(\big[\PCRLB\big]_{\scriptscriptstyle1:3,1:3}\big)}$ 
and the \gls{meb} as
$\MEB{k,n}\!\coloneqq\!\sqrt{\tr\big(\big[\PCRLB\big]_{\scriptscriptstyle\bm{i}_k,\bm{i}_k}\big)}$ with index range $\bm{i}_k\!\coloneqq\!7\!+\!3(k\!-\!1):6\!+\!3k$ for \gls{sfv} $k$.

\subsection{Global Snapshot FIM}\label{sec:FIMg}  
The Bayesian snapshot information matrix $\FIMglobal{n} \!=\! \mathbb{E}_{\RVetaglobalSmall{n}|\RVobservationn{1:n-1}}\!(\FIMclassic{n})$ is computed as the expectation under the prior \gls{pdf} $f(\etaglobal{n}|\observationn{1:n-1})$ of the ``classic'' snapshot \gls{fim}~\cite{Tichavsky98PCRLB}.
Assuming that each anchor $j$ contributes independent information about $\RVetaglobal{n}$, i.e., assuming conditionally independent observations $\observation{n}{j}\!$, 
the classic snapshot \gls{fim}~\cite{Fascista23RadioStripesICC}
\begin{align}\label{eq:FIMclassic}
    \FIMclassic{n} = 
    \sum\nolimits_{\scriptscriptstyle j=1}^{\scriptscriptstyle J} \jacobgn{j} \FIMch{n}{j} {\jacobgn{j}}^\trp
    \quad \in  \realset{\dimGlobal}{\dimGlobal}
\end{align}
is computed as the 
sum of the local channel \glspl{fim} $\FIMch{n}{j} \!\in\! \realset{\dimLocal}{\dimLocal}$ contributed by all $J$ \glspl{pa}, and propagated via the Jacobian matrices $\jacobgn{j}\!\coloneqq\! \nicefrac{\partial {\etach{n}{j}}^{\!\!\trp}}{\partial \etaglobal{n}} \!\in\!  \realset{\dimGlobal}{\dimLocal}$ from local channel parameter level to global parameter level.
The Jacobian matrices in~\eqref{eq:FIMclassic} are derived in Supplementary Material, Sec.\,S-VII-B.

\subsection{Local Per-Anchor Channel FIM}\label{sec:FIMch}  

The \gls{pa}-local channel parameter vector 
$\RVetach{n}{j}\!\coloneqq\!\big[
{\RVelVecx{n}{j}}^\trp\ist 
{\RVazVecx{n}{j}}^\trp \ist  
{\RVdelayVecx{n}{j}}^\trp \ist 
{\RVphasevec{n}{(j)}}^\trp \ist 
{\RVmodulivec{n}{(j)}}^\trp \ist
\etan{n}\big]^\trp \! \in \! \realsetone{\dimLocal}$
is of dimension
$\dimLocal\!=\!5\widetilde{K}\!+\!1$,
which contains the stacked elevation and azimuth angles, delays, and amplitude phases and moduli of all $\widetilde{K}$ components.
The local channel \gls{fim} at \gls{pa} $j$ is defined as~\cite[eq.\,(10)]{VanTrees2007PCRLB} 
$\FIMch{n}{j}\!\coloneqq\!-\mathbb{E}_{\RVobservation{n}{j}|\RVetach{n}{j}}\!
\big( 
\nabla_{\!\etach{n}{j}} 
(
	\nabla_{\!\etach{n}{j}}
	\! \ln f(\observation{n}{j}|\etach{n}{j})
)^\trp
\big)$, which corresponds to 
\begin{align}\label{eq:FIMch}
	\hspace{-2mm}\big[\FIMch{n}{j}\big]_{\scriptstyle \grave{\imath},\acute{\imath}} \! = 
	\begin{cases}
		\!
		\frac{\Nz}{{\etan{n}}^2},
		\qquad\qquad\qquad~
		[\etach{n}{j}]_{\scriptscriptstyle \grave{\imath}}
		\equiv 
		[\etach{n}{j}]_{\scriptscriptstyle \acute{\imath}}
		\equiv 
		\etan{n}
		\\
		\frac{2}{\etan{n}}\Re \left(
			\frac{\partial {\widetilde{{\varrho}}_{\!\scriptscriptstyle k,n}^{\scriptscriptstyle(j)^{\!\ast}}} {\steerVecx{k}{n}{j}}^{\!\herm} }{\partial [\etach{n}{j}]_{\scriptscriptstyle \grave{\imath}}}
			\frac{\partial {\widetilde{{\varrho}}_{\!\scriptscriptstyle k',n}^{\scriptscriptstyle(j)}} \steerVecx{k'}{n}{j} }{\partial [\etach{n}{j}]_{\scriptscriptstyle \acute{\imath}}}
		\right),
		\quad
		\text{else}
	\end{cases}\hspace{-2mm}
\end{align}
under a complex Gaussian likelihood $f(\observation{n}{j}|\etach{n}{j})$~\cite[Sec. 15.7]{kay1993estimation}.
The individual local channel \gls{fim} entries in~\eqref{eq:FIMch} are derived in Supplementary Material, Sec.\,S-VII-A.

\section{Experiments and Results}\label{sec:Results}  
In the following, we consider two variants of our method.
The unmodified method from Section~\ref{sec:BP-method} is referred to as the \gls{nzm} method.
Its zero-mean counterpart, using the amplitude prior
$\RVamplitude{s}{n}{j}|\RVPFgamma{s}{n}{j} \!\sim\! \CN(0,\RVPFgamma{s}{n}{j})$,
is referred to as the \gls{zm} method. 
In the latter case,
\eqref{eq:musnj1}--\eqref{eq:musnj4} and
\eqref{eq:Kiota}--\eqref{eq:Komega} reduce to zero vectors and matrices, respectively.
We compare both \gls{sfv}-based variants with \gls{ref} from~\cite{LiaLeiMey:TSP2025},\!\footnote{\Gls{ref} was extended to three-dimensional scenarios, for array-equipped \glspl{pa}, and by our birth proposal \gls{pdf} adapted for \glspl{va}.}
which is based on the \gls{va} model and also uses a zero-mean amplitude prior as \gls{zm}.
We also show \gls{nzm}-\acrshort{nc}, i.e., \gls{nzm} applied to \gls{nc} data.

\subsection{Common Simulation Setup}

\Gls{mt} state particles $\{\state{0}\}_{\scriptscriptstyle p=1}^{\scriptscriptstyle P}$ are drawn from 
$f(\state{0})
\!=\!
\mathcal{U}(\pos{0};\mathbf{p}_{\text{\tiny min}},\mathbf{p}_{\text{\tiny max}})
\,\mathcal{U}(\vel{0};-\mathbf{v}_{\text{\tiny max}},\mathbf{v}_{\text{\tiny max}})$ 
with a Cartesian axis-aligned cuboidal support, bounded 
by $\mathbf{p}_{\text{\tiny min}}\!=\!-[3.5\iist1\iist1]^\trp$ and $\mathbf{p}_{\text{\tiny max}}\!=\![4.5\iist4.5\iist1]^\trp$ and maximum velocity $\mathbf{v}_{\text{\tiny max}}\!=\![0.5\iist0.5\iist0.5]^\trp\SI{}{\metre\per\second}$.
Noise variance particles are drawn from a uniform prior 
$f(\etan{0}^{\scriptscriptstyle(j)})\!=\!\mathcal{U}(\etan{0}^{\scriptscriptstyle(j)};\eta_{\text{\tiny min}},\eta_{\text{\tiny max}})$ with $\eta_{\text{\tiny min}}\!=\!10^{-9}$ and $\eta_{\text{\tiny max}}\!=\!10^{-4}$.
We choose $\Nfeatures{0}\!=\!0$, i.e., we initialize only the \gls{los} $s\!=\!0$, in the same way as we introduce new \glspl{pf} described below.

\begin{figure}
        \vskip 0pt	
        \begin{center}
            \newcommand{\LWrays}{0.05pt}
            \ifthenelse{\equal{\DEBUG}{false}}
            {
                \def\datapath{.}
                \setlength{\figurewidth}{0.8\linewidth}
                \tikzexternaldisable    
    	           \input{figures/scenario}
                \tikzexternalenable
                \tikzexternaldisable    
    	           \input{figures/scenario-det}
                \tikzexternalenable
            }
            {
              \includegraphics[width=\linewidth]{figures/scenarioDummy.png}
            }
        \end{center}
        \vspace{-6mm}
        \captionof{figure}{
        Scenarios for synthetic data generation with
        $K\!=\!4$ visible surfaces: a) realizations of random \gls{mt} trajectories, and b) a deterministic trajectory.
        }
        \label{fig:scenario}
    \vspace{-5mm}%
\end{figure}

\newcommand{\markRepeat}{15}
\newcommand{\markSize}{1}

\begin{figure*}
        \vskip 0pt	
        \begin{center}
            \newcommand{\LWbound}{1.0pt}
            \newcommand{\LWestimates}{0.5pt}
            \colorlet{mycolor1}{IEEEblue}
            \colorlet{mycolor2}{IEEEred}
            \colorlet{mycolor3}{IEEEgreen}
            \colorlet{mycolor4}{IEEElightblue}
            \def\datapath{.}
            \setlength{\figurewidth}{0.968\linewidth}
            \setlength{\figureheight}{0.13\linewidth}
            \tikzexternaldisable 
	        \input{figures/results/SFV-4plot}
            \tikzexternalenable
        \end{center}
        \vspace{-6mm}
        \captionof{figure}{Mapping 
        \gls{rmse}
        for $k\!\in\!\{1 \hdots 4\}$ (left to right) vs. \gls{meb} evaluated on synthetic data of Experiment 1. 
        Estimates $\psfvHat{s}{n}{J}$ and ground truth $\psfv{\,k}{n}{J}$ were associated using the Hungarian method~\cite{Kuhn1955Hungarian}, missed detections and false alarms were not evaluated.}
        \label{fig:MEB}
    \vspace{-5mm}%
\end{figure*}

\paragraph*{\Gls{mt} state evolution}
The \gls{mt} state-transition \gls{pdf} 
$f(\state{n}|\state{n-1})=\mathcal{N}(\state{n};\transitionmatrix_{\text{\tiny a}}\state{n-1},\processNoiseCov_{\text{\tiny a}})$
is defined according to a \gls{ncv} state-transition model~\cite[Sec.\,6.3.2]{BarShalom04Tracking}
with state transition matrix
\begin{align*}
\left[\transitionmatrix_{\text{\tiny a}}\right]_{\scriptscriptstyle \grave{\imath},\acute{\imath}} = 
    \begin{cases}
        1, & \grave{\imath}=\acute{\imath} \\
        \mathrm{T}, & 
        (\grave{\imath},\acute{\imath}) \!=\! (1,4)
        \lor 
        (\grave{\imath},\acute{\imath}) \!=\! (2,5)
        \lor 
        (\grave{\imath},\acute{\imath}) \!=\! (3,6)
        \\
        0, & \text{else}
    \end{cases}
\end{align*}
with $\mathrm{T}$ being the time interval between time steps $n$ and $n\!-\!1$,
and kinematic process noise covariance matrix
$\processNoiseCov_{\text{\tiny a}}\!=\!\sigma_{\scriptscriptstyle \mathrm{v}}^2\mathbf{\Gamma}\mathbf{\Gamma}^\trp$
with gain matrix $\mathbf{\Gamma}\!=\!\big[\frac{\mathrm{T}^2}{2}\eye{3},\mathrm{T}\eye{3}\big]^\trp$, where ${\sigma}_{\scriptscriptstyle \mathrm{v}}^2$ is the process noise variance of the kinematic \gls{mt} state.

\paragraph*{ Noise variance evolution}
Following~\cite{LiaLeiMey:TSP2025}, the noise variance evolution is described by a Gamma \gls{pdf} $f(\etan{n}^{\scriptscriptstyle(j)}|\etan{n-1}^{\scriptscriptstyle(j)})\!=\!\mathcal{G}\big(\etan{n}^{\scriptscriptstyle(j)};\mathrm{c}_\eta,\frac{\etan{n-1}^{\scriptscriptstyle(j)}}{\mathrm{c}_\eta}\big)$ with mean $\etan{n-1}^{\scriptscriptstyle(j)}$ and variance $\frac{{(\etan{n-1}^{\scriptscriptstyle(j)}})^2}{\mathrm{c}_\eta}$, parameterized by a chosen constant $\mathrm{c}_\eta\!=\!\num{100}$. 

\paragraph*{ PF state time-evolution}
We assume that the temporal state-transition \gls{pdf} of the continuous \gls{pf} state factorizes as $f(\PFphi{s}{n}{1}|\PFphi{s}{n\!-\!1}{J})\!=\!f(\psfv{s}{n}{1}|\psfv{s}{n-1}{J})f(\PFmu{s}{n}{1}|\PFmu{s}{n-1}{J})f(\PFgamma{s}{n}{1}|\PFgamma{s}{n-1}{J})$.
We model the temporal evolution of the amplitude prior variance with a Gamma \gls{pdf} $f(\PFgamma{s}{n}{1}|\PFgamma{s}{n-1}{J})\!=\!\mathcal{G}\big(\PFgamma{s}{n}{1};\mathrm{c}_\gamma,\nicefrac{\PFgamma{s}{n-1}{J}}{\mathrm{c}_\gamma}\big)$, 
parameterized by a constant $\mathrm{c}_\gamma\!=\!\num{100}$.
For the time evolution of the amplitude prior mean we choose a complex Gaussian \gls{pdf} $f(\PFmu{s}{n}{1}|\PFmu{s}{n-1}{J})\!=\!\mathcal{CN}(\PFmu{s}{n}{1};\PFmu{s}{n-1}{J},\sigma_\mu^{2})$. 
The temporal evolution of \glspl{sfv} is described by a Gaussian state-transition \gls{pdf} $f(\psfv{s}{n}{1}|\psfv{s}{n-1}{J})\!=\!\mathcal{N}(\psfv{s}{n}{1};\psfv{s}{n-1}{J},\eye{3}\sigma_{\text{\tiny sfv}}^2)$. 
We choose survival probabilities 
$\psurvival\!=\!\num{0.8}$ and
$\psurvivalPR\!=\!\num{0.9}$, and
a revival probability $\precoveryPR\!=\!\num{0.1}$.
Declaration and pruning thresholds are set to
$T_{\text{\tiny dec}}\!=\!0.5$ and
$T_{\text{\tiny pru}}\!=\!0.1$, respectively.

\paragraph*{Births of new PFs} For new \glspl{pf}, we choose
$f(\PFgamma{s}{n}{j})\!=\!\mathcal{U}(\PFgamma{s}{n}{j};0, \gamma_{\text{\tiny max}})$
and 
$f(\PFmu{s}{n}{j})\!=\!\mathcal{U}(\PFmu{s}{n}{j};\mathcal{R}_\mu)$.
For simplicity, we choose $Q\!=\!1$, and 
$f(\psfv{s}{n}{1})\!=\!\mathcal{U}(\psfv{s}{n}{1};2\mathbf{p}_{\text{\tiny min}},2\mathbf{p}_{\text{\tiny max}})$. Note that $Q\!=\!1$ leads to $N$ births and $N\!-\!\Nfeatures{N}$ deaths along the track.
The Poisson mean is $\mu_{\scriptscriptstyle \mathrm{B}}\!=\!\num{0.5}$.
We choose $\gamma_{\text{\tiny max}}\!=\!5$ and $
\mathcal{R}_\mu\!\coloneqq\!\{\mu\!\in\!\complexsetone{}: |\mu|\!\le\! \mu_{\text{\tiny max}}\}$,
with $\mu_{\text{\tiny max}}\!=\!0.001$ chosen intentionally small to promote robust noncoherent convergence after initialization and coherent tracking thereafter along the track.
New \glspl{ppr} are introduced with \gls{ppr} birth probability $p_{\scriptscriptstyle \mathrm{B}}^{\scriptscriptstyle\mathrm{PR}}\!=\!\num{0.3}$.

\paragraph*{System Parameters}
We choose $\Nfrequency\!=\!10$ frequency bins equally spaced over a bandwidth $B$ and centered at $\fc \!=\!\SI{3.5}{\giga\hertz}$.
We use $J\!=\!4$ \glspl{ura} each with $(\Nantennasy\!\times\!\Nantennasz)\!=\!(4\!\times\!4)$ antennas equally spaced at $\nicefrac{\lambda}{2}$.
This leads to an observation length $\Nz\!=\!160$.
The environment in Fig.\,\ref{fig:scenario} consists of $\Ncomponents\!=\!4$ surfaces \wallrefA\,--\,\wallrefD\,, which---assuming correct model order estimation---leads to $\Nfeaturest{n}\!=\!5$.
All experiments are conducted on an NVIDIA RTX PRO\textsuperscript{\texttrademark} 4000 Blackwell \gls{gpu} with a peak memory bandwidth of \SI{672}{\giga\byte\per\second}.
We choose a time-constant noise variance $\etan{n}$ s.t. $\mathrm{SNR}\!=\!\frac{P^{\scriptscriptstyle\mathrm{ch}}_{\scriptscriptstyle 1}}{\etan{n}}$ is \SI{20}{\dB} with sum-channel power $P^{\scriptscriptstyle\mathrm{ch}}_{\scriptscriptstyle n}\!\coloneqq\!
\frac{1}{\Nz J}
\sum_{j=1}^{J}
\lVert
\sum_{k=0}^{\Ncomponents}
\varrho_{\scriptscriptstyle k,n}^{\scriptscriptstyle(j)}\steerVecx{k}{n}{j}
\rVert^2$
evaluated at time $n\!=\!1$.
In two experiments with trajectories of $N\!=\!200$ steps we conducted \gls{mc} analyses of \num{1000} estimation runs of \gls{nzm}, \gls{nzm}-\gls{nc}, \gls{zm}, and \gls{ref}.

\begin{figure}
        \vskip 0pt	
        \begin{center}
            \newcommand{\LWbound}{1.0pt}
            \newcommand{\LWestimates}{0.5pt}
            \def\datapath{.}
            \setlength{\figurewidth}{0.8\linewidth}
            \setlength{\figureheight}{0.3\linewidth}
            \tikzexternaldisable    
	           \input{figures/results/PEB}
            \tikzexternalenable
        \end{center}
        \vspace{-6mm}
        \captionof{figure}{Position \acrshort{rmse}s 
        a) vs. \glspl{peb} in Experiment 1 and 
        c) under \gls{olos} in Experiment 2 (left). 
        Cumulative frequencies of the position errors (right).
        }
        \label{fig:PEB}
\end{figure}

\begin{figure}
        \vskip 0pt	
        \begin{center}
            \def\datapath{.}
            \setlength{\figurewidth}{0.95\linewidth}
            \setlength{\figureheight}{0.22\linewidth}
            \tikzexternaldisable    
	           \pgfplotsset{every axis/.append style={
  label style={font=\footnotesize},
  legend style={font=\scriptsize},
  tick label style={font=\footnotesize},
}}

\def\datapath{./figures/data}   
\colorlet{compcolor1}{orange}
\colorlet{compcolor2}{blue}
\colorlet{compcolor3}{red}
\colorlet{compcolor4}{green}
\colorlet{compcolor5}{IEEElightblue}

\pgfdeclarelayer{background}
\pgfdeclarelayer{main}
\pgfdeclarelayer{foreground}
\pgfsetlayers{background,main,foreground}

\begin{tikzpicture}

\begin{axis}[%
width=0.7\figurewidth,
height=1\figureheight,
at={(0\figurewidth,0.5\figureheight)},
name=boundary,  
axis line style = thick,	
scale only axis,
xmin=1,
xmax=199.9,
xlabel={Step $n$},                     
ymin=0,
ymax=19.9,
ytick = {0,5,...,20},
ylabel={Input SNR in dB},
tick align=inside,
grid=both, 
    major grid style={very thin, lightgray, draw opacity=0.5},     minor grid style={ultra thin, lightgray, draw opacity=0.2}, 
    minor x tick num = 4, 
    minor y tick num = 4
]
\pgfplotsset{every major tick/.append style={-Triangle Cap,thick},thick}

\pgfplotstableread[col sep=space, header=false]{\datapath/comp-SNR/node-g-s1.dat}\nodegtable
\pgfplotstablegetelem{0}{[index]0}\of\nodegtable
\edef\nodegx{\pgfplotsretval}
\pgfplotstablegetelem{0}{[index]1}\of\nodegtable
\edef\nodegy{\pgfplotsretval}
\coordinate (nodeg) at (axis cs:\nodegx,\nodegy);

\pgfplotstableread[col sep=space, header=false]{\datapath/comp-SNR/node-m-s1.dat}\nodemtable
\pgfplotstablegetelem{0}{[index]0}\of\nodemtable
\edef\nodemx{\pgfplotsretval}
\pgfplotstablegetelem{0}{[index]1}\of\nodemtable
\edef\nodemy{\pgfplotsretval}
\coordinate (nodem) at (axis cs:\nodemx,\nodemy);


\path[  draw = black,
        color        =black,
        line width   =0.15mm,
        line join    =round, 
        line cap     =round] 
        (nodeg) -- ($(nodeg)+(1.5mm,1.5mm)$) -- ($(nodeg)+(5mm,1.5mm)$);  
            \node[right,align=left] at ($(nodeg)+(5mm,1.5mm)$) {\footnotesize contribution from $\PFgamma{\,0}{n}{j}$};
\path[  draw = black,
        color        =black,
        line width   =0.15mm,
        line join    =round, 
        line cap     =round] 
        (nodem) -- ($(nodem)+(1.5mm,-1.5mm)$) -- ($(nodem)+(5mm,-1.5mm)$);  
            \node[right,align=left] at ($(nodem)+(5mm,-1.5mm)$) {\footnotesize contribution from $\PFmu{\,0}{n}{j}$};

\path[  draw = black,
        line width   =0.15mm,
        line join    =round, 
        line cap     =round] 
        (axis cs:160,16.8) -- (axis cs:160,13);  
    \node[align=center,%
        rectangle,inner sep=0.5pt,rounded corners=0.65mm,minimum size=0.7em,
              fill=IEEEyellow,
              text=white] (LoS) at (axis cs:160,12) {\adjustbox{max width=0.65em}{$\scriptstyle 0$}};
    \node[right,align=left,text=black,xshift=1mm] at (LoS) {\footnotesize LoS};

\begin{pgfonlayer}{foreground}
    \coordinate (annotationCenter) at (axis cs:70,4.15);
    
    
    \path[  draw = black,
            line width   =0.15mm,
            line join    =round, 
            line cap     =round] 
            ($(annotationCenter)+(-0.4mm,1.4mm)$) -- 
            ($(annotationCenter)+(0mm,1.4mm)$) -- 
            ($(annotationCenter)+(0mm,-1.4mm)$) -- 
            ($(annotationCenter)+(-0.4mm,-1.4mm)$);  
    
    \path[  draw = black,
            line width   =0.15mm,
            line join    =round, 
            line cap     =round] 
            (annotationCenter) -- 
            ($(annotationCenter)+(2mm,0mm)$) -- 
            ($(annotationCenter)+(4mm,-2mm)$) -- 
            ($(annotationCenter)+(6mm,-2mm)$);  
     \node[right,align=left,text=black,xshift=0mm] 
     (MPCs) at ($(annotationCenter)+(6mm,-2mm)$) 
     {\footnotesize MPCs \wallrefA\,--\,\wallrefD};
     
\end{pgfonlayer}

\pgfplotsinvokeforeach{1,...,5}{%
    \addplot[
        color=compcolor#1, 
        draw opacity = 0.3,
        line width   =2.0pt,
        line join    =round, 
	        line cap     =round
    ] table {\datapath/comp-SNR/SNR-true-s#1.dat};
    \ifnum#1=1\relax
        \label{pgf:SNR-true}
    \fi
    \addplot[
        color=compcolor#1, 
        draw opacity = 0.8,
        line width   =0.7pt,
        line join    =round, 
        line cap     =round
    ] table {\datapath/comp-SNR/SNR-hat-s#1.dat};
    \ifnum#1=1\relax
	    \label{pgf:SNR-hat}
	\fi
    \addplot[color=compcolor#1!80!black, 
        draw opacity = 0.9,
        dash pattern =on 2pt off 1pt on 2pt off 1pt,
        line width   =0.4pt,
        line join    =round, 
        line cap     =round]
    table {\datapath/comp-SNR/SNR-hat-m-s#1.dat};
    \ifnum#1=1\relax
	    \label{pgf:SNR-hat-m}
	\fi
    \addplot[color=compcolor#1!80!black, 
            draw opacity = 0.9,
            dash pattern =on 0.01pt off 0.7pt,
            line width   =0.4pt,
            line join    =round, 
            line cap     =round] 
    table {\datapath/comp-SNR/SNR-hat-g-s#1.dat};
    \ifnum#1=1\relax
	    \label{pgf:SNR-hat-g}
	\fi
}

\node[  below left, 
        align=right,
        font={\footnotesize},
        fill=white,
        opacity=0.9,
        inner sep=0.75mm, 
        xshift=0.65mm, 
        yshift=0.65mm,
        draw, 
        line width = 0.8pt] at %
(axis cs:200,20){$j\!=\!4$};

\end{axis}

\newcommand{\compSNRtrueMultiLine}[1][0.6cm]{%
    \tikz[baseline=-0.5ex, x=#1]{%
        \foreach \i [
            evaluate=\i as \xstart using {(\i-1)/5},
            evaluate=\i as \xend using {\i/5}
        ] in {1,...,5}{%
            \draw[
                color=compcolor\i,
                draw opacity=0.3,
                line width=2.0pt,
                line cap=butt
            ] (\xstart,0) -- (\xend,0);
        }%
    }%
}

\newcommand{\compSNRestMultiLine}[1][0.6cm]{%
    \tikz[baseline=-0.5ex, x=#1]{%
        \foreach \i [
            evaluate=\i as \xstart using {(\i-1)/5},
            evaluate=\i as \xend using {\i/5}
        ] in {1,...,5}{%
            \draw[
                color=compcolor\i,
                draw opacity = 0.8,
                line width   =0.7pt,
                line cap=butt
            ] (\xstart,0) -- (\xend,0);
        }%
    }%
}

\newcommand{\compSNRestmMultiLine}{%
    \tikz[baseline=-0.5ex, x=0.6cm, y=0.2cm]{%
        \path[use as bounding box] (0,-0.5) rectangle (1,0.5);
        \foreach \i [
            evaluate=\i as \y using {(\i-3)/4}
        ] in {1,...,5}{%
            \draw[
                color=compcolor\i!80!black,
                draw opacity=0.9,
                dash pattern=on 2pt off 1pt on 2pt off 1pt,
                line width=0.4pt,
                line cap=round
            ] (0,\y) -- (1,\y);
        }%
    }%
}

\newcommand{\compSNRestgMultiLine}{%
    \tikz[baseline=-0.5ex, x=0.6cm, y=0.2cm]{%
        \path[use as bounding box] (0,-0.5) rectangle (1,0.5);
        \foreach \i [
            evaluate=\i as \y using {(\i-3)/4}
        ] in {1,...,5}{%
            \draw[
                color=compcolor\i!80!black,
                draw opacity = 0.9,
                dash pattern =on 0.01pt off 0.7pt,
                line width   =0.4pt,
                line cap=round
            ] (0,\y) -- (1,\y);
        }%
    }%
}

\def\legendwidth{0.24\figurewidth}
\node[align=center,draw,thick,fill=white,inner sep=0pt,right=11mm,line cap = round, fill opacity=0.8, text opacity = 1, draw opacity = 1 , xshift = 0mm, minimum width=\legendwidth, text width=0.001pt, fit = (boundary.north east) (boundary.south east)] at (boundary.east) 
{~

\footnotesize
    \setlength{\tabcolsep}{1pt} 
    \makebox[\linewidth][c]{%
        \,\,\begin{tabular}{rl}
            Line~  &  SNR \\ \hline
            \compSNRtrueMultiLine  & true  \\
            \compSNRestMultiLine   & estimates  \\
            \compSNRestmMultiLine & coherent \\
            \compSNRestgMultiLine & noncoh. 
        \end{tabular}
    }
};

\end{tikzpicture}%
            \tikzexternalenable
        \end{center}
        \vspace{-6mm}%
        \captionof{figure}{Estimated vs. true component \glspl{snr} for the \gls{los} and \glspl{mpc}.
        }
        \label{fig:comp-SNR}
\end{figure}

\subsection{Experiment 1: Synthetic Data, Random Trajectory}

Random measurements are generated according to \eqref{eq:observation-generative} and realizations of \textit{random} states $\{\RVstate{n},\RVpsfv{k}{n}{}\}_{n=0}^{N}$ are 
generated according to the specified state-transition \glspl{pdf} (see Supplementary Material, Sec.\,S-VII for details), yielding the \gls{ncv} trajectory realizations shown in Fig.\,\ref{fig:scenario}\,a).
The \glspl{rmse} of \gls{zm} approach both the \textit{noncoherent} \gls{peb} in Fig.\,\ref{fig:PEB}\,a) and \glspl{meb} in Fig.\,\ref{fig:MEB}.
In this scenario without \gls{los} obstructions, \gls{ref} performs almost identically to \gls{zm} in position \gls{rmse}.
Fig.\,\ref{fig:PEB}\,a) and the cumulative frequency of position errors in Fig.\,\ref{fig:PEB}\,b) show that \gls{nzm} outperforms \gls{zm} and \gls{ref} in terms of \gls{mt} position accuracy.
The mapping \gls{rmse} of \gls{nzm} approaches the \textit{coherent} \gls{meb} in Fig.\,\ref{fig:MEB}, which leverages the full aperture of the \gls{dmimo} infrastructure.
This highlights the aperture-retaining properties of \gls{nzm}, which outperforms \gls{zm} and \gls{ref} that are tied to 
the noncoherent \gls{pcrlb} and fail to retain aperture.
On noncoherent data, it gracefully degrades to noncoherent operation (see \gls{nzm}-\gls{nc}).
In all figures, $\interval{50\%}$ denotes the 
interquartile range
of the empirical error distribution.
Fig.\,\ref{fig:comp-SNR} shows the true component input
$\mathrm{SNR}_{\scriptscriptstyle k,n}^{\scriptscriptstyle(j)}\!\coloneqq\!
\frac{\lVert
\varrho_{\scriptscriptstyle k,n}^{\scriptscriptstyle(j)}
\steerVecx{k}{n}{j}
\rVert^2}{\Nz \etan{n}^{\scriptscriptstyle(j)}}$
for components $k$ evaluated at \gls{pa} $j\!=\!4$ for one \gls{mc} run.
It is compared against the \gls{nzm} estimates $\widehat{\mathrm{SNR}}_{\scriptscriptstyle k,n}^{\scriptscriptstyle(j)}\!\coloneqq\!
\frac{
(|\widehat{\mu}_{\!\scriptscriptstyle k,n}|^2
+
\widehat{\gamma}_{\!\scriptscriptstyle k,n})
\lVert
 \steerVec{j}(\stateHat{n}^{\text{\tiny MMSE}},\PFphiHat{s}{n}{j}^{\text{\tiny MMSE}})
\rVert^2}
{\Nz \etanHat{n}^{\,\text{\tiny MMSE}\ist \scriptscriptstyle(j)}}$, 
with 
$\widehat{\mu}_{\!\scriptscriptstyle k,n}\!\coloneqq\!\big[\PFphiHat{s}{n}{j}^{\text{\tiny MMSE}}\big]_{\scriptscriptstyle 5}$
and $\widehat{\gamma}_{\!\scriptscriptstyle k,n}\!\coloneqq\!\big[\PFphiHat{s}{n}{j}^{\text{\tiny MMSE}}\big]_{\scriptscriptstyle 4}$,
which are sums of two terms.
One term models the coherent signal power contained in the means $\PFmu{k}{n}{j}$ and the other the noncoherent signal power contained in the variances $\PFgamma{k}{n}{j}$.
The latter dominates during convergence, whereas the former takes over as the estimator builds up coherence.
In this experiment, we used $P\!=\!\num{30000}$, ${\sigma}_{\scriptscriptstyle \mathrm{v}}\!=\!\SI{0.5}{\metre\per\second\squared}$,
$B\!=\!\SI{200}{\mega\hertz}$, 
$\sigma_\mu\!=\!\num{3e-3}$, and
$\sigma_{\text{\tiny sfv}}\!=\!\SI{0.1}{\milli\metre}$. %
%

\subsection{Experiment 2: Synthetic Data, Deterministic Trajectory}

The \gls{mt} state follows the \textit{deterministic} trajectory in Fig.\,\ref{fig:scenario}\,b) for which the \gls{pcrlb} is not a lower bound.
Both the position \gls{rmse} in Fig.\,\ref{fig:PEB}\,c) and the cumulative frequencies of position errors in Fig.\,\ref{fig:PEB}\,d) show that \gls{nzm} performs more accurately than \gls{zm} and \gls{ref}.
The surface in the center of Fig.\,\ref{fig:scenario}\,b) causes \gls{pa}-local path obstructions.
Fig.\,\ref{fig:vis} shows the true number of visible paths per \gls{pa} $j$ vs. the sum of their estimated \gls{ppr} existence probabilities $\sum_{s\in\setFeaturest{n}}\existenceProbPA{s}{n}{j}$ for \gls{nzm}.
Algorithms \gls{nzm} and \gls{zm} using infrastructure-global \gls{sfv} point-map features---capable of \gls{pf} data fusion across \glspl{pa}---share at least one \gls{pr} with each surface at all times, which keeps their \glspl{pf} alive even under partial visibility.
\Gls{ref}, on the other hand, suffers from deaths of \gls{pa}-local \gls{va} point-map features under obstructions, leading to robustness disadvantages.
In \gls{los} conditions, \gls{ref} again performs close to \gls{zm}. 
The position and existence estimates show that robust tracking is possible with our \gls{bp} method even under partial \gls{olos} in a challenging multipath scenario.
In this experiment, we used a more economical number of particles $P\!=\!\num{15000}$, $B\!=\!\SI{100}{\mega\hertz}$, ${\sigma}_{\scriptscriptstyle \mathrm{v}}\!=\!\SI{7.7}{\metre\per\second\squared}$ chosen as \num{1.5} times the maximum \gls{mt} acceleration,
$\sigma_\mu\!=\!\num{3e-2}$, and
$\sigma_{\text{\tiny sfv}}\!=\!\SI{4}{\milli\metre}$.

\textit{Discussion:} 
Fig.\,\ref{fig:PEB}\,a) and Fig.\,\ref{fig:MEB} demonstrate that distributed processing using a zero-mean Type-II likelihood (i.e., \gls{zm} and \gls{ref}) inevitably results in the loss of aperture and hence reduced estimation performance. 
Using a nonzero-mean Type-II likelihood (i.e., \gls{nzm}) can restore coherence and hence retain aperture and performance in \gls{dmimo} systems. 
%
\ifthenelse{\equal{\IEEEversion}{true}}%
{%
}%
{%
Note that the proposed system model assumes frequency synchronization (but not phase calibration) between the \gls{mt} clock and the \gls{dmimo} infrastructure. 
While this is a strong assumption, it can be achieved in practice (cf.\,\cite{Deutschmann24SPAWC}). 
If it is violated, the estimator can still operate \textit{noncoherently}, thereby accommodating phase asynchrony, provided that the amplitude prior means $\PFmu{s}{n}{j}$ are estimated to zero. 
In this case, \gls{nzm} effectively reduces to the \gls{zm} model, as confirmed by experiments.

An additional advantage of the proposed \gls{bp} framework is its flexibility: the factor graph can be readily extended to include unknown timing and phase offsets of the \glspl{pa} as latent parameters (similar to noise variances) \cite{Fascista25RadioStripes}. Since these parameters are shared across \glspl{sfv}, they can be jointly inferred provided that the number of measurement equations induced by common features exceeds the number of unknowns \cite{Fascista25RadioStripes,VenLeiWit:Arxiv2026}. This concept potentially also enables robust over-the-air calibration \cite{KolBjoGoeLar:PIMRC2025}.
}%

\subsection{Run Time Analysis}
Although the arithmetic operations scale according to the computational complexity stated in Section~\ref{sec:GPU-acceleration}, a \gls{gpu}-based implementation of our \gls{bp} method is memory-bound rather than math-bound. 
Hence, its execution time scales approximately linearly with the memory traffic and inversely with the \gls{gpu} memory bandwidth.
At the above system parameters $\{P,J,\Nz,{\widetilde{\Ncomponents}}^2\}$, with respect to which the memory traffic practically scales linearly, the \gls{gpu}-accelerated \textsc{Matlab} implementation of our \gls{bp} method has average runtimes per time step $n$ of \SI{240}{\milli\second} and \SI{420}{\milli\second}, in Experiments 2 and 1, respectively, demonstrating possible real-time capability.
When scaled to the same system parameters and compared with the runtime of the original algorithm in~\cite[Table\,1]{LiaLeiMey:TSP2025-SuppDoc}, our algorithm runs \num{13} times faster.%

\begin{figure}
        \vskip 0pt	
        \begin{center}
            \newcommand{\LWbound}{1.0pt}
            \newcommand{\LWestimates}{0.5pt}
            \def\datapath{.}
            \setlength{\figurewidth}{1\linewidth}
            \setlength{\figureheight}{0.4\linewidth}
            \tikzexternaldisable    
	           \pgfplotsset{every axis/.append style={
  label style={font=\footnotesize},
  legend style={font=\scriptsize},
  tick label style={font=\footnotesize},
}}

\def\datapath{./figures/data}   

\begin{tikzpicture}

\begin{axis}[%
width=0.465\figurewidth,
height=0.485\figureheight,
at={(0\figurewidth,0.5\figureheight)},
name=boundaryLeft,  
axis line style = thick,	
scale only axis,
xmin=1,
xmax=200,
xticklabel=\empty,                      
ymin=0,
ymax=5.5,
ytick = {0,1,...,5.5},
ylabel={Num. paths},
tick align=inside,
grid=both, 
    major grid style={very thin, lightgray, draw opacity=0.5},     minor grid style={ultra thin, lightgray, draw opacity=0.2}, 
    minor x tick num = 4, 
    minor y tick num = 1, 
]
\pgfplotsset{every major tick/.append style={-Triangle Cap,thick},thick}

\addplot [color=white!70!black, line width=2.0pt,line join=round, line cap=round]
  table{\datapath/visibilities/vis-coh-det-j1.dat};%
    \label{pgf:vis-coh}

\addplot [color=IEEEgreen, line width=0.7pt,line join=round, line cap=round]
table{\datapath/visibilities/sum-exist-prob-j1.dat};%
    \label{pgf:sum-exist-prob}

\node[above left, align=right,font={\footnotesize},fill=white,
opacity=0.9,inner sep=0.75mm, xshift=0.65mm, yshift=-0.7mm, draw, line width = 0.8pt] at %
(axis cs:200,0){a) $j\!=\!1$};

\end{axis}

\begin{axis}[%
width=0.465\figurewidth,
height=0.485\figureheight,
at={(0.475\figurewidth,0.5\figureheight)},
name=boundaryRight,  
axis line style = thick,	
scale only axis,
xmin=1,
xmax=199,
xticklabel=\empty,                      
ymin=0,
ymax=5.5,
ytick = {0,1,...,5.5},
yticklabel=\empty,                      
tick align=inside,
grid=both, 
    major grid style={very thin, lightgray, draw opacity=0.5},     minor grid style={ultra thin, lightgray, draw opacity=0.2}, 
    minor x tick num = 4, 
    minor y tick num = 1, 
]
\pgfplotsset{every major tick/.append style={-Triangle Cap,thick},thick}

\addplot [color=white!70!black, line width=2.0pt,line join=round, line cap=round]
  table{\datapath/visibilities/vis-coh-det-j2.dat};%

\addplot [color=IEEEgreen, line width=0.7pt,line join=round, line cap=round]
table{\datapath/visibilities/sum-exist-prob-j2.dat};%
    
\node[above left, align=right,font={\footnotesize},fill=white,
opacity=0.9,inner sep=0.75mm, xshift=0.5mm, yshift=-0.7mm, draw, line width = 0.8pt] at %
(axis cs:200,0){b) $j\!=\!2$};

\end{axis}

\begin{axis}[%
width=0.465\figurewidth,
height=0.485\figureheight,
at={(0\figurewidth,0\figureheight)},
axis line style = thick,	
scale only axis,
xmin=1,
xmax=200,
xlabel={Step $n$},
ymin=0,
ymax=5.5,
ytick = {0,1,...,5.5},
yticklabels={0,1,...,4},
ylabel={Num. paths},
tick align=inside,
grid=both, 
    major grid style={very thin, lightgray, draw opacity=0.5},     minor grid style={ultra thin, lightgray, draw opacity=0.2}, 
    minor x tick num = 4, 
    minor y tick num = 1, 
]
\pgfplotsset{every major tick/.append style={-Triangle Cap,thick},thick}

\addplot [color=white!70!black, line width=2.0pt,line join=round, line cap=round]
  table{\datapath/visibilities/vis-coh-det-j3.dat};%

\addplot [color=IEEEgreen, line width=0.7pt,line join=round, line cap=round]
table{\datapath/visibilities/sum-exist-prob-j3.dat};%

\node[above left, align=right,font={\footnotesize},fill=white,
opacity=0.9,inner sep=0.75mm, xshift=0.65mm, yshift=-0.7mm, draw, line width = 0.8pt] at %
(axis cs:200,0){c) $j\!=\!3$};
\end{axis}

\begin{axis}[%
width=0.465\figurewidth,
height=0.485\figureheight,
at={(0.475\figurewidth,0\figureheight)},
axis line style = thick,	
scale only axis,
xmin=1,
xmax=199,
xlabel={Step $n$},
ymin=0,
ymax=5.5,
ytick = {0,1,...,5.5},
yticklabels={0,1,...,4},
yticklabel=\empty,                      
tick align=inside,
grid=both, 
    major grid style={very thin, lightgray, draw opacity=0.5},     minor grid style={ultra thin, lightgray, draw opacity=0.2}, 
    minor x tick num = 4, 
    minor y tick num = 1, 
]
\pgfplotsset{every major tick/.append style={-Triangle Cap,thick},thick}

\addplot [color=white!70!black, line width=2.0pt,line join=round, line cap=round]
  table{\datapath/visibilities/vis-coh-det-j4.dat};%

\addplot [color=IEEEgreen, line width=0.7pt,line join=round, line cap=round]
table{\datapath/visibilities/sum-exist-prob-j4.dat};%

\node[above left, align=right,font={\footnotesize},fill=white,
opacity=0.9,inner sep=0.75mm, xshift=0.5mm, yshift=-0.7mm, draw, line width = 0.8pt] at %
(axis cs:200,0){d) $j\!=\!4$};

\end{axis}

\node[below,draw,thick,fill=white,inner sep=0pt,above right=0.0em,line cap = round, fill opacity=0.8, text opacity = 1, draw opacity = 1 , yshift = 2.5mm, minimum height=0.4cm, fit = (boundaryLeft.north west) (boundaryRight.north east)] 
{~\vspace{-3mm}

\footnotesize
    \setlength{\tabcolsep}{3pt} 
    \begin{tabular}{{@{}R{1.0cm}L{2.5cm}|R{1.0cm}L{2.5cm}@{}}} 
    \ref{pgf:vis-coh} & true num. vis. paths ~&
    \ref{pgf:sum-exist-prob} & existence prob. sum \,
    \end{tabular}};
    
\end{tikzpicture}%
            \tikzexternalenable
        \end{center}
        \vspace{-6mm}
        \captionof{figure}{True number of visible paths in Experiment 2 vs. the sum of the estimated existence probabilities at each \gls{pa} $j\in\setAnchors$.
        }
        \label{fig:vis}
\end{figure}

\section{Conclusion}\label{sec:conclusion}  
In this work, we proposed a scalable direct \gls{mpslam} method based on particle-based \gls{bp} for coherent sensing and localization in \gls{dmimo}/\gls{xlmimo} systems. Motivated by the increasing availability of spatially distributed arrays in future wireless infrastructures, the method enables coherent data fusion directly at the raw-signal level while jointly estimating the \gls{mt} state and geometric map of the environment.
We show that estimators commonly using a \textit{zero-mean} (\gls{zm}) Type-II likelihood inherently lead to noncoherent processing and aperture loss, with estimation performance fundamentally tied to the \textit{noncoherent} \gls{pcrlb}.
The key element for coherent processing is a phase-preserving \emph{nonzero-mean} (\gls{nzm}) Type-II likelihood that maintains a shared complex mean across distributed \glspl{pa}.
Such an estimator 
can preserve the global phase structure and retain aperture gain, allowing it to perform close to the \textit{coherent} \gls{pcrlb}.
Meanwhile, the variance term captures residual noncoherent signal power, allowing the model to naturally accommodate 
partial coherence by shifting unmodeled power into the variance.
\ifthenelse{\equal{\IEEEversion}{true}}%
{%
}%
{%
This construction retains full aperture gain whenever coherence is available, while gracefully degrading to noncoherent operation otherwise.
}%
Combined with the \gls{sfv} surface model, the proposed method enables consistent fusion of map features across distributed arrays and propagation paths, while accounting for near-field propagation and visibility effects. 
Aided by its increased spatial resolution, 
infrastructure-global \gls{sfv} point-map features, and a \gls{gpu} implementation, our proposed \gls{nzm} method achieves higher accuracy, improved robustness, and an order of magnitude runtime speedup over the reference method (\gls{ref}).
Promising directions for future research include extending the proposed framework toward hybrid inference--deep-learning architectures \cite{ WeiLiaMey:TSP2026, VenDeuFucKnoLei:Arxiv2026}, and incorporating multiple-bounce paths \cite{Li25adaptiveDMIMOslam}, surface roughness \cite{WieVenWilLei:JAIF2023,WieVenWilWitLei:Fusion2024}, multiple feature types such as point scatterers \cite{KimGranSveKimWym:TVT2022, LiCaiLeiTuf:ICC2024}, or timing and phase synchronization parameters for \glspl{pa} \cite{Fascista25RadioStripes,KolBjoGoeLar:PIMRC2025}.


\bibliographystyle{IEEEtran}
\balance
\renewcommand{\baselinestretch}{0.98}
\bibliography{IEEEabrv,bibliography,ThisPaper}

\ifthenelse{\equal{\IEEEversion}{true}}
{
}%
{%
    \includepdf[pages=-]{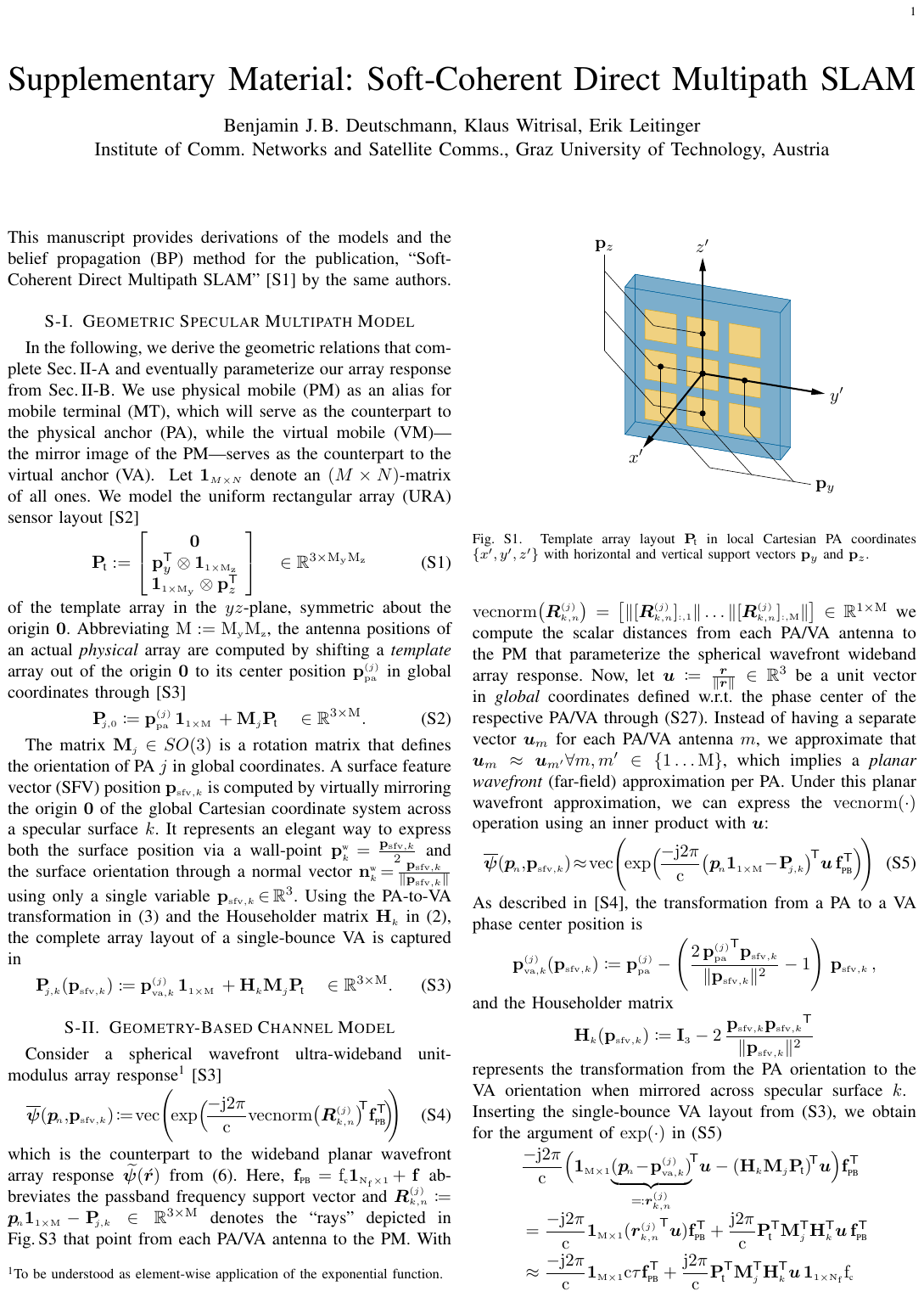}
    \includepdf[pages=-]{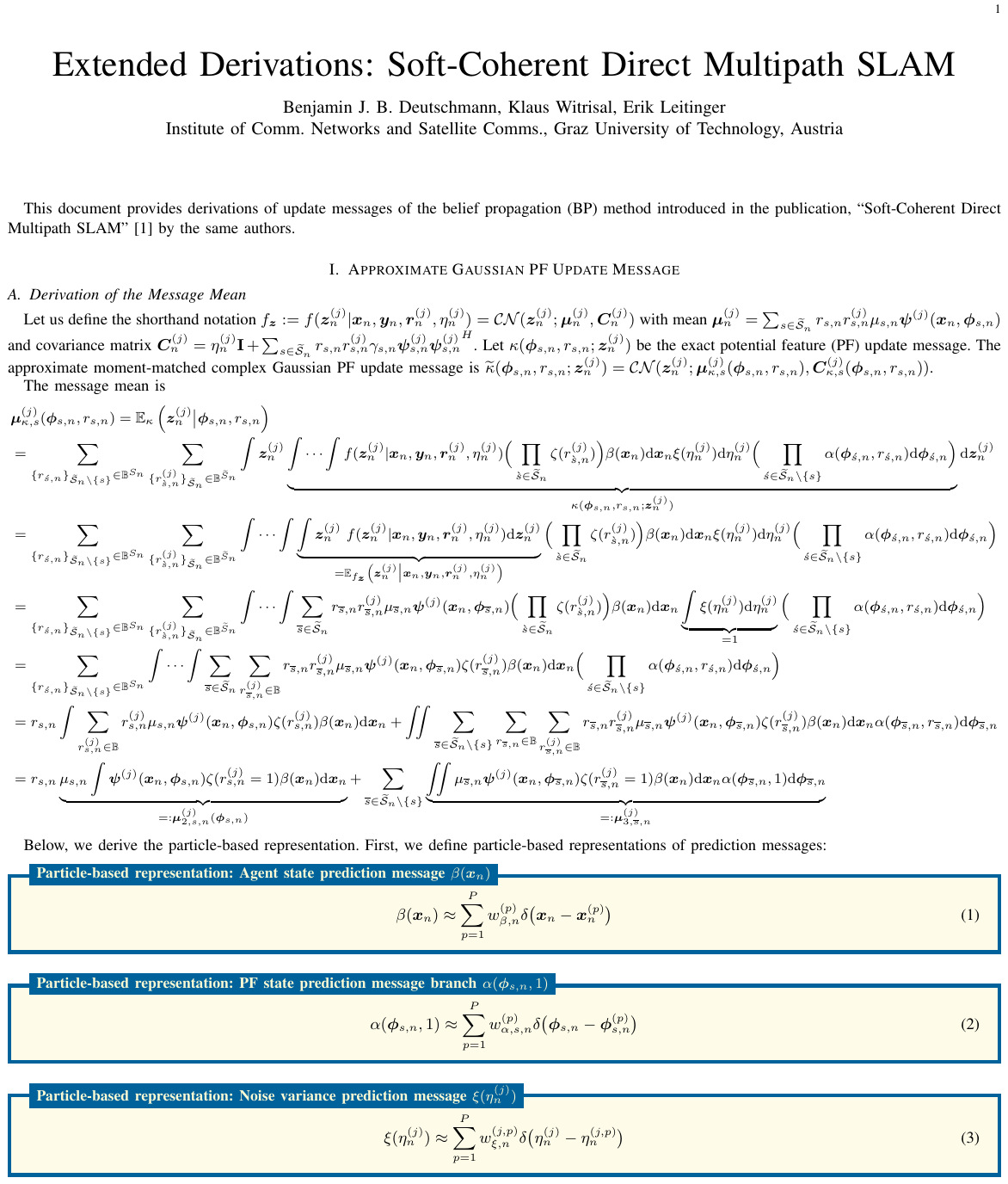}
}

\end{document}


\title{\resizebox{\textwidth}{!}{Supplementary Material: \PaperTitle}}

\author{\IEEEauthorblockN{Benjamin J.\,B. Deutschmann, Klaus Witrisal, Erik Leitinger}

\IEEEauthorblockA{Institute of Comm. Networks and Satellite Comms., Graz University of Technology, Austria}}

\maketitle

\noindent%
This manuscript provides derivations of the models and the \gls{bp} method for the publication, ``\PaperTitle''~\cite{ThisPaper} by the same authors.
\section{Geometric Specular Multipath Model}\label{Ssec:geometric-SMC-model}
\ifthenelse{\equal{\IEEEversion}{true}}%
{%
}{
    In the following, we derive the geometric relations that complete Sec.\,\ref{sec:geometry} and eventually parameterize our array response from Sec.\,\ref{sec:signal-model}.
    We use \gls{pm} as an alias for \gls{mt}, which will serve as the counterpart to the \gls{pa}, while the \gls{vm}---the mirror image of the \gls{pm}---serves as the counterpart to the \gls{va}.
}
Let $\onematrix{M}{N}$ denote an $(M\times N)$-matrix of all ones. 
We model the \gls{ura} sensor layout~\cite{DeutschmannCISA2025} 
\begin{align}\label{eq:Ptemplate}
    \Ptemplate := 
    \begin{bmatrixs}
        \bm{0} \\
        \mathbf{p}_y^\trp \otimes \onematrix{1}{\Nantennasz} \\
        \onematrix{1}{\Nantennasy} \otimes \mathbf{p}_z^\trp \\
    \end{bmatrixs}
    \quad \in \realset{3}{\Nantennasy\Nantennasz}
\end{align}
of the template array in the $yz$-plane, symmetric about the origin $\origin$.
Abbreviating $\Nantennas:=\Nantennasy\Nantennasz$, the antenna positions of an actual \textit{physical} array are computed by shifting a \textit{template} array out of the origin $\origin$ to its center position $\posPA{j}$ in global coordinates through~\cite{Deutschmann25OJSP}
\begin{align}\label{eq:PA-layout}
    \PlayoutDet{j}{0} \coloneqq  \posPA{j} \, \onematrix{1}{\Nantennas}\, + \rotM{j} \Ptemplate \quad \in \realset{3}{\Nantennas} .
\end{align}

The matrix $\rotM{j} \in SO(3)$ is a rotation matrix that defines the orientation of \gls{pa} $j$ in global coordinates.
A \gls{sfv} position $\psfvDet{k}$ is computed by virtually mirroring the origin $\origin$ of the global Cartesian coordinate system across a specular surface $k$.
It represents an elegant way to express both the surface position via a wall-point 
    $\pwkDet \!=\! \frac{\psfvDet{k}}{2} $ 
and the surface orientation through a normal vector
    $\nwDet \!=\! \frac{\psfvDet{k}}{\lVert \psfvDet{k} \rVert}$ 
using only a single variable $\psfvDet{k} \!\in\! \realsetone{3}$.
Using the \gls{pa}-to-\gls{va} transformation in \eqref{eq:posPA-SB} and the Householder matrix $\houseDet{k}$ in \eqref{eq:house}, 
the complete array layout of a single-bounce \gls{va} is captured in
\begin{align}\label{Seq:VA-layout-SB}
    \PlayoutDet{j}{k}(\psfvDet{k}) \coloneqq  \posVAdet{k}{j} \, \onematrix{1}{\Nantennas}\, + \houseDet{k} \rotM{j} \Ptemplate \quad \in \realset{3}{\Nantennas}.
\end{align}

\ifthenelse{\equal{\IEEEversion}{true}}
{
}
{
\section{Geometry-Based Channel Model}\label{Ssec:geometry-based-channel-model}
Consider a spherical wavefront ultra-wideband unit-modulus array response\footnote{To be understood as element-wise application of the exponential function.}~\cite{Deutschmann25OJSP}
\begin{align}\label{Seq:spherical-wavefront-response}
    \overline{\V{\psi}}(\pos{n},\!\psfvDet{k})
    \!\coloneqq\!
    \vect
    \!\Bigg(\!
        \!\exp
        \!\Big(\!
            \frac{-\mathrm{j}2 \pi}{\lightspeed}\mathrm{vecnorm}
            \big(
                \Rknj{k}{n}{j}
            \big)^{\!\!\trp}
                \fpb^\trp
        \!\Big)    
    \!\!\Bigg)
\end{align}
which is the counterpart to the wideband 
planar wavefront array response $\widetilde{\V{\psi}}(\acute{\V{r}})$ from \eqref{eq:unit-modulus-array-response}.
Here, $\fpb\!=\!\fc\onematrix{\Nfrequency}{1}\!+\mathbf{f}$ abbreviates the passband frequency support vector and 
$\Rknj{k}{n}{j}\coloneqq \pos{n} \onematrix{1}{\Nantennas}\!-\! \PlayoutDet{j}{k}\in\realset{3}{\Nantennas}$ denotes the ``rays'' depicted in Fig.\,\ref{fig:AP5NLoS} that point from each \gls{pa}/\gls{va} antenna to the \gls{pm}.
\ifthenelse{\equal{\IEEEversion}{true}}%
{%
}{%
    \begin{figure}[t]
    \begin{center}
    	\setlength{\figurewidth}{1.0\linewidth}
%
%
\definecolor{mycolor1}{rgb}{0.00000,0.38431,0.60784}%
\definecolor{mycolor2}{rgb}{0.88235,0.63922,0.00000}%
\begin{tikzpicture}

\tikzset{line cap=round,line join=round}

\pgfdeclarelayer{behindBackground}
\pgfdeclarelayer{background}
\pgfdeclarelayer{foreground}
\pgfdeclarelayer{layer1}
\pgfdeclarelayer{layer2}
\pgfdeclarelayer{layer3}
\pgfsetlayers{behindBackground,background,main,layer1,layer2,layer3,foreground}   

\begin{axis}[%
width=0.655\figurewidth,
height=0.75\figurewidth,
at={(0\figurewidth,0\figurewidth)},
scale only axis,
plot box ratio=1 1.148 1.148,
clip=false,
xmin=-0.0285616626666667,
xmax=0.1675,
tick align=outside,
ymin=-0.1,
ymax=0.125,
zmin=-0.1,
zmax=0.125,
view={110}{27},
ticks=none,         
grid=none,          
axis line style={draw=none} 
]
\begin{pgfonlayer}{foreground}
  \addplot3 [color=black,line width=1pt,arrows = {-Stealth[inset=0pt, scale=1.05, angle'=25]}]
   table[row sep=crcr] {%
  0	0	0\\
  0.1675	0	0\\
  };
   \node[align=center, inner sep=0]
  at (axis cs:0.186,0,0) {$x'$};
  \addplot3 [color=black,line width=1pt,arrows = {-Stealth[inset=0pt, scale=1.05, angle'=25]}]
   table[row sep=crcr] {%
  0	0	0\\
  0	0.125	0\\
  };
   \node[align=center, inner sep=0]
  at (axis cs:0,0.137,0) {$y'$};
  \addplot3 [color=black,line width=1pt,arrows = {-Stealth[inset=0pt, scale=1.05, angle'=25]}]
   table[row sep=crcr] {%
  0	0	0\\
  0	0	0.125\\
  };
   \node[align=center, inner sep=0]
  at (axis cs:0,0,0.137) {$z'$};
\end{pgfonlayer}

\addplot3[area legend, line width=0.2pt, draw=mycolor1, fill=mycolor1, fill opacity=0.3, line cap=round,line join=round,forget plot]
table[row sep=crcr] {%
x	y	z\\
-1.85e-05	-0.0792308639	-0.0792308639\\
-1.85e-05	0.0792308639	-0.0792308639\\
-1.85e-05	0.0792308639	0.0792308639\\
-1.85e-05	-0.0792308639	0.0792308639\\
}--cycle;

\addplot3[area legend, line width=0.2pt, draw=mycolor1, fill=mycolor1, fill opacity=0.3, line cap=round,line join=round,forget plot]
table[row sep=crcr] {%
x	y	z\\
-0.0285616626666667	-0.0792308639	-0.0792308639\\
-0.0285616626666667	0.0792308639	-0.0792308639\\
-0.0285616626666667	0.0792308639	0.0792308639\\
-0.0285616626666667	-0.0792308639	0.0792308639\\
}--cycle;

\addplot3[area legend, line width=0.2pt, draw=mycolor1, fill=mycolor1, fill opacity=0.3, line cap=round,line join=round,forget plot]
table[row sep=crcr] {%
x	y	z\\
-0.0285616626666667	-0.0792308639	-0.0792308639\\
-0.0285616626666667	0.0792308639	-0.0792308639\\
-1.85e-05	0.0792308639	-0.0792308639\\
-1.85e-05	-0.0792308639	-0.0792308639\\
}--cycle;

\addplot3[area legend, line width=0.2pt, draw=mycolor1, fill=mycolor1, fill opacity=0.3, line cap=round,line join=round,forget plot]
table[row sep=crcr] {%
x	y	z\\
-0.0285616626666667	-0.0792308639	0.0792308639\\
-0.0285616626666667	0.0792308639	0.0792308639\\
-1.85e-05	0.0792308639	0.0792308639\\
-1.85e-05	-0.0792308639	0.0792308639\\
}--cycle;

\addplot3[area legend, line width=0.2pt, draw=mycolor1, fill=mycolor1, fill opacity=0.3, line cap=round,line join=round,forget plot]
table[row sep=crcr] {%
x	y	z\\
-0.0285616626666667	-0.0792308639	-0.0792308639\\
-0.0285616626666667	-0.0792308639	0.0792308639\\
-1.85e-05	-0.0792308639	0.0792308639\\
-1.85e-05	-0.0792308639	-0.0792308639\\
}--cycle;

\addplot3[area legend, line width=0.2pt, draw=mycolor1, fill=mycolor1, fill opacity=0.3, line cap=round,line join=round,forget plot]
table[row sep=crcr] {%
x	y	z\\
-0.0285616626666667	0.0792308639	-0.0792308639\\
-0.0285616626666667	0.0792308639	0.0792308639\\
-1.85e-05	0.0792308639	0.0792308639\\
-1.85e-05	0.0792308639	-0.0792308639\\
}--cycle;

\addplot3[area legend, line width=0.2pt, draw=mycolor2, fill=white!50!mycolor2, line cap=round,line join=round,forget plot]
table[row sep=crcr] {%
x	y	z\\
0	-0.05888780425	-0.05888780425\\
0	-0.02676718375	-0.05888780425\\
0	-0.02676718375	-0.02676718375\\
0	-0.05888780425	-0.02676718375\\
}--cycle;

\addplot3[area legend, line width=0.2pt, draw=mycolor2, fill=white!50!mycolor2, line cap=round,line join=round,forget plot]
table[row sep=crcr] {%
x	y	z\\
0	-0.05888780425	-0.01606031025\\
0	-0.02676718375	-0.01606031025\\
0	-0.02676718375	0.01606031025\\
0	-0.05888780425	0.01606031025\\
}--cycle;

\addplot3[area legend, line width=0.2pt, draw=mycolor2, fill=white!50!mycolor2, line cap=round,line join=round,forget plot]
table[row sep=crcr] {%
x	y	z\\
0	-0.05888780425	0.02676718375\\
0	-0.02676718375	0.02676718375\\
0	-0.02676718375	0.05888780425\\
0	-0.05888780425	0.05888780425\\
}--cycle;

\addplot3[area legend, line width=0.2pt, draw=mycolor2, fill=white!50!mycolor2, line cap=round,line join=round,forget plot]
table[row sep=crcr] {%
x	y	z\\
0	-0.01606031025	-0.05888780425\\
0	0.01606031025	-0.05888780425\\
0	0.01606031025	-0.02676718375\\
0	-0.01606031025	-0.02676718375\\
}--cycle;

\addplot3[area legend, line width=0.2pt, draw=mycolor2, fill=white!50!mycolor2, line cap=round,line join=round,forget plot]
table[row sep=crcr] {%
x	y	z\\
0	-0.01606031025	-0.01606031025\\
0	0.01606031025	-0.01606031025\\
0	0.01606031025	0.01606031025\\
0	-0.01606031025	0.01606031025\\
}--cycle;

\addplot3[area legend, line width=0.2pt, draw=mycolor2, fill=white!50!mycolor2, line cap=round,line join=round,forget plot]
table[row sep=crcr] {%
x	y	z\\
0	-0.01606031025	0.02676718375\\
0	0.01606031025	0.02676718375\\
0	0.01606031025	0.05888780425\\
0	-0.01606031025	0.05888780425\\
}--cycle;

\addplot3[area legend, line width=0.2pt, draw=mycolor2, fill=white!50!mycolor2, line cap=round,line join=round,forget plot]
table[row sep=crcr] {%
x	y	z\\
0	0.02676718375	-0.05888780425\\
0	0.05888780425	-0.05888780425\\
0	0.05888780425	-0.02676718375\\
0	0.02676718375	-0.02676718375\\
}--cycle;

\addplot3[area legend, line width=0.2pt, draw=mycolor2, fill=white!50!mycolor2, line cap=round,line join=round,forget plot]
table[row sep=crcr] {%
x	y	z\\
0	0.02676718375	-0.01606031025\\
0	0.05888780425	-0.01606031025\\
0	0.05888780425	0.01606031025\\
0	0.02676718375	0.01606031025\\
}--cycle;

\addplot3[area legend, line width=0.2pt, draw=mycolor2, fill=white!50!mycolor2, line cap=round,line join=round,forget plot]
table[row sep=crcr] {%
x	y	z\\
0	0.02676718375	0.02676718375\\
0	0.05888780425	0.02676718375\\
0	0.05888780425	0.05888780425\\
0	0.02676718375	0.05888780425\\
}--cycle;
\addplot3 [color=black, line width=0.3pt]
 table[row sep=crcr] {%
0	-0.042827494	0\\
0	-0.042827494	-0.05\\
};
 \addplot3 [color=black, line width=0.3pt]
 table[row sep=crcr] {%
0	-0.042827494	-0.05\\
0	0.007172506	-0.1\\
};
 \addplot3 [color=black, line width=0.3pt]
 table[row sep=crcr] {%
0	0	0\\
0	0	-0.05\\
};
 \addplot3 [color=black, line width=0.3pt]
 table[row sep=crcr] {%
0	0	-0.05\\
0	0.05	-0.1\\
};
 \addplot3 [color=black, line width=0.3pt]
 table[row sep=crcr] {%
0	0.042827494	0\\
0	0.042827494	-0.05\\
};
 \addplot3 [color=black, line width=0.3pt]
 table[row sep=crcr] {%
0	0.042827494	-0.05\\
0	0.092827494	-0.1\\
};
 \addplot3 [color=black, line width=0.3pt]
 table[row sep=crcr] {%
0	0.007172506	-0.1\\
0	0.110327494	-0.1\\
};
 \node[align=center, inner sep=0]
at (axis cs:0,0.125,-0.1) {$\mathbf{p}_y$};
\addplot3 [color=black, only marks, mark size=1.5pt, mark=*, mark options={solid, black}]
 table[row sep=crcr] {%
0	0	-0.042827494\\
0	0	0\\
0	0	0.042827494\\
};
 \addplot3 [color=black, line width=0.3pt]
 table[row sep=crcr] {%
0	0	-0.042827494\\
0	-0.05	-0.04282744\\
};
 \addplot3 [color=black, line width=0.3pt]
 table[row sep=crcr] {%
0	-0.05	-0.042827494\\
0	-0.1	0.007172506\\
};
 \addplot3 [color=black, line width=0.3pt]
 table[row sep=crcr] {%
0	0	0\\
0	-0.05	0\\
};
 \addplot3 [color=black, line width=0.3pt]
 table[row sep=crcr] {%
0	-0.05	0\\
0	-0.1	0.05\\
};
 \addplot3 [color=black, line width=0.3pt]
 table[row sep=crcr] {%
0	0	0.042827494\\
0	-0.05	0.042827494\\
};
 \addplot3 [color=black, line width=0.3pt]
 table[row sep=crcr] {%
0	-0.05	0.042827494\\
0	-0.1	0.092827494\\
};
 \addplot3 [color=black, line width=0.3pt]
 table[row sep=crcr] {%
0	-0.1	0.007172506\\
0	-0.1	0.110327494\\
};
 \node[align=center, inner sep=0]
at (axis cs:0,-0.1,0.12) {$\mathbf{p}_z$};
\addplot3 [color=black, only marks, mark size=1.5pt, mark=*, mark options={solid, black}]
 table[row sep=crcr] {%
0	-0.042827494	0\\
0	0	0\\
0	0.042827494	0\\
};
 \end{axis}

\end{tikzpicture}%
    \end{center}
    \vspace{-1.2cm}
    \captionof{figure}{Template array layout $\Ptemplate$ in local Cartesian \gls{pa} coordinates $\{x',y',z'\}$ with horizontal and vertical support vectors $\mathbf{p}_y$ and $\mathbf{p}_z$.}
    \label{fig:AP0}
    \vspace{-1mm}
    \end{figure}%
}%
With
$\mathrm{vecnorm}
\big(
    \Rknj{k}{n}{j}
\big) = \big[
    \lVert [\Rknj{k}{n}{j}]_{\scriptscriptstyle:,1} \rVert
    \hdots
    \lVert [\Rknj{k}{n}{j}]_{\scriptscriptstyle:,\Nantennas} \rVert
    \big]\in\realset{1}{\Nantennas}$
we compute the scalar distances from each \gls{pa}/\gls{va} antenna to the \gls{pm} that parameterize the spherical wavefront wideband array response.
Now, let 
$\V{u}\coloneqq\frac{{\V{r}}}{\lVert {\V{r}}\rVert}\in\realsetone{3}$ 
be a unit vector in \textit{global} coordinates defined w.r.t. the phase center of the respective \gls{pa}/\gls{va} through \eqref{eq:range}.
Instead of having a separate vector $\V{u}_m$ for each \gls{pa}/\gls{va} antenna $m$, we approximate that 
$\V{u}_{m}\approx \V{u}_{m'} \forall m,m' \in \{1 \hdots \Nantennas\}$, 
which implies a \textit{planar wavefront} (far-field) approximation per \gls{pa}.
Under this planar wavefront approximation, we can express the $\mathrm{vecnorm}(\cdot)$ operation using an inner product with $\V{u}$:
\begin{align}\label{Seq:planar-wavefront-response}
    \overline{\V{\psi}}(\pos{n},\!\psfvDet{k})
    \!\approx\!
    \vect
    \!\Bigg(\!
        \!\exp
        \!\Big(\!
            \frac{-\mathrm{j}2 \pi}{\lightspeed}
            \big(
                \pos{n} \onematrix{1}{\Nantennas}\!-\! \PlayoutDet{j}{k}
            \big)^{\!\trp}
            \V{u}\,
            \fpb^\trp
        \Big)    
    \!\Bigg)
\end{align}%
\ifthenelse{\equal{\IEEEversion}{true}}%
{%
}{%
    As described in~\cite{Leitinger23mvaSLAM}, the transformation from a \gls{pa} to a \gls{va} phase center position is 
    \begin{align*}
        \posVAdet{k}{j}(\psfvDet{k}) \coloneqq \posPA{j} -\left( 
            \frac{2 \, {\posPA{j}}^\trp \psfvDet{k}}{\lVert \psfvDet{k} \rVert^2} - 1
        \right) \, \psfvDet{k} \, ,
    \end{align*}
    and the Householder matrix
    \begin{align*}
        \houseDet{k}(\psfvDet{k}) \coloneqq \eye{3} - 2 \, \frac{{\psfvDet{k}} {\psfvDet{k}}^\trp}{\lVert \psfvDet{k} \rVert^2}
    \end{align*}
    represents the transformation from the \gls{pa} orientation to the \gls{va} orientation when mirrored across specular surface $k$.
    \begin{figure}
        \begin{center}
        	\setlength{\figurewidth}{1.0\linewidth}
            \input{./figures/geometricModel/AP4}
        \end{center}
        \vspace{-1cm}
        \captionof{figure}{\Gls{pa} layout $\PlayoutDet{j}{0}$ and \gls{va} layout $\PlayoutDet{j}{k}$ in global Cartesian coordinates $\{x,y,z\}$. Planar wavefront model.}
        \label{fig:AP4}
    \vspace{-1mm}
    \end{figure}
    \begin{figure}
        \begin{center}
        	\setlength{\figurewidth}{1.0\linewidth}
            \input{./figures/geometricModel/AP5NLoS}
        \end{center}
        \vspace{-1cm}
        \captionof{figure}{Template \gls{pa} layout $\PlayoutDet{j}{0}$ and \gls{va} layout $\PlayoutDet{j}{k}$ in global Cartesian coordinates $\{x,y,z\}$. Spherical wavefront model.}
        \label{fig:AP5NLoS}
    \vspace{-1mm}
    \end{figure}
}%
Inserting the single-bounce \gls{va} layout from~\eqref{Seq:VA-layout-SB}, 
we obtain for the argument of $\exp(\cdot)$ in \eqref{Seq:planar-wavefront-response} 
\begingroup             
\allowdisplaybreaks[4]  
\begin{align*}
    &\frac{-\mathrm{j}2 \pi}{\lightspeed}
    \Big(
        \onematrix{\Nantennas}{1}
        {\underbrace{
        (\pos{n} \!-\! \posVAdet{k}{j})
        }_{\eqqcolon \range{k}{n}{j}}}^{\!\trp}
        \V{u}
        -
        (\houseDet{k} \rotM{j} \Ptemplate)^{\!\trp}
        \V{u}
    \Big)
    \fpb^\trp
    \\[-3pt]
    &=\frac{-\mathrm{j}2 \pi}{\lightspeed}
        \onematrix{\Nantennas}{1}
        ({\range{k}{n}{j}}^\trp
        \V{u})
        \fpb^\trp
        +
        \frac{\mathrm{j}2 \pi}{\lightspeed}
        \Ptemplate^\trp \rotM{j}^\trp \houseDet{k}^\trp
        \V{u}\,
    \fpb^\trp
    \\
    &\approx
    \frac{-\mathrm{j}2 \pi}{\lightspeed}
        \onematrix{\Nantennas}{1}
        \lightspeed \tau
        \fpb^\trp
        +
        \frac{\mathrm{j}2 \pi}{\lightspeed}
        \Ptemplate^\trp \rotM{j}^\trp \houseDet{k}^\trp
        \V{u}\,
        \onematrix{1}{\Nfrequency}\fc
\end{align*}
\endgroup
where the approximation $\fpb^\trp\approx\onematrix{1}{\Nfrequency}\fc$ in the right term of the last line affects only the spatial response and assumes that the array has an approximately constant electrical aperture (measured in wavelengths squared) across the entire bandwidth.
Reinserting, we can reformulate 
$\overline{\V{\psi}}(\pos{n},\!\psfvDet{k})$
as an outer product of the temporal and spatial responses\footnote{For simplicity, we assume a unit-spectrum $S(f)\!=\!1$ over the sampled baseband.}
\begingroup             
\allowdisplaybreaks[4]  
\begin{align}\label{eq:response-outer-product}
    \overline{\V{\psi}}(\pos{n},\!\psfvDet{k})
    \!\approx\!
    \vect
    \!\Bigg(\!
        \underbrace{
        \!\exp
        \!\Big(\!
            \overbrace{
            \frac{\mathrm{j}2 \pi}{\lambda}
            \Ptemplate^\trp \rotM{j}^\trp \houseDet{k}^\trp
            \V{u}
            }^{(\Nantennas\times 1)}
        \!\Big) 
        }_{= \bm{a}(\V{u})}
        ~~~~~
        \nn\\[-4pt]
        \times 
        \underbrace{
        \exp
        \!\Big(\!
            \overbrace{
            -\!\mathrm{j}2 \pi
            \tau
            \mathbf{f}^\trp
            }^{(1 \times \Nfrequency)}
        \!\Big) 
        }_{\eqqcolon \bm{b}^\trp\!(\delay)}
        \!\Bigg)
        \times
        \underbrace{
        \exp
        \!\Big(\!
            -\!\mathrm{j}2 \pi
            \tau
            \fc
        \!\Big) 
        }_{\text{carrier phase term}}
    \nn
    \\
    =
    \bm{b}(\delay) \otimes \bm{a}(\bm{u}) \times 
    \exp \!\Big(\!\!-\!\mathrm{j}\frac{2\pi}{\lightspeed}\fc 
        \lVert \acute{\V{r}}\rVert 
        \Big) \,,
\end{align}
\endgroup
where we have used the identity $\vect(\V{A}\V{B})=\V{B}^\trp \otimes \V{A}$.
To arrive at our planar wavefront array response from \eqref{eq:unit-modulus-array-response},
the final step that remains to show is that the spatial array response $\bm{a}(\elevation,\azimuth)$  corresponds to the first term $\bm{a}(\bm{u})$ in \eqref{eq:response-outer-product}.
%
Now, let 
$\acute{\V{u}}\coloneqq\frac{\acute{\V{r}}}{\lVert \acute{\V{r}}\rVert}$ 
be a unit vector that is parameterized through directional cosines (cf.\,\cite[p.235 f.]{VanTrees2002optimumASP})
\begin{align}\label{eq:directional-cosines}
    \acute{\V{u}}(\elevation,\azimuth)
    =
    \begin{bmatrix}
        \acute{u}_x \\ \acute{u}_y \\ \acute{u}_z
    \end{bmatrix}
    =
    \begin{bmatrix}
        \sin \elevation \cos \azimuth \\ 
        \sin \elevation \sin \azimuth  \\ 
        \cos \elevation
    \end{bmatrix}
\end{align}
in \textit{local} spherical \gls{pa}/\gls{va} coordinates defined w.r.t. the phase center of the respective \gls{pa}/\gls{va}.
We reformulate $\bm{a}(\bm{u})$ in terms of $\acute{\V{u}}$ by exploiting that\footnote{Note that this corresponds to the inverse mapping from \eqref{Seq:VA-layout-SB}, where $(\houseDet{k} \rotM{j})^{-1}=(\houseDet{k} \rotM{j})^{\trp}$ due to the orthogonality of both $\houseDet{k}$ and $\rotM{j}$.} $\acute{\V{u}}= \rotM{j}^\trp \houseDet{k}^\trp\V{u}$ which gives
\begin{align}
    \bm{a}(\acute{\V{u}}) &= 
        \!\exp
        \!\Big(
            \frac{\mathrm{j}2 \pi}{\lambda}
            \Ptemplate^\trp \rotM{j}^\trp \houseDet{k}^\trp\V{u}
        \!\Big) 
        = 
        \!\exp
        \!\Big(
            \frac{\mathrm{j}2 \pi}{\lambda}
            \Ptemplate^\trp \acute{\V{u}}
        \Big) 
        \nn\\
        &=
        \!\exp
        \!\Big(
            \frac{\mathrm{j}2 \pi}{\lambda}
            \mathbf{p}_y \otimes \onematrix{\Nantennasz}{1} 
            \acute{u}_y
            +
            \frac{\mathrm{j}2 \pi}{\lambda}
            \onematrix{\Nantennasy}{1} \otimes \mathbf{p}_z
            \acute{u}_z
        \Big) 
        \\
        &=
        \!\exp
        \!\Big(
            \frac{\mathrm{j}2 \pi}{\lambda}
            \mathbf{p}_y \otimes \onematrix{\Nantennasz}{1} 
            \acute{u}_y
        \Big) 
        \odot
        \exp
        \!\Big(
            \frac{\mathrm{j}2 \pi}{\lambda}
            \onematrix{\Nantennasy}{1} \otimes \mathbf{p}_z
            \acute{u}_z
        \Big) 
        \nn
\end{align}
using $(\V{A}\otimes \V{B})^\trp=(\V{A}^\trp\otimes \V{B}^\trp)$ and
after insertion of the template layout from \eqref{eq:Ptemplate}.
Now, let 
$\bm{a}_{y}(\acute{u}_y)\coloneqq \exp\big(\frac{\mathrm{j}2 \pi}{\lambda}\mathbf{p}_y  \acute{u}_y\big)$ 
and
$\bm{a}_{z}(\acute{u}_z)\coloneqq 
\exp\big(\frac{\mathrm{j}2 \pi}{\lambda}\mathbf{p}_z\acute{u}_z\big)$.
We obtain
\begin{align}
    \bm{a}(\acute{\V{u}}) 
    &= 
    \big(
        \bm{a}_{y}(\acute{u}_y)
        ~\otimes~
        \onematrix{\Nantennasz}{1} 
    \big)
    ~ \odot ~ 
    \big(
        \onematrix{\Nantennasy}{1}
        ~\otimes~
        \bm{a}_{z}(\acute{u}_z)
    \big)
    \nn\\
    &=
    \big(
        \bm{a}_{y}(\acute{u}_y)
        ~\odot~
        \onematrix{\Nantennasy}{1}
    \big)
    ~ \otimes ~ 
    \big(
        \onematrix{\Nantennasz}{1}
        ~\odot~
        \bm{a}_{z}(\acute{u}_z)
    \big)
    \nn\\
    &= 
    \bm{a}_{y}(\acute{u}_y) \otimes \bm{a}_{z}(\acute{u}_z) 
    \\
    &= 
    \exp\Big(\frac{\mathrm{j}2 \pi}{\lambda}\mathbf{p}_y  \sin \elevation \sin \azimuth\Big)
    \otimes
    \exp\Big(\frac{\mathrm{j}2 \pi}{\lambda} \mathbf{p}_z \cos \elevation\Big)
    \nn
\end{align}
using the identity
$
\big(
        \V{A}
        \otimes
        \V{B}
    \big)
    \odot 
    \big(
        \V{C}
        \otimes
        \V{D}
    \big)
    =
    \big(
        \V{A}
        \odot
        \V{C}
    \big)
     \otimes  
    \big(
        \V{B}
        \odot
        \V{D}
    \big)
$
which confirms that $\bm{a}(\acute{\V{u}}) = \bm{a}_y(\elevation,\azimuth) \otimes \bm{a}_z(\elevation)$
and that \eqref{eq:response-outer-product} indeed corresponds to the array response in \eqref{eq:unit-modulus-array-response}.

Note that extensive tests of our algorithm have shown that it operates both with 
the planar wavefront model from \eqref{eq:unit-modulus-array-response} and
the spherical wavefront model from~\eqref{Seq:spherical-wavefront-response} with indistinguishable performance.
We chose to work with the former, because it leads to a simpler \gls{pcrlb} derivation and explainable Jacobian matrices in Sec.\,\ref{Ssec:Jacobian}.
}%

\section{Particle-Based Representations of Beliefs}\label{Ssec:PR-beliefs}

We derive the \gls{pr} of the \gls{mt} state belief in \eqref{eq:belief-state} as
\begin{align}
    &\belief(\state{n}) 
    \propto \beta(\state{n}) \prod\nolimits_{\scriptscriptstyle j\in \setAnchors} \iota(\state{n};\observation{n}{j}) 
    \nn\\[-2pt]
    &\approx
    \!\sum\nolimits_{\scriptscriptstyle p=1}^{\scriptscriptstyle P}
    \!\!\weight{\beta,n}{p} \delta\big( \state{n}-\particle{\state{n}}{p} \big)
    \prod\nolimits_{\scriptscriptstyle j\in \setAnchors} 
    \widetilde{\iota}(\state{n};\observation{n}{j})
    \\[-4pt]
    &
    =
    \!\sum_{\scriptscriptstyle p=1}^{\scriptscriptstyle P}
    \underbrace{
        \weight{\beta,n}{p} 
        \prod_{\scriptscriptstyle j\in \setAnchors}\! 
        \mathcal{CN}\Big(\!\observation{n}{j};\!\muiota(\particle{\state{n}}{p}),\Ciota(\particle{\state{n}}{p})\!\Big)
    }_{= \weightt{\bm{x},n}{p}}
    \delta\big(\state{n}\!-\!\particle{\state{n}}{p} \big)
    \nn
\end{align}
where the last line 
uses the property 
$f(x)\delta(x-a) \!=\! f(a)\delta(x-a)$ of the Dirac delta function~\cite[p.\,60]{Dirac1947QuantumMechanics}.
We derive the \gls{pr} of the belief about the \gls{pf} states in \eqref{eq:belief-PFy} as
\begingroup             
\begin{align}
    &\belief(\PFphi{s}{n}{j},1) \propto \alpha(\PFphi{s}{n}{j},1) 
    \prod\nolimits_{\scriptscriptstyle j\in \setAnchors}
    \kappa(\PFphi{s}{n}{j},1;\observation{n}{j}) 
    \nn\\[-1pt]
    &\approx
    \!\sum\nolimits_{\scriptscriptstyle p=1}^{\scriptscriptstyle P}
    \!\!\weight{\alpha,s,n}{p}
    \delta\big( \PFphi{s}{n}{j}-\particle{\PFphi{s}{n}{j,p}}{p} \big)
    \prod\nolimits_{\scriptscriptstyle j\in \setAnchors}
    \widetilde{\kappa}(\PFphi{s}{n}{j},1;\observation{n}{j}) 
    \nn\\[-3pt]
    &
    =
    \!\sum_{\scriptscriptstyle p=1}^{\scriptscriptstyle P}
    \underbrace{
        \weight{\alpha,s,n}{p}
        \prod\nolimits_{\scriptscriptstyle j\in \setAnchors}
        \mathcal{CN}\Big(\!\observation{n}{j};\!
        \mukappa(\particle{\PFphi{s}{n}{j,p}}{p},1),
        \Ckappa(\particle{\PFphi{s}{n}{j,p}}{p},1)
        \!\Big)
    }_{= \weightt{\bm{y},s,n}{p}}
    \nn\\[-2pt]
    &~\hspace{4.7cm}~\times\delta\big(\PFphi{s}{n}{j}\!-\!\particle{\PFphi{s}{n}{j,p}}{p} \big)\,,
\end{align}
\endgroup
and the \gls{pr} of the noise variance belief in \eqref{eq:belief-eta} as
\begin{align}
    &\belief(\etan{n}^{\scriptscriptstyle(j)}) \propto \xi(\etan{n}^{\scriptscriptstyle(j)}) \nu(\etan{n}^{\scriptscriptstyle(j)};\observation{n}{j})
    \nn\\[-2pt]
    &\approx
    \!\sum\nolimits_{\scriptscriptstyle p=1}^{\scriptscriptstyle P}
    \!\!\weight{\xi,n}{j,p} \delta\big( \etan{n}^{\scriptscriptstyle(j)}-\particle{\etan{n}}{j,p} \big)
    \widetilde{\nu}(\etan{n}^{\scriptscriptstyle(j)};\observation{n}{j})
    \\[-3pt]
    &
    =
    \!\sum_{\scriptscriptstyle p=1}^{\scriptscriptstyle P}
    \underbrace{
        \weight{\xi,n}{j,p} 
        \mathcal{CN}\Big(\!\observation{n}{j};\!\munu,\Cnu(\particle{\etan{n}}{j,p}) \!\Big)
    }_{= \weightt{\eta,n}{j,p}}
    \delta\big(\etan{n}^{\scriptscriptstyle(j)}\!-\!\particle{\etan{n}}{j,p} \big).
    \nn
\end{align}

\section{
The PF Belief Normalization Constant}\label{Ssec:PR-normalization-constant}
Since we require the beliefs to be normalized to integrate and sum to one over their support, making them valid probability distributions, we know that the belief $\belief(\PFy{s}{n}{j}) \!=\! \frac{1}{\normConst{\bm{y},s,n}{}} \alpha(\PFy{s}{n}{j}) \prod_{\scriptscriptstyle j\in\setAnchors}\widetilde{\kappa}(\PFy{s}{n}{j};\observation{n}{j})$ from \eqref{eq:belief-PFy} is properly normalized using a normalization constant $\normConst{\bm{y},s,n}{}$ ensuring that 
\begin{align}\label{Seq:norm-const-general}
    \int \belief(\PFphi{s}{n}{j},1) \mathrm{d}\PFphi{s}{n}{j} 
    &+
    \int \belief(\PFphi{s}{n}{j},0) \mathrm{d}\PFphi{s}{n}{j} = 1\,.
\end{align}
Inserting our \gls{pr}
$\alpha(\PFphi{s}{n}{j},1)\approx \sum\nolimits_{\scriptscriptstyle p=1}^{\scriptscriptstyle P}
\weight{\alpha,s,n}{p} \delta(\PFphi{s}{n}{j}\!-\! \particle{\PFphi{s}{n}{j,p}}{p})$
for the \gls{pf} prediction message, the first term $\int \belief(\PFphi{s}{n}{j},1) \mathrm{d}\PFphi{s}{n}{j}$ becomes
\begin{align}\label{Seq:norm-const-1}
    & \int\! \belief(\!\PFphi{s}{n}{j},\!1) \mathrm{d}\PFphi{s}{n}{j}\!\approx\!\!
    \!\int\!\!\frac{\alpha(\PFphi{s}{n}{j},1)}{\normConst{\bm{y},s,n}{}}
    \prod_{\scriptscriptstyle j\in\setAnchors}
    \widetilde{\kappa}(\!\PFphi{s}{n}{j},\!1;\!\observation{n}{j}\!)
    \mathrm{d}\PFphi{s}{n}{j} 
    \nn\\[-3pt]
    &\approx\!
    \sum_{\scriptscriptstyle p=1}^{\scriptscriptstyle P}
    \frac{\weight{\alpha,s,n}{p}}{\normConst{\bm{y},s,n}{}}
    \!\int\!\!\delta(\PFphi{s}{n}{j}\!-\! \particle{\PFphi{s}{n}{j,p}}{p}) 
    \!\prod_{\scriptscriptstyle j\in\setAnchors}
    \widetilde{\kappa}(\PFphi{s}{n}{j},\!1;\observation{n}{j})
    \mathrm{d}\PFphi{s}{n}{j} 
     \\[-3pt]
    &=\!
    \frac{1}{\normConst{\bm{y},s,n}{}}
    \!\!\sum_{\scriptscriptstyle p=1}^{\scriptscriptstyle P}
    \underbrace{\weight{\alpha,s,n}{p}\!
    \prod_{\scriptscriptstyle j\in\setAnchors}
    \mathcal{CN}\big(\observation{n}{j};\!\mukappa(\particle{\PFphi{s}{n}{j,p}}{p},1),\!\Ckappa(\particle{\PFphi{s}{n}{j,p}}{p},1)\big)
    }_{= \, \weightt{\bm{y},s,n}{p}}
    \nn 
\end{align}
while the second term becomes
\begingroup             
\allowdisplaybreaks[4]
\begin{align}\label{Seq:norm-const-0}
    & \int\!\belief(\!\PFphi{s}{n}{j},\!0) \mathrm{d}\PFphi{s}{n}{j}
    \!\approx\!\!
    \frac{1}{\normConst{\bm{y},s,n}{}}
    \!\int\!\!
    \alpha(\PFphi{s}{n}{j},0)
    \prod_{\scriptscriptstyle j\in\setAnchors}
    \underbrace{
        \widetilde{\kappa}(\!\PFphi{s}{n}{j},\!0;\!\observation{n}{j}\!)
    }_{\text{const. in\,}\PFphi{s}{n}{j}}
    \mathrm{d}\PFphi{s}{n}{j} 
    \nn \\[-7pt]
    &=
    \frac{1}{\normConst{\bm{y},s,n}{}}
    \Big(\prod_{\scriptscriptstyle j\in\setAnchors}
    \widetilde{\kappa}(\!\PFphi{s}{n}{j},\!0;\!\observation{n}{j}\!)
    \Big)
    \!\int\!\!
    \alpha(\PFphi{s}{n}{j},0)
    \mathrm{d}\PFphi{s}{n}{j} 
    \nn \\[-1pt]
    &=
    \frac{1}{\normConst{\bm{y},s,n}{}}
    \Big(\prod_{\scriptscriptstyle j\in\setAnchors}
    \widetilde{\kappa}(\!\PFphi{s}{n}{j},\!0;\!\observation{n}{j}\!)
    \Big)
    \Big(1 - \int\!\alpha(\PFphi{s}{n}{j},1) \mathrm{d}\PFphi{s}{n}{j} \Big)
    \nn\\[-3pt]
    &\approx
    \frac{1}{\normConst{\bm{y},s,n}{}}
    \Big(\prod_{\scriptscriptstyle j\in\setAnchors}
    \widetilde{\kappa}(\!\PFphi{s}{n}{j},\!0;\!\observation{n}{j}\!)
    \Big)
    \Big(1 - \sum_{\scriptscriptstyle p=1}^{\scriptscriptstyle P} \weight{\alpha,s,n}{p}  \Big)
\end{align}
\endgroup
using the marginal prior existence probability
$p(\PFr{s}{n}{j}\!=\!0|\observationn{1:n-1})\!=\! 1-p(\PFr{s}{n}{j}\!=\!1|\observationn{1:n-1})
\approx 1\!-\!\int\!\alpha(\PFphi{s}{n}{j},1) \mathrm{d}\PFphi{s}{n}{j}
\!\approx \!
1\!-\!\sum_{\scriptscriptstyle p=1}^{\scriptscriptstyle P} \weight{\alpha,s,n}{p} $.
Reinserting \eqref{Seq:norm-const-1} and \eqref{Seq:norm-const-0} into \eqref{Seq:norm-const-general}, we obtain the normalization constant $\normConst{\bm{y},s,n}{}$ 
as
\begin{align}
    &1 = \frac{\sum_{\scriptscriptstyle p=1}^{\scriptscriptstyle P}
    \weightt{\bm{y},s,n}{p}\!}{\normConst{\bm{y},s,n}{}}
    \!+\!
    \frac{1 \!-\!\! \sum_{\scriptscriptstyle p=1}^{\scriptscriptstyle P} \weight{\alpha,s,n}{p} }{\normConst{\bm{y},s,n}{}}
    \prod_{\scriptscriptstyle j\in\setAnchors}
    \widetilde{\kappa}(\!\PFphi{s}{n}{j},\!0;\!\observation{n}{j}\!)
    \\[-3pt]
    &\normConst{\bm{y},s,n}{} 
    = 
    \Big(\!
    \sum_{\scriptscriptstyle p=1}^{\scriptscriptstyle P}
    \weightt{\bm{y},s,n}{p}\!
    \Big)
    \!+\!
    \Big(1 \!-\! \sum_{\scriptscriptstyle p=1}^{\scriptscriptstyle P} \weight{\alpha,s,n}{p} \!\Big)
    \prod_{\scriptscriptstyle j\in\setAnchors}
    \widetilde{\kappa}(\!\PFphi{s}{n}{j},\!0;\!\observation{n}{j}\!)\,.
    \nn
\end{align}
The normalization constant of the \gls{ppr} belief in \eqref{eq:PR-normalization} is derived analogously.

\section{Efficient Implementation}\label{Ssec:efficient-implementation}
Particle-based evaluations of approximate Gaussian messages
$\widetilde{\iota}(\particle{\state{n}}{p};\observation{n}{j})$,
$\widetilde{\kappa}(\particle{\PFphi{s}{n}{j,p}}{p},\PFr{s}{n}{};\observation{n}{j})$, 
$\widetilde{\omega}(\PRr{s}{n}{j};\observation{n}{j})$, 
and
$\widetilde{\nu}(\particle{\etan{n}}{j,p};\observation{n}{j})$ demand the computation of determinants and inverses of 
$\Ciota(\particle{\state{n}}{p})$,
$\Ckappa(\particle{\PFphi{s}{n}{j,p}}{p},\PFr{s}{n}{j})$, 
$\Comega(\PRr{s}{n}{j})$, 
and
$\Cnu(\particle{\etan{n}}{j,p})$, which are of complexity $\mathcal{O}(\Nz^3)$.
We seek an approximation that reduces inverses and determinants from matrices of dimension $\Nz$, the observation length, to matrices of dimension $\Nfeaturest{n}$, the number of features, resulting in 
$\mathcal{O}(\Nfeaturest{n}^3)$.
For this purpose, we use the alternative approximation from \cite{LiaLeiMey:TSP2025-SuppDoc}, leading to the algorithm ``direct-\acrshort{slam}-fast''.
The idea is to find low-rank approximations of the covariance matrices \eqref{eq:Ciota}--\eqref{eq:Comega}.
For brevity, we present only the following \glspl{pr} of update message covariance matrices introduced at the beginning of the following subsections, detailed derivations of which are provided in the Extended Derivations~\cite{ThisPaperDerivations}.

For brevity, we define the approximate marginal prior nonexistence probabilities 
$\dot{p}_{s,n}^{(j)}\!\coloneqq\!1\!-\!\zeta(\PRr{s}{n}{j}\!=\!1)$,
$\dot{p}_{s,n}\!\coloneqq\!1\!-\!\sum_{\scriptscriptstyle p=1}^{P} \weight{\alpha,s,n}{p}$,
and
$\ddot{p}_{s,n}^{(j)}\!\coloneqq\!1-\zeta(\PRr{s}{n}{j}=1)\sum_{\scriptscriptstyle p=1}^{\scriptscriptstyle P} \weight{\alpha,s,n}{p}$,
and the \gls{pr} of the noise variance prior mean $\etaXit{n}\!\coloneqq\!\sum_{\scriptscriptstyle p=1}^{\scriptscriptstyle P} \weight{\xi,n}{j,p} \particle{\etan{n}}{j,p}$.

\subsection{\gls{pf} State Update Message}
A \gls{pr} of the \gls{pf} covariance matrix
$\Ckappa\!(\particle{\PFphi{s}{n}{j,p}}{p},\!\PFr{s}{n}{j})$
is
\begin{align*}
    \PFr{s}{n}{j}
    q_{\scriptscriptstyle s,n}^{\scriptscriptstyle (j,p)} 
    \steerVecx{s}{n}{j,p}
    {\steerVecx{s}{n}{j,p}}^\herm
    \!\!+\!
    \overbrace{\eye{\Nz} \etaXit{n}
    \!+\!
    \SMkappa{\SMkappa}^\herm
    }^{\eqqcolon \bm{A}_{\scriptscriptstyle s,n}^{\scriptscriptstyle\kappa(j)}}
\end{align*}
with $q_{\scriptscriptstyle s,n}^{\scriptscriptstyle (j,p)}\!\coloneqq\!\big(
            \particle{\PFgamma{s}{n}{}}{p}
            \!+\!
            |\particle{\PFmu{s}{n}{}}{p}|^2
            \dot{p}_{s,n}^{(j)}
        \big)
        \Mzeta{s}{j}(1)$
and
$\SMkappa\!=\!\big[ 
\mvec{0,n}{(j)} \cdots 
\mvec{s-1,n}{(j)},
\mvec{s+1,n}{(j)} \cdots 
\mvec{S_n,n}{(j)}
\big]\!\!\in\!\!\complexset{\Nz}{\Nfeatures{n}}$
using 
$\mvec{s,n}{(j)}\!\coloneqq\!
        \Big(\!\sum_{\scriptscriptstyle\pA=1}^{\scriptscriptstyle P} \!\weight{\alpha,s\!,n}{\pA}\!\Big)\!
        \sum_{\scriptscriptstyle\pB=1}^{P} 
        \!\weight{\beta,n}{\pB}
            \sqrt{
            \particle{\PFgamma{s\!}{n}{}}{\pB} \!+\! |\particle{\PFmu{s\!}{n}{\pB}}{\pB}|^2\ddot{p}_{s,n}^{(j)}}
        \steerVecx{s}{n}{j,\pB}
        \Mzeta{s}{j}\!(1)$.
The 
matrix inversion lemma~\cite[eq.\,(156)]{Cookbook} 
yields
\begin{align*}
    {\Ckappa}^{\!-1}=
    {\bm{A}_{\scriptscriptstyle s,n}^{\scriptscriptstyle\kappa(j)}}^{-1}
    -
    \frac{
        \PFr{s}{n}{j} 
        q_{\scriptscriptstyle s,n}^{\scriptscriptstyle (j,p)} 
        {\bm{A}_{\scriptscriptstyle s,n}^{\scriptscriptstyle\kappa(j)}}^{-1} 
        \steerVecx{s}{n}{j,p} 
        {\steerVecx{s}{n}{j,p}}^\herm
        {\bm{A}_{\scriptscriptstyle s,n}^{\scriptscriptstyle\kappa(j)}}^{-1} 
    }{
    1 + \PFr{s}{n}{j} 
        q_{\scriptscriptstyle s,n}^{\scriptscriptstyle (j,p)} 
        {\steerVecx{s}{n}{j,p}}^\herm
        {\bm{A}_{\scriptscriptstyle s,n}^{\scriptscriptstyle\kappa(j)}}^{-1}
        {\steerVecx{s}{n}{j,p}}
    }\,,
\end{align*}
and the generalized determinant lemma~\cite[p.\,420]{harville2008matrixAlgebra}
gives
\begin{align*}
    \det (\Ckappa)\!=\!
    \big|
        1
        \!+\!
        \PFr{s}{n}{j} 
        q_{\scriptscriptstyle s,n}^{\scriptscriptstyle (j,p)} 
        {\steerVecx{s}{n}{j,p}}^\herm\!\!
        {\bm{A}_{\scriptscriptstyle s,n}^{\scriptscriptstyle\kappa(j)^{-1}}}\!
        \steerVecx{s}{n}{j,p}
    \big|
    \det (
        {\bm{A}_{\scriptscriptstyle s,n}^{\scriptscriptstyle\kappa(j)}}
    ).
\end{align*}
Now, we apply the inversion lemma~\cite[eq.\,(159)]{Cookbook} again
\begin{align*}
    {\bm{A}_{\scriptscriptstyle s,n}^{\scriptscriptstyle\kappa(j)}}^{\scriptscriptstyle\!-1}
    \!\!\!&=\!\!
    \frac{\eye{\Nz}}{\etaXit{n}}
    \!-\!
    \frac{
    \SMkappa
    \!(
    \eye{\Nfeatures{n}}
    \!\!+\!
    \left.{\etaXit{n}}\right.^{\!\!\scriptscriptstyle-1}{\SMkappa}^\herm\!\SMkappa
    \!)^{\scriptscriptstyle\!-1}
    \!{\SMkappa}^\herm\!
    }{\left.{\etaXit{n}}\right.^{\!\!\scriptscriptstyle2}}
\end{align*}
and evaluate expressions $\V{a}^\herm{\bm{A}_{\scriptscriptstyle s,n}^{\scriptscriptstyle\kappa(j)}}^{\!-1}\V{b}$ with $\V{a},\V{b}\in\complexsetone{\Nz}$ as
\begin{align}\label{Seq:Maha-expression}
    &\frac{\V{a}^\herm\V{b}}{\etaXit{n}}
    \!-\!
    \frac{\V{a}^\herm\SMkappa
    \!(
    \eye{\Nfeatures{n}}
    \!\!+\!
    \left.{\etaXit{n}}\right.^{\!\!\scriptscriptstyle-1}{{\SMkappa}^\herm\!\SMkappa}
    \!)^{\scriptscriptstyle\!-1}
    \!{\SMkappa}^{\scriptscriptstyle\herm}\V{b}
    }{\left.{\etaXit{n}}\right.^{\!\!\scriptscriptstyle2}}
\end{align}
in an order s.t. only\footnote{Except for the matrix-matrix product forming the Gram matrix ${\SMkappa}^\herm\!\SMkappa$.} matrix-vector products are computed.
Let $\ekappa{(j,p)}\!\coloneqq\!\observation{n}{j}\!-\!\mukappa(\particle{\PFphi{s}{n}{j,p}}{p},\PFr{s}{n}{j})$ with $
\mukappa(\particle{\PFphi{s}{n}{j,p}}{p},\PFr{s}{n}{j})\!=\!\PFr{s}{n}{j} \musnj{2}(\particle{\PFphi{s}{n}{j,p}}{p}) \!+\! \sum\nolimits_{s' \in \setFeaturest{n} \setminus \{s\}} \musnjp{3}$,
the approximate complex Gaussian \gls{pf} update message evaluated at particle $\particle{\PFphi{s}{n}{j,p}}{p}$ is
\begin{align*}
    &\widetilde{\kappa}(\particle{\PFphi{s}{n}{j,p}}{p}\!,\!\PFr{s}{n}{j};\!\observation{n}{j}) 
    \!=\! 
    \mathcal{CN}\!\big(\observation{n}{j}\!;\!\mukappa\!(\!\particle{\PFphi{s}{n}{j,p}}{p}\!,\!\PFr{s}{n}{j}\!),\!\Ckappa\!(\!\particle{\PFphi{s}{n}{j,p}}{p}\!,\!\PFr{s}{n}{j}\!)\!\big)
    \nn\\
    &\!\approx\!
    \frac{
    \exp\!\!
    \left(\!\!
        \frac{
        \PFr{s}{n}{j}
        q_{\scriptscriptstyle s,n}^{\scriptscriptstyle (j,p)} 
        \big|
            {\steerVecx{s}{n}{j,p}}^{\herm}
            {\bm{A}_{\scriptscriptstyle s,n}^{\scriptscriptstyle\kappa(j)}}^{\!-1}
            \ekappa{(j,p)}
        \big|^2
        }{
        1
        +
        \PFr{s}{n}{j}
        q_{\scriptscriptstyle s,n}^{\scriptscriptstyle (j,p)} 
        {\steerVecx{s}{n}{j,p}}^{\herm}
        {\bm{A}_{\scriptscriptstyle s,n}^{\scriptscriptstyle\kappa(j)}}^{\!-1}
        \steerVecx{s}{n}{j,p}
        }
        \!-\!
        \ekappa{(j,p)\herm}
        {\bm{A}_{\scriptscriptstyle s,n}^{\scriptscriptstyle\kappa(j)}}^{\!-1}
        \ekappa{(j,p)}
    \!\!\!\right)
    }
    {
    \pi^{\Nz}
    \Big|
        1
        +
        \PFr{s}{n}{j} 
        q_{\scriptscriptstyle s,n}^{\scriptscriptstyle (j,p)} 
        {\steerVecx{s}{n}{j,p}}^\herm\!
        {\bm{A}_{\scriptscriptstyle s,n}^{\scriptscriptstyle\kappa(j)}}^{\!-1}
        \steerVecx{s}{n}{j,p}
    \Big|
    \det \!\left(
        {\bm{A}_{\scriptscriptstyle s,n}^{\scriptscriptstyle\kappa(j)}}
    \right)
    }
\end{align*}
where $\det \!\left({\bm{A}_{\scriptscriptstyle s,n}^{\scriptscriptstyle\kappa(j)}}\right)$ is never computed explicitly as it cancels in the normalization of $\weightt{\bm{y},s,n}{j,p}$. 
Evaluating terms 
$\ekappa{(j,p)\herm}
{\bm{A}_{\scriptscriptstyle s,n}^{\scriptscriptstyle\kappa(j)}}^{\!-1}
\ekappa{(j,p)}$,
${\steerVecx{s}{n}{j,p}}^{\herm}
{\bm{A}_{\scriptscriptstyle s,n}^{\scriptscriptstyle\kappa(j)}}^{\!-1}
\ekappa{(j,p)}$, and
${\steerVecx{s}{n}{j,p}}^{\herm}
{\bm{A}_{\scriptscriptstyle s,n}^{\scriptscriptstyle\kappa(j)}}^{\!-1}
\steerVecx{s}{n}{j,p}$
using \eqref{Seq:Maha-expression} has complexity 
$\mathcal{O}\big(\Nz \Nfeatures{n}\big)$ for each particle, plus the particle-independent inverse at $\mathcal{O}(\Nz\Nfeatures{n}^2 + \Nfeatures{n}^3)$ which is precomputed once.
These terms have to be computed for $\Nfeaturest{n}$ \glspl{pf}.

\subsection{PR State Update Message}
We find a \gls{pr} of the \gls{ppr} covariance matrix as
\begin{align*}
&\Comega\!(\PRr{s}{n}{j})
    \!\approx\! 
    \PRr{s}{n}{j}
    \bm{m}_{s,n}^{\omega(j)} 
    {\bm{m}_{s,n}^{\omega(j)}}^\herm
    \!\!+\!
    \overbrace{\eye{\Nz} \etaXit{n}
    \!+\!
    \SMomega{\SMomega}^\herm
    }^{\eqqcolon \bm{A}_{\scriptscriptstyle s,n}^{\scriptscriptstyle\omega(j)}}
\end{align*}
with 
$\mvec{s,n}{\omega(j)}\!\coloneqq\!
        \Big(\!\sum_{\scriptscriptstyle\pA=1}^{\scriptscriptstyle P} \!\weight{\alpha,s\!,n}{\pA}\!\Big)\!
        \sum_{\scriptscriptstyle\pB=1}^{P} 
        \!\weight{\beta,n}{\pB}
            \sqrt{
            \particle{\PFgamma{s\!}{n}{}}{\pB} \!+\! |\particle{\PFmu{s\!}{n}{\pB}}{\pB}|^2\dot{p}_{s,n}}
        \steerVecx{s}{n}{j,\pB}$
and
$\SMomega\!=\!
\SMkappa$.
The matrix inversion lemma 
yields
\begin{align*}
    {\Comega}^{\!-1}=
    {\bm{A}_{\scriptscriptstyle s,n}^{\scriptscriptstyle\omega(j)}}^{-1}
    -
    \frac{
        \PRr{s}{n}{j} 
        {\bm{A}_{\scriptscriptstyle s,n}^{\scriptscriptstyle\omega(j)}}^{-1} 
        \mvec{s,n}{\omega(j)}
        {\mvec{s,n}{\omega(j)}}^\herm
        {\bm{A}_{\scriptscriptstyle s,n}^{\scriptscriptstyle\omega(j)}}^{-1} 
    }{
    1 + \PRr{s}{n}{j} 
        {\mvec{s,n}{\omega(j)}}^\herm
        {\bm{A}_{\scriptscriptstyle s,n}^{\scriptscriptstyle\omega(j)}}^{-1}
        {\mvec{s,n}{\omega(j)}}
    }\,,
\end{align*}
and the generalized determinant lemma 
gives
\begin{align*}
    \det (\Comega)\!=\!
    \big|
        1
        \!+\!
        \PRr{s}{n}{j} 
        {\mvec{s,n}{\omega(j)}}^\herm\!
        {\bm{A}_{\scriptscriptstyle s,n}^{\scriptscriptstyle\omega(j)^{-1}}}\!
        \mvec{s,n}{\omega(j)}
    \big|
    \det (
        {\bm{A}_{\scriptscriptstyle s,n}^{\scriptscriptstyle\omega(j)}}
    )\,.
\end{align*}
Applying the inversion lemma 
again, we obtain ${\bm{A}_{\scriptscriptstyle s,n}^{\scriptscriptstyle\omega(j)}}^{\!-1}$ as
\begin{align*}
    \eye{\Nz} \left.{\etaXit{n}}\right.^{\!\!-1}
    \!\!-\!
    \left.{\etaXit{n}}\right.^{\!\!-2}\!
    \SMomega
    \!(
    \eye{\Nfeatures{n}}
    \!\!+\!
    \left.{\etaXit{n}}\right.^{\!\!-1}\!{\SMomega}^\herm\!\SMomega
    \!)^{\!-1}
    \!{\SMomega}^\herm\!
\end{align*}
and evaluate expressions $\V{a}^\herm{\bm{A}_{\scriptscriptstyle s,n}^{\scriptscriptstyle\omega(j)}}^{\!-1}\V{b}$ 
as in \eqref{Seq:Maha-expression}
in an order s.t. only matrix-vector products are computed.
Let $\eomega{(j)}\!\coloneqq\!\observation{n}{j}\!-\!\muomega(\PRr{s}{n}{j})$ with $
\muomega(\PRr{s}{n}{j})\!=\!\PRr{s}{n}{j} \musnj{4} \!+\! \sum\nolimits_{s' \in \setFeaturest{n} \setminus \{s\}} \musnjp{3}$,
the approximate complex Gaussian \gls{pr} update message evaluated is
\begin{align*}
    &\widetilde{\omega}(\PRr{s}{n}{j};\!\observation{n}{j}) 
    \!=\! 
    \mathcal{CN}\!\big(\observation{n}{j}\!;\!\muomega\!(\PRr{s}{n}{j}\!),\!\Comega\!(\PRr{s}{n}{j}\!)\!\big)
    \nn\\
    &\approx
    \frac{
    \exp\!
    \left(
        \frac{
        \PRr{s}{n}{j}
        \big|
            {\bm{m}_{s,n}^{\omega(j)}}^{\herm}
            {\bm{A}_{\scriptscriptstyle s,n}^{\scriptscriptstyle\omega(j)}}^{\!-1}
            \eomega{(j)}
        \big|^2
        }{
        1
        +
        \PRr{s}{n}{j} 
        {\bm{m}_{s,n}^{\omega(j)}}^{\herm}
        {\bm{A}_{\scriptscriptstyle s,n}^{\scriptscriptstyle\omega(j)}}^{\!-1}
        \bm{m}_{s,n}^{\omega(j)}
        }
        \!-\!
        \eomega{(j)\herm}
        {\bm{A}_{\scriptscriptstyle s,n}^{\scriptscriptstyle\omega(j)}}^{\!-1}
        \eomega{(j)}
    \right)
    }
    {
    \pi^{\Nz}
    \Big|
        1
        +
        \PRr{s}{n}{j} 
        {\bm{m}_{s,n}^{\omega(j)}}^\herm\!
        {\bm{A}_{\scriptscriptstyle s,n}^{\scriptscriptstyle\omega(j)}}^{\!-1}
        \bm{m}_{s,n}^{\omega(j)}
    \Big|
    \det \!\left(
        {\bm{A}_{\scriptscriptstyle s,n}^{\scriptscriptstyle\omega(j)}}
    \right)
    }
\end{align*}
where $\det \!\left({\bm{A}_{\scriptscriptstyle s,n}^{\scriptscriptstyle\omega(j)}}\right)$ is never computed explicitly, as it cancels in the normalization of the belief $\beliefp(r_{s,n}^{(j)}\!=\!1)$.
Forming the Gram matrix and the inverse has complexity $\mathcal{O}(\mathrm{N_z}S_n^2 +S_n^3)$.

\subsection{\gls{mt} State Update Message}%

We find a \gls{pr} of the \gls{mt} covariance matrix as 
\begin{align}
    \Ciota(\particle{\state{n}}{p}) \approx 
    \eye{\Nz}\etaXit{n} + \SMiota{\SMiota}^H
    \eqqcolon \bm{A}_{\scriptscriptstyle n}^{\scriptscriptstyle\iota(j,p)}
\end{align}
with
$\SMiota\!=\!\big[ 
\mvec{0,n}{\iota(j,p)} \cdots 
\mvec{S_n,n}{\iota(j,p)}
\big]\!\!\in\!\!\complexset{\Nz}{\Nfeaturest{n}}$
using 
$\mvec{s,n}{\iota(j,p)}\!\coloneqq\!
        \sqrt{
        \Mzeta{s}{j}\!(1)
        \big(\!\sum_{\scriptscriptstyle\pA=1}^{\scriptscriptstyle P} \!\weight{\alpha,s\!,n}{\pA}\!
        \big)
        \big(
            \particle{\PFgamma{s\!}{n}{}}{p} \!+\! |\particle{\PFmu{s\!}{n}{}}{p}|^2
            \ddot{p}_{s,n}^{(j)}
        \big)
        }
        \steerVecx{s}{n}{j,p}$.
%
%
Let
$\eiota{(j,p)}\!\coloneqq\!\observation{n}{j}\!-\!\muiota(\particle{\state{n}}{p})$ with 
$\muiota(\particle{\state{n}}{p}) 
=\sum\nolimits_{s \in \setFeaturest{n}} \musnj{1}(\particle{\state{n}}{p})$,
the approximate complex Gaussian \gls{mt} state update message evaluated at particle $\particle{\state{n}}{p}$ is 
\begin{align*}
    \widetilde{\iota}(\particle{\state{n}}{p};\!\observation{n}{j}) 
    &=
    \mathcal{CN}\!\big(\observation{n}{j};\muiota\!(\particle{\state{n}}{p}),\!\Ciota\!(\particle{\state{n}}{p})\big)
    \nn\\
    &\approx
    \frac{
    \exp\!
    \left(
        -
        \eiota{(j,p)\herm}
        {\bm{A}_{\scriptscriptstyle n}^{\scriptscriptstyle\iota(j,p)}}^{\!-1}
        \eiota{(j,p)}
    \right)
    }
    {
    \big(\pi \etaXit{n}\big)^{\Nz}
    \det \!\left(
        \eye{\Nfeaturest{n}}
        +
        \left.{\etaXit{n}}\right.^{-1}
        {\SMiota}^\herm\SMiota
    \right)
    }
\end{align*}
with 
$\eiota{(j,p)\herm}
{\bm{A}_{\scriptscriptstyle n}^{\scriptscriptstyle\iota(j,p)}}^{\!-1}
\eiota{(j,p)}$
computed as
\begin{align*}
    &\frac{\lVert\eiota{(j,p)}\rVert^2}{\etaXit{n}}
    \!-\!
    \frac{
        \Big\|
        \!\Big(
        \eye{\Nfeaturest{n}}
        \!\!+\!
        \left.{\etaXit{n}}\right.^{\!\!\scriptscriptstyle-1}\!{\SMiota}^\herm\!\SMiota
        \!\Big)^{\!\!-\frac{1}{2}}\!
        {\SMiota}^\herm\eiota{(j,p)}
        \Big\|^2
    }{\left.{\etaXit{n}}\right.^{2}}
\end{align*}
which has complexity $\mathcal{O}\big(\Nz \Nfeaturest{n}^2 \!+\! \Nfeaturest{n}^3\big)$ for each parallelized \gls{pa}-particle-pair $(j,p)$. 
Here, the inverse and determinant of complexity $\mathcal{O}(\Nfeaturest{n}^3)$ are now particle- and \gls{pa}-dependent 
but only need to be computed once for all \glspl{pf}.

\subsection{Noise Variance Update Message}
We find a \gls{pr} of the covariance matrix 
\begin{align}
    \Cnu(\particle{\etan{n}}{j,p}) \approx 
    \eye{\Nz}\particle{\etan{n}}{j,p} + \SMnu{\SMnu}^H
    \eqqcolon \bm{A}_{\scriptscriptstyle n}^{\scriptscriptstyle\nu(j,p)}
\end{align}
with
$\SMnu\!=\!\big[ 
\mvec{0,n}{(j)} \cdots 
\mvec{S_n,n}{(j)}
\big]\!\!\in\!\!\complexset{\Nz}{\Nfeaturest{n}}$.
Let 
$\enu{(j)}\!\coloneqq\!\observation{n}{j}\!-\!\munu$ with 
$\munu=\sum\nolimits_{s \in \setFeaturest{n}} \musnj{3}$,
the approximate 
noise variance update message evaluated at particle $\particle{\etan{n}}{j,p}$ is 
\begin{align*}
    \widetilde{\nu}(\particle{\etan{n}}{j,p};\!\observation{n}{j}) 
    &=
    \mathcal{CN}\!\big(\observation{n}{j};\munu,\!\Cnu\!(\particle{\etan{n}}{j,p})\big)
    \nn\\
    &\approx
    \frac{
    \exp\!
    \left(
        -
        \enu{(j)\herm}
        {\bm{A}_{\scriptscriptstyle n}^{\scriptscriptstyle\nu(j,p)}}^{\!-1}
        \enu{(j)}
    \right)
    }
    {
    \big(\pi \particle{\etan{n}}{j,p}\big)^{\Nz}
    \det \!\left(
        \eye{\Nfeaturest{n}}\!
        +
        {\particle{\etan{n}}{j,p}}^{-1}
        {\SMnu}^\herm\SMnu
    \right)
    }
\end{align*}
with 
$\enu{(j)\herm}
{\bm{A}_{\scriptscriptstyle n}^{\scriptscriptstyle\nu(j,p)}}^{\!-1}
\enu{(j)}$ 
computed as
\begin{align*}
    &~\frac{\lVert\enu{(j)}\rVert^2}{\particle{\etan{n}}{j,p}}
    \!-\!
    \frac{
        \Big\|
        \!\Big(
        \eye{\Nfeaturest{n}}
        \!\!+\!
        {\particle{\etan{n}}{j,p}}^{-1}
        {\SMnu}^\herm\SMnu
        \!\Big)^{\!\!-\frac{1}{2}}
        {\SMnu}^\herm\enu{(j)}
        \Big\|^2
    }
    {
        {\particle{\etan{n}}{j,p}}^{2}
    }
\end{align*}
which has complexity $\mathcal{O}\big(\Nz \Nfeaturest{n} \!+\! \Nfeaturest{n}^3\big)$ for each parallelized \gls{pa}-particle-pair $(j,p)$. 
Again, the inverse and determinant of complexity $\mathcal{O}(\Nfeaturest{n}^3)$ are particle- and \gls{pa}-dependent but do not need to be computed for each \gls{pf}.

\ifthenelse{\equal{\IEEEversion}{true}}
{
}{
\section{PR Existence Revival Across Time}
Our \gls{ppr}-existence model accommodates partial visibility, where rays can quickly disappear or reappear in dynamic scenarios depending on obstructions or limited surface extents.
Below, we show that even in the severe case where the old marginal posterior \gls{ppr} existence probability $\existenceProbPA{s}{n-1}{j}\ll1$ at time $n-1$, the \gls{ppr} existence can be ``revived,'' i.e.,
$\existenceProbPA{s}{n}{j} \approx 1$, at time $n$.

The 
marginal posterior \gls{ppr} existence probability is
\begin{align}
    \existenceProbPA{s}{n}{j} \!\approx\! 
    \beliefp(\PRr{s}{n}{j}\!=\!1)
    &=
    \frac{
        \Mzeta{s}{j}\!(1)\widetilde{\omega}(1;\observation{n}{j})
    }{
        \Mzeta{s}{j}\!(1)\widetilde{\omega}(1;\observation{n}{j})
        +
        (1\!-\!\Mzeta{s}{j}\!(1))\widetilde{\omega}(0;\observation{n}{j})
    }
    \nonumber \\
    &=
    \frac{
        1
    }{
        1+
        \frac{
        (1-\Mzeta{s}{j}\!(1))
        }
        {
        \Mzeta{s}{j}\!(1)
        }
        \frac{
        \widetilde{\omega}(0;\observation{n}{j})
        }
        {
        \widetilde{\omega}(1;\observation{n}{j})
        }
    }
    \nonumber \\
    &=
    \frac{1}{1 + \exp(-u) } = \sigma(u)
\end{align}
where we define
\begin{align}
    u \coloneqq & 
    \log \left(
    \frac{
         \Mzeta{s}{j}\!(1)
    }
    {
       (1-\Mzeta{s}{j}\!(1))
    }
    \frac{
        \widetilde{\omega}(1;\observation{n}{j})
    }
    {
        \widetilde{\omega}(0;\observation{n}{j})
    }
    \right)
    \nn\\
    =&
    \log 
    \underbrace{
    \frac{
        \Mzeta{s}{j}\!(1)
    }
    {
        1-\Mzeta{s}{j}\!(1)
    }
    }_{\text{``prior odds''}}
    +
    \log 
    \underbrace{
    \frac{
        \widetilde{\omega}(\mathcal{H}_1^{\scriptscriptstyle\mathrm{PR}};\observation{n}{j})
    }
    {
        \widetilde{\omega}(\mathcal{H}_0^{\scriptscriptstyle\mathrm{PR}};\observation{n}{j})
    }
    }_{\text{likelihood ratio}}
\end{align}
where $\sigma(u)$ is the standard logistic function depicted in Fig.\,\ref{fig:logisticCurve}.
\begin{figure}[t]
    \setlength{\plotWidth}{0.99\linewidth}
    \setlength{\plotHeight}{0.35\linewidth}
    \centering
%
%

\pgfplotsset{every axis/.append style={
  label style={font=\footnotesize},
  legend style={font=\footnotesize},
  tick label style={font=\footnotesize},
}}

\begin{tikzpicture}

\begin{axis}[%
width=0.951\plotWidth,
height=\plotHeight,
at={(0\plotWidth,0\plotHeight)},
scale only axis,
xmin=-9.99,
xmax=9.99,
xtick={-8,-6,-4,-2,0,...,8},
xlabel={$u$},
ymin=0,
ymax=1.15,
ylabel={$\sigma(u)$},
line cap = round,
line join = round,
axis lines=middle,
axis line style={->, arrows = {-Stealth[inset=0pt, scale=1.05, angle'=25]}, line width = 0.5pt},  
ylabel style={xshift=1mm,yshift=1mm},
xlabel style={yshift=1mm}
]
\addplot [color=IEEEblue, line width=1.2pt, forget plot, line cap = round, line join = round,]
  table[row sep=crcr]{%
-10	4.53978687024344e-05\\
-9.48717948717949	7.5811898756789e-05\\
-8.97435897435897	0.00012659904630485\\
-8.46153846153846	0.00021140181779907\\
-7.94871794871795	0.000352989939663747\\
-7.43589743589744	0.000589351981207617\\
-6.92307692307692	0.00098382626779454\\
-6.41025641025641	0.00164190247765048\\
-5.8974358974359	0.00273895556801556\\
-5.38461538461538	0.00456566317227383\\
-4.87179487179487	0.00760138135556447\\
-4.35897435897436	0.012629944471589\\
-3.84615384615385	0.020914959270220\\
-3.33333333333333	0.0344451956662112\\
-2.82051282051282	0.0562257148748818\\
-2.30769230769231	0.090487887990891\\
-1.79487179487179	0.1424764652054\\
-1.28205128205128	0.217201252300403\\
-0.769230769230769	0.316645529812217\\
-0.256410256410256	0.436246350533022\\
0.256410256410256	0.563753649466978\\
0.769230769230769	0.683354470187783\\
1.28205128205128	0.782798747699597\\
1.79487179487179	0.85752353479451\\
2.30769230769231	0.909512112009109\\
2.82051282051282	0.943774285125118\\
3.33333333333333	0.965554804333789\\
3.84615384615385	0.979085040729779\\
4.35897435897436	0.987370055528411\\
4.87179487179487	0.992398618644435\\
5.38461538461538	0.995434336827726\\
5.8974358974359	0.997261044431984\\
6.41025641025641	0.99835809752235\\
6.92307692307692	0.999016173732205\\
7.43589743589744	0.999410648018792\\
7.94871794871795	0.999647010060336\\
8.46153846153846	0.999788598182201\\
8.97435897435897	0.999873400953695\\
9.48717948717949	0.999924188101243\\
10	0.999954602131298\\
};
\end{axis}

\end{tikzpicture}%
    \caption{Standard logistic function $\sigma(u)$.}
    \label{fig:logisticCurve}
    \vspace{-1mm}
\end{figure}
The approximate marginal posterior existence probability $\existenceProbPA{s}{n}{j} \approx \beliefp(\PRr{s}{n}{j}\!=\!1)$ implicitly implements a binary hypothesis test between the null hypothesis $\mathcal{H}_0^{\scriptscriptstyle\mathrm{PR}} :\PRr{s}{n}{j}=0$ and the alternative hypothesis
$\mathcal{H}_1^{\scriptscriptstyle\mathrm{PR}} :\PRr{s}{n}{j}=1$.

Assume now that a \gls{ppr} is very close to \textit{not} existing at time $n-1$, i.e., $\beliefp(\PRr{s}{n-1}{j}\!=\!1)\!\approx\!0$, then the predicted existence probability $p(\PRr{s}{n}{j}=1| \observationn{1:n-1})\! \approx\!\zeta(\PRr{s}{n}{j}\!=\!1) 
\! \ll\! 1$, hence the \gls{ppr} is also very close to \textit{not} existing at time $n$, before updating using the observation $\observation{n}{j}$.
An almost nonexistent \gls{ppr} can be revived if the log-likelihood ratio outweighs the ``prior odds'':
Even for a \textit{very small} predicted existence probability $p(\PRr{s}{n}{j}=1| \observationn{1:n-1}) \approx \zeta(\PRr{s}{n}{j}\!=\!1) \ll 1$ it can be possible that
\begin{align*}
    -\log
    \frac{
        \Mzeta{s}{j}\!(1)
    }
    {
        1-\Mzeta{s}{j}\!(1)
    }
    \!\approx\!
    -\log
        \Mzeta{s}{j}\!(1)
    \!\ll \!
    \log
    \frac{
        \widetilde{\omega}(\mathcal{H}_1^{\scriptscriptstyle\mathrm{PR}};\observation{n}{j})
    }
    {
        \widetilde{\omega}(\mathcal{H}_0^{\scriptscriptstyle\mathrm{PR}};\observation{n}{j})
    }
\end{align*}
i.e., if the alternative hypothesis is much better supported by the observations $\observation{n}{j}$ leading to $u\! \gg\! 0$ and thus $p(\PRr{s}{n}{j}=1| \observationn{1:n}) \!\approx\! \beliefp(\PRr{s}{n}{j}\!=\!1) \!=\! \sigma(u) \!\approx\! 1$, as shown in Fig.\,\ref{fig:logisticCurve}, hence the marginal posterior existence probability can be ``revived'' even if the prior existence $\Mzeta{s}{j}\!(1) \!\ll\! 1$ was very low.
This means that our method supports partial visibility, a common propagation effect in \gls{dmimo}.
}

\section{Posterior Cram\'er--Rao Lower Bound}\label{Ssec:PCRLB}  
Note that the (i) noncoherent and (ii) coherent bounds derived in the following do \textit{not} represent the \gls{pcrlb} for the statistical model under which the proposed \gls{nzm} method has been derived, but merely serve as a benchmark to compare its estimation performance against.
Note that a closed-form solution for (i) the noncoherent \gls{pcrlb} can be found under the respective statistical model of the \gls{zm} method, however, the derivation of which is 
involved and therefore omitted in this Supplementary Material.
However, our \gls{mc} analysis demonstrates that \gls{zm} tends to perform close to (i) the noncoherent bound, while \gls{nzm} tends to perform close to (ii) the coherent bound derived hereafter.
This is important for \gls{dmimo} infrastructures because only estimators performing at the latter bound are able to retain the jointly coherent aperture formed by distributed \glspl{pa}, as we discuss in~\cite{RelatedPaper}.
As mentioned in Sec.\,\ref{sec:PCRLB}, throughout the \gls{pcrlb} derivation, the amplitudes 
$\widetilde{{\varrho}}_{\!\scriptscriptstyle k,n}^{\scriptscriptstyle(j)}
\coloneqq
\varrho_{\!\scriptscriptstyle k,n}^{\scriptscriptstyle(j)}
g_{\scriptscriptstyle k,n}^{\scriptscriptstyle (j)}
\exp (\!-\mathrm{j}\frac{2\pi}{\lightspeed}\fc \lVert \acute{\V{r}}\rVert )$
absorb the pathloss $g_{\scriptscriptstyle k,n}^{\scriptscriptstyle (j)}\!=\!\frac{\lambda}{4 \pi \lVert \acute{\V{r}} \rVert}$ instead of the steering vectors $\steerVec{j}$ as well as the carrier-phase term $\exp (\!-\mathrm{j}\frac{2\pi}{\lightspeed}\fc \lVert \acute{\V{r}}\rVert)$ from~\eqref{eq:unit-modulus-array-response}. 

{\slshape Approximations.} 
In the following derivation, these amplitudes $\widetilde{\unslant[-.25]{\varrho}}_{\!\scriptscriptstyle k,n}^{\scriptscriptstyle(j)}$ at time $n$ and $n-1$ are going to be treated as independent \glspl{rv}, i.e., $\widetilde{\unslant[-.25]{\varrho}}_{\!\scriptscriptstyle k,n}^{\scriptscriptstyle(j)} \perp \widetilde{\unslant[-.25]{\varrho}}_{\!\scriptscriptstyle k,n-1}^{\scriptscriptstyle(j)}$. 
For the polar form of the amplitudes 
$\widetilde{\unslant[-.25]{\varrho}}_{\!\scriptscriptstyle k,n}^{\scriptscriptstyle(j)}=
\rv{a}_{\!\scriptscriptstyle k,n}^{\scriptscriptstyle(j)}
\exp(\mathrm{j} \unslant[-.25]{\varphi}_{\!\scriptscriptstyle k,n}^{\scriptscriptstyle(j)})$, as well as for the noise variance $\RVetan{n}$ we introduce pseudo state-transition \glspl{pdf} 
$f(\modulivec{n}{}|\modulivec{n-1}{})=\mathcal{N}(\modulivec{n}{};\mathbf{0},\sigma_a^2 \eye{\widetilde{K}J})$,
$f(\phasevec{n}{}|\phasevec{n-1}{})=\mathcal{N}(\phasevec{n}{};\mathbf{0},\sigma_\varphi^2 \eye{\dimPhase})$,
and
$f(\etan{n}|\etan{n-1})=\mathcal{N}(\etan{n};0,\sigma_\eta^2)$, with $\{\sigma_a^2,\sigma_\varphi^2,\sigma_\eta^2\}$ large enough such that the information about $\{{\RVphasevec{n}{}}, {\RVmodulivec{n}{}}, \RVetan{n}\}$ injected through $\FIMstep{n}{n\!-\!1}$ into $\FIMstep{n}{n}$ is negligible compared to the information injected through $\FIMglobal{n}$, yet allowing to compute $\FIMstep{n}{n\!-\!1}$ and the inverses involved in its computation~\eqref{eq:priorFIM} without rank-deficiency.
Furthermore, we make the simplifying approximation that the amplitude moduli 
$\RVmodulivec{n}{(j)}$
do not contribute information on the parameters of interest 
$\{\RVstate{n} , {\RVpMVAposStacked}\}$, hence the respective Jacobians are zero matrices, i.e., 
$\jacobP{\modulus} = \frac{\partial {\modulivec{n}{(j)}}^\trp}{\partial \pos{n}} \approx 
\mathbf{0}$
and
$\jacobM{\modulus}=\frac{\partial {\modulivec{n}{(j)}}^\trp}{\partial \psfv{s}{n}{j}}\approx \mathbf{0}$. 
This is a valid approximation because the information contributed by the \textit{other} local channel parameters 
$\big\{
\RVelVecx{n}{j}, 
\RVazVecx{n}{j},
\RVdelayVecx{n}{j}, 
\RVphasevec{n}{(j)} 
\big\}$
on the parameters of interest is generally much larger.

The \gls{pcrlb} matrix $\PCRLB$ is a lower bound on the 
\gls{mse} matrix 
$\mathbb{E}_{\,\RVetaglobalSmall{n}, \RVobservationn{n}|\RVobservationn{1:n-1}}
\!\big((\RVetaglobalHat{n}\!-\RVetaglobal{n} ) (\RVetaglobalHat{n}\!-\RVetaglobal{n} )^\trp\big)$
of any \textit{Bayesian} estimator~\cite[eq.\,(29)]{VanTrees2007PCRLB}, implying that \textit{both} the observations $\RVobservationn{n}$ and the state $\RVetaglobal{n}$ are \glspl{rv}. 
For the \gls{pcrlb} $\PCRLB$ to lower-bound the estimation error, the \gls{mt} position $\mathbf{p}_{\scriptscriptstyle n}$ and \glspl{sfv} positions $\psfvDet{k}$ are no longer deterministic unknowns as in a ``classic'' estimation problem, but they must become \glspl{rv} $\RVpos{n}$ and $\RVpsfv{k}{n}{j}$ not only in the estimator, but also in the data-generating process~\cite{Fritsche16ParametricCRLB}.

{\slshape Synthetic Data Generation:} 
That is, for each \gls{mc} run, at each time $n$, an 
\gls{mt} state $\state{n} \!\coloneqq\! [\pos{n}^\trp \iist \vel{n}^\trp]^\trp$
realization is drawn from its state-transition \gls{pdf} $f(\state{n}|\state{n-1})=\mathcal{N}(\state{n};\transitionmatrix_{\text{\tiny a}}\state{n-1},\processNoiseCov_{\text{\tiny a}})$
leading to a different \gls{ncv} trajectory $\big\{\pos{n}\big\}_{\scriptscriptstyle n=1}^{\scriptscriptstyle N}$.
Likewise, \gls{sfv} positions $\psfv{k}{n}{1}$ are drawn from their state-transition \glspl{pdf} $f(\psfv{k}{n}{1}|\psfv{k}{n-1}{J})=\mathcal{N}(\psfv{k}{n}{1};\eye{3}\psfv{k}{n-1}{J},\sigma_{\text{\tiny sfv}}^2\eye{3})$
s.t. $\transitionmatrix_{\text{\tiny sfv}}=\eye{3K}$
and $\processNoiseCov_{\text{\tiny sfv}}=\sigma_{\text{\tiny sfv}}^2\eye{3K}$
leading to different random walk trajectories $\big\{\psfv{k}{n}{1}\big\}_{\scriptscriptstyle n=1}^{\scriptscriptstyle N}~\forall k \in \{1 \hdots K\}$.
We draw $\pos{1}$ and $\psfv{k}{1}{1}$ at time $n=1$ from the state-transition \glspl{pdf}
$f(\state{1};\mathbf{x}_{\scriptscriptstyle 0})$ and $f(\psfv{k}{1}{1};\psfvDet{k})$ parameterized by fictional starting points $\mathbf{x}_{\scriptscriptstyle 0}\!\coloneqq\! [\mathbf{p}_{\scriptscriptstyle 0}^\trp \iist \mathbf{v}_{\scriptscriptstyle 0}^\trp]^\trp$ and $\psfvDet{k}$ at time $n\!=\!0$, which define $\FIMstep{1}{0}$ as the starting-point of the recursion in~\eqref{eq:posteriorCRLB}--\eqref{eq:priorFIM}.
We use the \gls{mc} analysis both to compute the \gls{mse} matrices of our estimators, and to numerically implement the expectation $\FIMglobal{n} \!=\! \mathbb{E}_{\RVetaglobalSmall{n}|\RVobservationn{1:n-1}}\!(\FIMclassic{n})$ in the Bayesian snapshot information matrix from Sec.\,\ref{sec:FIMg} by means of \gls{mc} integration (cf.\,\cite{Hernandez02PCRLB}).
Under the above approximations, the linear Gaussian state-transition \gls{pdf} factorizes as
$f(\etaglobal{n}|\etaglobal{n-1})\!=\!
f(\state{n}|\state{n-1})
\big(\prod\nolimits_{\scriptscriptstyle k=1}^{\scriptscriptstyle K}f(\psfv{k}{n}{1}|\psfv{k}{n-1}{J})\big)
f(\phasevec{n}{}|\phasevec{n-1}{})
f(\modulivec{n}{}|\modulivec{n-1}{})$
$f(\etan{n}|\etan{n-1})$
such that
$f(\etaglobal{n}|\etaglobal{n-1})\!=\!\mathcal{N}(\etaglobal{n};\transitionmatrix \etaglobal{n-1},\processNoiseCov)$ is parameterized by state transition matrix and process noise covariance matrix
\begin{align*}
    \transitionmatrix \!=\! 
    \begin{bmatrixs}
        \transitionmatrix_{\text{\tiny a}} & \mathbf{0}     & \mathbf{0} & \mathbf{0} & 0 \\
        \mathbf{0} & \transitionmatrix_{\text{\tiny sfv}}   & \mathbf{0} & \mathbf{0} & 0 \\
        \mathbf{0} & \mathbf{0} & \mathbf{0}       & \mathbf{0} & 0 \\
        \mathbf{0} & \mathbf{0} & \mathbf{0} & \mathbf{0}   & 0      \\
        {0} & {0} & {0} & {0}   & 0
    \end{bmatrixs}
\text{,} \,
    \processNoiseCov \!=\! 
    \begin{bmatrixs}
        \processNoiseCov_{\text{\tiny a}} & \mathbf{0}     & \mathbf{0} & \mathbf{0} & 0 \\
        \mathbf{0} & \processNoiseCov_{\text{\tiny sfv}}   & \mathbf{0} & \mathbf{0} & 0 \\
        \mathbf{0} & \mathbf{0} & \sigma_\varphi^2 \eye{\dimPhase}       & \mathbf{0} & 0 \\
        \mathbf{0} & \mathbf{0} & \mathbf{0} & \sigma_a^2 \eye{\widetilde{K}J}   & 0      \\
        {0} & {0} & {0} & {0}   & \sigma_\eta^2 
    \end{bmatrixs} \,
\end{align*}
respectively.
Under this (approximate) linear Gaussian state-transition \gls{pdf}, the predicted information matrix from \eqref{eq:priorFIM} takes the form
$\FIMstep{n}{n\!-\!1} = \big(\transitionmatrix \, \FIMstep{n\!-\!1}{n\!-\!1}^{-1} \, \transitionmatrix^\trp + \processNoiseCov\big)^{-1} \,,$
as shown by Hernandez et al.\,\cite[eq.\,(16)]{Hernandez02PCRLB}.

With the prior information matrix defined, what is missing to complete the recursion in~\eqref{eq:posteriorCRLB}--\eqref{eq:priorFIM} are the entries of the local channel \glspl{fim} $\FIMch{n}{j}$ from~\eqref{eq:FIMclassic} and the Jacobian matrices $\jacobgn{j}= \nicefrac{\partial {\etach{n}{j}}^{\!\!\!\trp}}{\partial \etaglobal{n}}$ from~\eqref{eq:FIMclassic} mapping from \gls{pa}-local channel parameters $\etach{n}{j}$ to global (infrastructure) parameters $\etaglobal{n}$.

Recall the definitions of the delay, elevation, and azimuth
\begin{align}\label{eq:delay}
    \delay(\acute{\V{r}}) &= \nicefrac{\lVert \acute{\V{r}} \rVert}{\lightspeed} \,,
    \\
    \label{eq:elevation}
    \elevation(\acute{\V{r}}) &= \arccos (\rz/ \lVert \acute{\V{r}} \rVert)\,,
    \\
    \label{eq:azimuth}
    \azimuth(\acute{\V{r}}) &= \arctantwo (\ry,\rx)\,,
\end{align}
representing local spherical \gls{pa}-coordinates parameterized by 
\begin{align}\label{eq:rangepSupp}
    \rangep{k}{n}{j}(\state{n},\psfv{k}{n}{},\posPA{j},\rotM{j}) 
    &\coloneqq \rotM{j}^{-1} \ranget{k}{n}{j} 
    = \rotM{j}^{\trp} \ranget{k}{n}{j} \,,
\\
    \label{eq:ranget}
    \ranget{k}{n}{j}(\state{n},\psfv{k}{n}{},\posPA{j}) &\coloneqq \houseDet{k} \range{k}{n}{j}\,,
\\
\label{eq:range}
    \range{k}{n}{j}(\state{n},\psfv{k}{n}{},\posPA{j}) &\coloneqq \pos{n} - \posVA{k}{n}{j}\,.
\end{align}
Here, the vector $\range{k}{n}{j}$ points from a \gls{va} to the \gls{mt} in global coordinates, $\ranget{k}{n}{j}$ points from the \gls{pa} to the \gls{vm} in global coordinates, and $\rangep{k}{n}{j}$ points from the \gls{pa} to the \gls{vm} in local \gls{pa}-coordinates.

\subsection{Derivation of Local Channel FIM Terms}\label{Ssec:FIMch}  
The \gls{pa}-local channel parameter vector $\RVetach{n}{j}$ in Sec.\,\ref{sec:FIMch} contains the stacked elevation vector
$\elVecx{n}{j} \!=\! \big[\elx{0}{n}{j} \ist \hdots \ist \elx{K}{n}{j} \big]^{\trp}\!\in\!\realsetone{\widetilde{K}}$, the stacked azimuth vector
$\azVecx{n}{j} \!=\! \big[ \azx{0}{n}{j} \ist \hdots \ist \azx{K}{n}{j}\big]^{\trp}\!\in\!\realsetone{\widetilde{K}}$,
and the stacked delay vector
$\delayVecx{n}{j} \!=\! \big[ \delayx{0}{n}{j} \ist \hdots \ist \delayx{K}{n}{j}\big]^{\trp} \!\in\!\realsetone{\widetilde{K}}$.
Here, 
$\delayx{k}{n}{j}\coloneqq \tau(\rangep{k}{n}{j})$ is defined through \eqref{eq:delay},
$\elx{k}{n}{j}\coloneqq \elevation(\rangep{k}{n}{j})$ is defined through \eqref{eq:elevation}, and
$\azx{k}{n}{j}\coloneqq \azimuth(\rangep{k}{n}{j})$ is defined through \eqref{eq:azimuth}.
For the derivation of the channel \gls{fim} terms in \eqref{eq:FIMch}, we assume a complex Gaussian likelihood with spatially and temporally uncorrelated, circular \gls{awgn} $\RVnoise{n}{j}$ leading to
$f(\observation{n}{j}|\etach{n}{j})\!\coloneqq\!\CN\big(\observation{n}{j}; \sum\nolimits_{k = 0}^{K}  \widetilde{\varrho}_{\!\scriptscriptstyle k,n}^{\scriptscriptstyle(j)} \signalatomnl{k}{n}{j},\etan{n} \eye{\Nz}\big)$
where the amplitudes are decomposed into polar form as vectors of moduli $\modulivec{n}{(j)}\!\in\!\realsetone{\widetilde{K}}$ and phases $\phasevec{n}{(j)}\!\in\!\realsetone{\widetilde{K}}$.
Evaluating~\eqref{eq:FIMch}, the entries of our local channel \gls{fim} are given by Table~\ref{Stab:channel-fim}.
\begin{table}[t]
    \setlength{\tabcolsep}{2.2pt} 
    \renewcommand{\arraystretch}{1.05} 
    \centering \small
    \caption
    {Individual entries $[\FIMch{n}{j}]_{\scriptstyle \grave{\imath},\acute{\imath}}$ of the channel \gls{fim} in~\eqref{eq:FIMch}.
    Iterating over the listed row and column indices 
    $\widetilde{k},\widetilde{k}'$ 
    yields the entries that correspond to the upper triangular matrix of $\FIMch{n}{j}$.
        }%
\label{Stab:channel-fim}
    \begin{tabularx}{1\columnwidth}{@{}cc|c|c|c}
        \toprule
        $[\etach{n}{j}]_{\scriptscriptstyle \grave{\imath}}$ & $[\etach{n}{j}]_{\scriptscriptstyle \acute{\imath}}$
        & Term & \scalebox{0.9}{Row\,$\grave{\imath}$} & \scalebox{0.9}{Column\,$\acute{\imath}$}  \\
        \midrule 
        $\elx{k}{n}{j}$&$\elx{k'\!}{n}{j}$ &
            $\frac{2}{\etan{n}}\Re\{
                     \left.\complexamplitude{\ncomponent}\right.^{\!\ast}\complexamplitude{\ncomponent'\!}
                    {\atomelnl{(j) \herm}{\ncomponent,n}}
                    \atomelnl{(j)}{\ncomponent'\!,n}
                \}$  &
            $\widetilde{k}$ &  
            $\widetilde{k}' $ \\ 
        $\elx{k}{n}{j}$&$\azx{k'\!}{n}{j}$ &
            $\frac{2}{\etan{n}}\Re\{
                    \left.\complexamplitude{\ncomponent}\right.^{\!\ast}\complexamplitude{\ncomponent'\!} 
                    {\atomelnl{(j) \herm}{\ncomponent,n}}
                    \atomaznl{(j)}{\ncomponent'\!,n}
                \}$  &
            $\widetilde{k}$ &  
            $\widetilde{k}' + \widetilde{K}$ \\ 
        $\elx{k}{n}{j}$&$\delayx{k'\!}{n}{j}$ & 
            $\frac{2}{\etan{n}}\Re\{
                    \left.\complexamplitude{\ncomponent}\right.^{\!\ast}\complexamplitude{\ncomponent'\!} 
                    {\atomelnl{(j) \herm}{\ncomponent,n}} 
                    \atomdelaynl{(j)}{\ncomponent'\!,n}
                \}$  &
            $\widetilde{k}$ &  
            $\widetilde{k}' + 2\widetilde{K}$ \\ 
        $\elx{k}{n}{j}$&$\phasex{k'\!}{n}{j}$ & 
            $\frac{2}{\etan{n}}\Re\{
                    \mathrm{j} \, \left.\complexamplitude{\ncomponent}\right.^{\!\ast}\complexamplitude{\ncomponent'\!}
                    {\atomelnl{(j) \herm}{\ncomponent,n}}
                    {\signalatomnl{k'\!}{n}{j}}
                \}$  &
            $\widetilde{k}$ &  
            $\widetilde{k}' + 3\widetilde{K}$ \\ 
        $\elx{k}{n}{j}$&$\modulusx{k'\!}{n}{j}$ & 
            $\frac{2}{\etan{n}}\Re\{
                    \left.\complexamplitude{\ncomponent}\right.^{\!\ast} e^{\mathrm{j}\phasex{\ncomponent'\!}{n}{j}} 
                    {\atomelnl{(j) \herm}{\ncomponent,n}}
                    {\signalatomnl{k'\!}{n}{j}}
                \}$  &
            $\widetilde{k}$ &  
            $\widetilde{k}' + 4\widetilde{K}$ \\ 
        $\elx{k}{n}{j}$&$\etan{n}$ & 
            $0$  &
            $\widetilde{k}$ &  
            $1 + 5\widetilde{K}$ \\ 
   \midrule
        $\azx{k}{n}{j}$&$\azx{k'\!}{n}{j}$ & 
            $\frac{2}{\etan{n}}\Re\{
                    \left.\complexamplitude{\ncomponent}\right.^{\!\ast} \complexamplitude{\ncomponent'\!} 
                    {\atomaznl{(j) \herm}{\ncomponent,n}} 
                    \atomaznl{(j)}{\ncomponent'\!,n}
                \}$  &
            $\widetilde{k}+ \widetilde{K}$ &  
            $\widetilde{k}' + \widetilde{K}$ \\ 
        $\azx{k}{n}{j}$&$\delayx{k'\!}{n}{j}$ & 
            $\frac{2}{\etan{n}}\Re\{
                    \left.\complexamplitude{\ncomponent}\right.^{\!\ast} \complexamplitude{\ncomponent'\!} 
                    {\atomaznl{(j) \herm}{\ncomponent,n}}
                    \atomdelaynl{(j)}{\ncomponent'\!,n}
                \}$  &
            $\widetilde{k}+ \widetilde{K}$ &  
            $\widetilde{k}' + 2\widetilde{K}$ \\ 
        $\azx{k}{n}{j}$&$\phasex{k'\!}{n}{j}$ & 
            $\frac{2}{\etan{n}}\Re\{
                    \mathrm{j} \, \left.\complexamplitude{\ncomponent}\right.^{\!\ast} \complexamplitude{\ncomponent'\!} 
                    {\atomaznl{(j) \herm}{\ncomponent,n}}
                    {\signalatomnl{k'\!}{n}{j}}
                \}$  &
            $\widetilde{k}+ \widetilde{K}$ &  
            $\widetilde{k}' + 3\widetilde{K}$ \\ 
        $\azx{k}{n}{j}$&$\modulusx{k'\!}{n}{j}$ & 
            $\frac{2}{\etan{n}}\Re\{
                    \left.\complexamplitude{\ncomponent}\right.^{\!\ast} e^{\mathrm{j}\phasex{\ncomponent'\!}{n}{j}} 
                    {\atomaznl{(j) \herm}{\ncomponent,n}}
                    {\signalatomnl{k'\!}{n}{j}}
                \}$  &
            $\widetilde{k}+ \widetilde{K}$ &  
            $\widetilde{k}' + 4\widetilde{K}$ \\ 
        $\azx{k}{n}{j}$&$\etan{n}$ & 
            $0$  &
            $\widetilde{k}+ \widetilde{K}$ &  
            $1 + 5\widetilde{K}$ \\ 
   \midrule
        $\delayx{k}{n}{j}$&$\delayx{k'\!}{n}{j}$ & 
            $\frac{2}{\etan{n}}\Re\{
                    \left.\complexamplitude{\ncomponent}\right.^{\!\ast} \complexamplitude{\ncomponent'\!} 
                    {\atomdelaynl{(j) \herm}{\ncomponent,n}}
                    \atomdelaynl{(j)}{\ncomponent'\!,n}
                \}$  &
            $\widetilde{k} + 2\widetilde{K}$ &  
            $\widetilde{k}' + 2\widetilde{K}$ \\ 
        $\delayx{k}{n}{j}$&$\phasex{k'\!}{n}{j}$ & 
            $\frac{2}{\etan{n}}\Re\{
                    \mathrm{j} \, \left.\complexamplitude{\ncomponent}\right.^{\!\ast} \complexamplitude{\ncomponent'\!} 
                    {\atomdelaynl{(j) \herm}{\ncomponent,n}}
                    {\signalatomnl{k'\!}{n}{j}}
                \}$  &
            $\widetilde{k}+ 2\widetilde{K}$ &  
            $\widetilde{k}' + 3\widetilde{K}$ \\ 
        $\delayx{k}{n}{j}$&$\modulusx{k'\!}{n}{j}$ & 
            $\frac{2}{\etan{n}}\Re\{
                    \left.\complexamplitude{\ncomponent}\right.^{\!\ast} e^{\mathrm{j}\phasex{\ncomponent'\!}{n}{j}} 
                    {\atomdelaynl{(j) \herm}{\ncomponent,n}}
                    {\signalatomnl{k'\!}{n}{j}}
                \}$  &
            $\widetilde{k}+ 2\widetilde{K}$ &  
            $\widetilde{k}'+ 4\widetilde{K}$ \\ 
        $\delayx{k}{n}{j}$&$\etan{n}$ & 
            $0$  &
            $\widetilde{k}+ 2\widetilde{K}$ &  
            $1 + 5\widetilde{K}$ \\ 
    \midrule
        $\phasex{k}{n}{j}$&$\phasex{k'\!}{n}{j}$ & 
            $\frac{2}{\etan{n}}\Re\{
                    \left.\complexamplitude{\ncomponent}\right.^{\!\ast} \complexamplitude{\ncomponent'\!} 
                    \left.\signalatomnl{k}{n}{j}\right.^{\!\herm}
                    {\signalatomnl{k'\!}{n}{j}}
                \}$  &
            $\widetilde{k}+ 3\widetilde{K}$ &  
            $\widetilde{k}' + 3\widetilde{K}$ \\ 
        $\phasex{k}{n}{j}$&$\modulusx{k'\!}{n}{j}$ & 
            $\frac{2}{\etan{n}}\Re\{
                    - \mathrm{j}  \!\left.\complexamplitude{\ncomponent}\right.^{\!\ast} \!\!e^{\mathrm{j}\phasex{\ncomponent'\!}{n}{j}}
                    \!\left.\signalatomnl{k}{n}{j}\right.^{\!\herm}
                    \!{\signalatomnl{k'\!}{n}{j}}
                \}$  &
            $\widetilde{k}+ 3\widetilde{K}$ &  
            $\widetilde{k}' + 4\widetilde{K}$ \\ 
            $\phasex{k}{n}{j}$&$\etan{n}$ & 
            $0$  &
            $\widetilde{k}+ 3\widetilde{K}$ &  
            $1 + 5\widetilde{K}$ \\ 
            \midrule
        $\modulusx{k}{n}{j}$&$\modulusx{k'\!}{n}{j}$ & 
            $\frac{2}{\etan{n}}\Re\{
                    e^{-\mathrm{j}\phasex{k}{n}{j}}  e^{\mathrm{j}\phasex{k'\!}{n}{j}}
                    \left.\signalatomnl{k}{n}{j}\right.^{\!\herm}
                    {\signalatomnl{k'\!}{n}{j}}
                \}$  &
            $\widetilde{k}+ 4\widetilde{K} $ &  
            $\widetilde{k}'+ 4\widetilde{K} $ \\ 
            $\modulusx{k}{n}{j}$&$\etan{n}$ & 
            $0$  &
            $\widetilde{k}+ 4\widetilde{K} $ &  
            $1 + 5\widetilde{K}$ \\ 
            \midrule
            $\etan{n}$&$\etan{n}$ & 
            $\frac{\Nz }{\etan{n}^2}$  &
            $1 + 5\widetilde{K}$ &  
            $1 + 5\widetilde{K}$ \\
            \bottomrule
    \end{tabularx}
    \vspace{-8mm}
\end{table}
We use the indices $\grave{\imath}$ and $\acute{\imath}$ to describe the rows and columns, respectively, of $\left[\FIMch{n}{j}\right]_{\grave{\imath},\acute{\imath}}$ depending on their component numbers $\ncomponent,\ncomponent' \in \{0 \hdots {K}\}$.
The entries in Table~\ref{Stab:channel-fim} describe the upper triangular part of the channel \gls{fim}, while the lower triangular part is to be completed by exploiting the symmetry property of \glspl{fim}, i.e., 
$\left[\FIMch{n}{j}\right]_{\grave{\imath},\acute{\imath}}= \left[\FIMch{n}{j}\right]_{\acute{\imath},\grave{\imath}}$. 
Contrary to the unit-modulus array response from~\eqref{eq:unit-modulus-array-response},
with the carrier-phase term $\exp (\!-\mathrm{j}\frac{2\pi}{\lightspeed}\fc \lVert \acute{\V{r}}\rVert)$ absorbed into the amplitudes $\widetilde{{\varrho}}_{\!\scriptscriptstyle k,n}^{\scriptscriptstyle(j)}$, the array response used for deriving the \glspl{pcrlb} factorizes as 
\begin{align}\label{Seq:signalatom}
    \signalatom 
    \big(\delay,\elevation,\azimuth\big) \coloneqq 
        \bm{b}(\delay)\otimes 
        \bm{a}_y(\elevation,\azimuth) \otimes \bm{a}_z(\elevation)
        \quad \in \complexsetone{\Nz} \,,
\end{align}
again, under the assumption of the \gls{ura} layout from~\eqref{eq:Ptemplate} and per-anchor plane-wave propagation.
For notational brevity, the dependencies of $\bm{b}(\delay)$, 
$\bm{a}_y(\elevation,\azimuth)$, and $\bm{a}_z(\elevation)$ on their parameters are occasionally omitted.
To completely describe the terms in Table~\ref{Stab:channel-fim}, we further derive the following partial derivatives of 
the array response w.r.t. the channel parameters in $\etach{n}{j}$:
The derivatives of the array response in~\eqref{Seq:signalatom} w.r.t. incidence angles $(\elevation, \azimuth)$ and delay $\delay$ compute to 
\begin{align}
    \atomel\big(\delay,\elevation,\azimuth\big) \coloneqq \frac{\partial \signalatom}{\partial \elevation}
    &= \bm{b}  \otimes \dot{\bm{a}}_{y,\elevation} \otimes \bm{a}_z + 
     \bm{b} \otimes \bm{a}_y \otimes \dot{\bm{a}}_{z,\elevation} 
     \\
    \atomaz\big(\delay,\elevation,\azimuth\big) \coloneqq \frac{\partial \signalatom}{\partial \azimuth}   
    &=  \bm{b}  \otimes \dot{\bm{a}}_{y,\azimuth} \otimes \bm{a}_z 
    {+ \underbrace{ \bm{b}  \otimes \bm{a}_y \otimes \dot{\bm{a}}_{z,\azimuth}}_{=\mathbf{0}} } 
    \\[-13pt]
    \atomdelay\big(\delay,\elevation,\azimuth\big) \coloneqq \frac{\partial \signalatom}{\partial \delay}   
    &= \dot{\bm{b}} \otimes \bm{a}_y \otimes \bm{a}_z 
\end{align}
with the individual derivatives of temporal (denoted $\dot{\bm{b}}$) and spatial (denoted $\dot{\bm{a}}$) array responses below.
%
The partial derivative of the horizontal spatial array response in~\eqref{eq:spatial-response-y} w.r.t. elevation angle $\elevation$ becomes
\begin{align*}
    \dot{\bm{a}}_{y,\elevation}(\elevation,\azimuth)  = \frac{\partial \bm{a}_y}{\partial \elevation } 
    = 
    \mathrm{j} 
    \frac{2 \pi}{\lambda} 
    \cos (\elevation) 
    \sin (\azimuth) 
    \,\mathbf{p}_y 
    \!\odot\!
    \bm{a}_y(\elevation,\azimuth)  
    \in \complexsetone{\Nantennasy} ,
\end{align*}
and the partial derivative w.r.t. azimuth angle $\azimuth$ becomes
\begin{align*}
    \dot{\bm{a}}_{y,\azimuth}(\elevation,\azimuth)  = \frac{\partial \bm{a}_y}{\partial \azimuth }
    = 
    \mathrm{j} \frac{2 \pi}{\lambda} 
    \sin (\elevation) 
    \cos (\azimuth) 
    \,\mathbf{p}_y
    \!\odot\! 
    \bm{a}_y(\elevation,\azimuth)  
    \in \complexsetone{\Nantennasy}.
\end{align*}
%
The partial derivative of the vertical spatial array response in \eqref{eq:spatial-response-z} w.r.t. elevation angle $\elevation$ becomes
\begin{align*}
    \dot{\bm{a}}_{z,\elevation}(\elevation) = \frac{\partial \bm{a}_z}{\partial \elevation }
      =  -\mathrm{j} \frac{2 \pi}{\lambda}  \sin (\elevation) ~ \mathbf{p}_z 
    \odot \bm{a}_z(\elevation)    
    \,,
\end{align*}
and the partial derivative w.r.t. azimuth angle $\azimuth$ becomes
\begin{align}\label{Seq:derivative-vertical-response-wrt-azimuth}
    \dot{\bm{a}}_{z,\azimuth} = \frac{\partial \bm{a}_z}{\partial \azimuth }
    =  \bm{0} \quad \in \complexsetone{\Nantennasz} \,,
\end{align}
%
For the temporal response in \eqref{eq:delay-array-response}, the derivative w.r.t. delay is
\begin{align}
    \dot{\bm{b}}(\delay) = \frac{\partial \bm{b}}
    {\partial \delay }
    =  -\mathrm{j} 2 \pi \mathbf{f} \odot \bm{b}(\delay) \quad \in \complexsetone{\Nfrequency} \,.
\end{align}
We further define the shorthand notation
$\signalatomnl{k}{n}{j}\coloneqq\signalatom\big(\delayx{k}{n}{j},\elx{k}{n}{j},\azx{k}{n}{j}\big)$,
$\atomdelaynl{(j)}{k,n}\coloneqq\atomdelay\big(\delayx{k}{n}{j},\elx{k}{n}{j},\azx{k}{n}{j}\big)$,
$\atomelnl{(j)}{k,n}\coloneqq\atomel\big(\delayx{k}{n}{j},\elx{k}{n}{j},\azx{k}{n}{j}\big)$, and
$\atomaznl{(j)}{k,n}\coloneqq\atomaz\big(\delayx{k}{n}{j},\elx{k}{n}{j},\azx{k}{n}{j}\big)$.

\subsection{Derivation of Jacobian Matrices}\label{Ssec:Jacobian}  

As mentioned in Sec.\,\ref{sec:PCRLB}, the difference between (i) the noncoherent \gls{pcrlb} and the (ii) coherent \gls{pcrlb} is fundamentally tied to whether the $\widetilde{K}$ component phases in 
$\RVphasevec{n}{(j)}$ 
map to (i) nuisance phases 
$\unslant[-.25]{\varphi}_{\!\scriptscriptstyle k,n}^{\scriptscriptstyle(j)}$ \textit{separate} across \glspl{pa} $j$ in $\RVetaglobal{n}$, 
or (ii) to a single nuisance phase 
$\unslant[-.25]{\varphi}_{\!\scriptscriptstyle k,n}$ \textit{common} among all $J$ \glspl{pa}~\cite[Sec.\,IV]{Fascista23RadioStripesICC}.
Notably, the inversion in~\eqref{eq:posteriorCRLB} in the case of (i) the noncoherent \gls{pcrlb} eventually costs the information on the global parameters of interest that the $\widetilde{K}J$ phases of all $J$ \glspl{pa} in $\RVetach{n}{j}$ could have contributed in the first place, because they map to $\widetilde{K}J$ nuisance phases in $\RVetaglobal{n}$ that are different across the $J$ \glspl{pa} (cf.\,\cite[eq.\,(8.43)]{VanTrees2002optimumASP}).
In contrast, in the case of (ii) the coherent \gls{pcrlb}, the $\widetilde{K}J$ phases of all $J$ \glspl{pa} in $\RVetach{n}{j}$ map to only $\widetilde{K}$ nuisance phases in $\RVetaglobal{n}$, where the excess information not consumed in the estimation of nuisance phases can be used for positioning and mapping.
Notably, this information gain comes from retaining the jointly coherent aperture of a \gls{dmimo} infrastructure of distributed \glspl{pa}.
The parameter \textit{mapping} from local channel parameters $\RVetach{n}{j}$ 
to global parameters $\RVetaglobal{n}$
is contained in the Jacobian matrices $\jacobgn{j}$ that we derive in the following.
Mapping from local channel parameters $\etach{n}{j}$ to global parameters $\etaglobal{n}$, the $(\dimGlobal\!\times\!\dimLocal)$ Jacobian matrices are defined as\footnote{Note that the unusual definition of the Jacobian matrices $\jacobgn{j}\!=\! \nicefrac{\partial {\etach{n}{j}}^{\!\!\!\trp}}{\partial \etaglobal{n}} \!\in\!  \realset{\dimGlobal}{\dimLocal}$ comes from the fact that they are defined to propagate Fisher information (i.e., precision) rather than covariance.}
\begin{align}\label{eq:jacobian-main}
    \jacobgn{j} \!\coloneqq\!   
    \frac{\partial \etach{n}{j}\!^\trp }{\partial \etaglobal{n}} \!=\!
    \begin{bmatrixs}
        \jacobP{\elevation} & \jacobP{\azimuth}& \jacobP{\delay}& \jacobP{\varphi}& \bm{0} & 0 \\
        \bm{0} & \bm{0} & \bm{0}            &\bm{0}                 & \bm{0}    & 0 \\
        \jacobM{\elevation}& \jacobM{\azimuth}& \jacobM{\delay}& \jacobM{\varphi}& \bm{0} & 0 \\
        \bm{0} & \bm{0} & \bm{0}            & \jacobC{\varphi}      & \bm{0}    & 0 \\
        \bm{0} & \bm{0} & \bm{0}            & \bm{0}                & \jacobA{a}& 0 \\
        0 & 0 & 0            & 0            & 0    & 1 \\
    \end{bmatrixs}
\end{align}
with 
Jacobian submatrices $\jacobP{\cdot}$ mapping information from the respective channel parameters to the \gls{mt} position $\pos{n}$,
submatrices $\jacobM{\cdot}$ mapping information to \gls{sfv} positions $\psfv{k}{n}{j}$,
submatrix $\jacobC{\varphi}$ maps information to nuisance phases $\phasevec{n}{(j)}$,
and 
$\jacobA{a}$ maps information to nuisance amplitudes in $\modulivec{n}{(j)}$ in $\etaglobal{n}$.
The zero rows belonging to the \gls{mt} velocity $\vel{n}$ in the Jacobian $\jacobgn{j}$ make it rank-deficient and the respective global per-snapshot \gls{fim} $\FIMclassic{n}$ in~\eqref{eq:FIMclassic} singular. 
Information about the \gls{mt} velocity $\RVvel{n}$ in the posterior \gls{fim} $\FIMstep{n}{n}$ of~\eqref{eq:posteriorCRLB}---restoring its full rank---is then obtained only by propagating information about the position over the state-space model in~\eqref{eq:priorFIM} (via off-diagonal elements in $\transitionmatrix_{\text{\tiny a}}$).

\ifthenelse{\equal{\IEEEversion}{true}}
{
}{
\begin{table}
	\centering
    \caption{Spatial vectors in global or local Cartesian coordinates from~\cite{Deutschmann25Asilomar}.}%
	\label{tab:spatial-vectors}
    \begin{tabularx}{0.78\columnwidth}{ccc|ccc}
        \toprule 
        \multicolumn{3}{c|}{Global coordinates} & \multicolumn{3}{c}{Local coordinates}   \\
		Vector & from & to & Vector & from & to \\
    \midrule 
		$\range{k}{n}{j}$ 	& \gls{va} 	& \gls{pm} \\
		$\ranget{k}{n}{j}$ & \gls{pa} 	& \gls{vm} & $\rangep{k}{n}{j}$ 		& \gls{pa} 		& \gls{vm} \\
    \bottomrule
    \end{tabularx}
    \vspace{-7mm}
\end{table}
}

\subsubsection{Fundamental Jacobian Matrices}
As a result of the chain rule, the Jacobian submatrices in~\eqref{eq:jacobian-main} compute as the product of more fundamental Jacobian building blocks, which we derive first:

Abbreviating $\rx:=[\rangep{k}{n}{j}]_{\scriptscriptstyle 1}$, $\ry:=[\rangep{k}{n}{j}]_{\scriptscriptstyle 2}$, and $\rz:=[\rangep{k}{n}{j}]_{\scriptscriptstyle 3}$ for notational brevity, we derive the mapping of Fisher information in \textit{local} spherical \gls{pa} coordinates $\{\delayx{k}{n}{j},\elx{k}{n}{j},\azx{k}{n}{j}\}$ to \textit{local} Cartesian \gls{pa} coordinates $\rangep{k}{n}{j}\in\realsetone{3}$.
%
The mapping in elevation $\elx{k}{n}{j} = \arccos (\rz/ \lVert \rangep{k}{n}{j} \rVert)$ from \eqref{eq:elevation} is 
\begin{align}
    \frac{\partial \elx{k}{n}{j} }{\partial\rangep{k}{n}{j} } = \frac{1}{\lVert \rangep{k}{n}{j} \rVert^2 \, \sqrt{ \rx^2 + \ry^2} } 
    \begin{bmatrix}
        \rx \rz \\
        \ry \rz \\
        -\!\left(\rx^2 + \ry^2\right) 
    \end{bmatrix}
    \quad \in \realset{3}{1} \,,
\end{align}
in azimuth $\azx{k}{n}{j} = \arctantwo(\ry,\rx)$ from \eqref{eq:azimuth} 
\begin{align}
     \frac{\partial \azx{k}{n}{j} }{\partial\rangep{k}{n}{j} } = 
    \frac{1}{\rx^2 + \ry^2}
    \begin{bmatrix}
        -\ry \\
        \rx \\
        0
    \end{bmatrix} 
    \quad \in \realset{3}{1} \,.
\end{align}
The mapping of Fisher information in \textit{local} spherical \gls{pa} coordinates to \textit{local} Cartesian \gls{pa} coordinates in perceived delay $\delayx{k}{n}{j} =  \frac{\lVert \rangep{k}{n}{j} \rVert}{\lightspeed} $ from~\eqref{eq:delay} is 
\begin{align}
     \frac{\partial \delayx{k}{n}{j} }{\partial\rangep{k}{n}{j} } = \frac{\rangep{k}{n}{j}}{\lightspeed ~ \lVert \rangep{k}{n}{j} \rVert} \quad \in \realset{3}{1} \,,
\end{align}
and in carrier phase 
$\phasex{k}{n}{j} = -\frac{2\pi}{\lightspeed}\fc \lVert \acute{\V{r}}\rVert + \angle \varrho_{\!\scriptscriptstyle k,n}^{\scriptscriptstyle(j)}$
is 
\begin{align}
     \frac{\partial \phasex{k}{n}{j} }{\partial\rangep{k}{n}{j} } = 
    \frac{-2 \pi ~ \rangep{k}{n}{j}}{\lambda ~ \lVert \rangep{k}{n}{j} \rVert} \quad \in \realset{3}{1} \,,
\end{align}
 with wavelength $\lambda = \frac{\lightspeed}{\fc}$.
As mentioned above, we assume no mapping from amplitude moduli to \gls{mt} position, i.e., 
$\jacobP{\modulus} = \frac{\partial {\modulivec{n}{(j)}}^\trp}{\partial \pos{n}} \approx \mathbf{0}$ under the approximation $\frac{\partial g_{\scriptscriptstyle k,n}^{\scriptscriptstyle (j)}}{\partial \rangep{k}{n}{j}}\approx \mathbf{0}$, because their information contribution on the \gls{mt} position is negligible compared to the contributions from the other channel parameters 
$\big\{
\RVelVecx{n}{j}, 
\RVazVecx{n}{j},
\RVdelayVecx{n}{j}, 
\RVphasevec{n}{(j)} 
\big\}$.

The following Jacobian terms map between coordinate systems:
The mapping from $\rangep{k}{n}{j}$ defined in~\eqref{eq:rangepSupp} in \textit{local} Cartesian \gls{pa} coordinates to \textit{global} Cartesian coordinates is 
\begin{align}
    \frac{\partial\rangep{k}{n}{j}\!^\trp}{\partial\ranget{k}{n}{j}}  = \rotM{j} 
    \quad \in SO(3) \,,
\end{align}
due to the orthogonality of rotation matrices, i.e., $\rotM{j}^{-1} = \rotM{j}^\trp$.
Through~\eqref{eq:ranget} and~\eqref{eq:range}, information about the vector $\ranget{k}{n}{j}$, pointing from \gls{pa} $j$ to the \gls{vm}, maps to the \gls{mt} position $\pos{n}$ via
\begin{align}\label{eq:Jacob-ranget-pos}
    \frac{\partial\ranget{k}{n}{j}\!^\trp}{\partial\pos{n}}  = 
    \frac{\partial\range{k}{n}{j}\!^\trp}{\partial\pos{n}} \frac{\partial\ranget{k}{n}{j}\!^\trp}{\partial\range{k}{n}{j}} = \house{k}
    \quad \in \realset{3}{3} \,.
\end{align}
using the symmetry of Householder matrices, i.e., $\house{k}^\trp = \house{k}$, and knowing that information about $\range{k}{n}{j}$, pointing from \gls{va} position $\posVAdet{k}{j}$ to the \gls{mt}, maps to \gls{mt} position $\pos{n}$ simply through $\partial\range{k}{n}{j}\!^\trp/\partial\pos{n} = \eye{3}$ per~\eqref{eq:range}.

These fundamental Jacobian building blocks will reappear in the submatrix blocks of the Jacobians $\jacobgn{j}$ in~\eqref{eq:jacobian-main}, which we define next.

\subsubsection{Positioning Submatrices}       
Through the chain rule of derivatives, the Jacobians mapping to \gls{mt} position are now assembled by building chains of these fundamental building blocks.
Each component $\widetilde{\ncomponent}=\ncomponent+1\in\{1 \hdots \widetilde{K}\}$ impinging at anchor $j$ leads to a $(3\times1)$ Jacobian submatrix in \eqref{eq:jacobian-main} mapping Fisher information to \gls{mt} position $\pos{n}$ via
\begingroup             
\allowdisplaybreaks[4]
\begin{align}
    \left[\jacobP{\elevation} \right]_{\scriptscriptstyle :,\widetilde{\ncomponent}}
    &= 
    \frac{\partial\ranget{k}{n}{j}\!^\trp}{\partial\pos{n}}
    \frac{\partial\rangep{k}{n}{j}\!^\trp}{\partial\ranget{k}{n}{j}}
    \frac{\partial \elx{k}{n}{j} }{\partial\rangep{k}{n}{j} } \quad \in \realset{3}{1}  &&\\
    &= 
    \house{k} ~ \rotM{j} ~ 
    \frac{1}{\lVert \rangep{k}{n}{j} \rVert^2 \, \sqrt{ \rx^2 + \ry^2} } 
    \begin{bmatrix}
        \rx \rz \\
        \ry \rz \\
        -\!\left(\rx^2 + \ry^2\right) 
    \end{bmatrix}  & 
    \nonumber&\\
    \left[\jacobP{\azimuth} \right]_{\scriptscriptstyle :,\widetilde{\ncomponent}}
    &= 
    \frac{\partial\ranget{k}{n}{j}\!^\trp}{\partial\pos{n}}
    \frac{\partial\rangep{k}{n}{j}\!^\trp}{\partial\ranget{k}{n}{j}}
     \frac{\partial \azx{k}{n}{j} }{\partial\rangep{k}{n}{j} }\nonumber &&\\
    &= 
    \house{k} ~ \rotM{j} ~ 
    \frac{1}{\rx^2 + \ry^2}
    \begin{bmatrix}
        -\ry \\
        \rx \\
        0
    \end{bmatrix} \quad \in \realset{3}{1} & 
    &\\
    \left[\jacobP{\delay} \right]_{\scriptscriptstyle :,\widetilde{\ncomponent}}
    &= 
    \frac{\partial\ranget{k}{n}{j}\!^\trp}{\partial\pos{n}}
    \frac{\partial\rangep{k}{n}{j}\!^\trp}{\partial\ranget{k}{n}{j}}
    \frac{\partial \delayx{k}{n}{j} }{\partial\rangep{k}{n}{j} } \nonumber &&\\
    &= 
    \house{k} ~ \rotM{j} ~ 
    \frac{\rangep{k}{n}{j}}{\lightspeed ~ \lVert \rangep{k}{n}{j} \rVert} \quad \in \realset{3}{1}     \, & 
    &\\
    \left[\jacobP{\varphi} \right]_{\scriptscriptstyle :,\widetilde{\ncomponent}}
    &= 
    \frac{\partial\ranget{k}{n}{j}\!^\trp}{\partial\pos{n}}
    \frac{\partial\rangep{k}{n}{j}\!^\trp}{\partial\ranget{k}{n}{j}}
    \frac{\partial \phasex{k}{n}{j} }{\partial\rangep{k}{n}{j} } \nonumber &&\\
    &= 
    \house{k} ~ \rotM{j} ~ 
    \frac{-2 \pi ~ \rangep{k}{n}{j}}{\lambda ~ \lVert \rangep{k}{n}{j} \rVert} \quad \in \realset{3}{1}  \,.   
    &
\end{align}
\endgroup

\subsubsection{Mapping submatrices}           
\ifthenelse{\equal{\IEEEversion}{true}}
{
}{
    In~\cite{Deutschmann25Asilomar,Li25adaptiveDMIMOslam} we have introduced the \gls{meb} $\MEB{k,n}\!\coloneqq\!\sqrt{\big[\PCRLB\big]_{\scriptscriptstyle\bm{i}_k,\bm{i}_k}}$, which lower-bounds the estimation error on \gls{sfv} positions $\RVpsfv{k}{n}{j}$.
    In this work, we extend the respective \gls{pcrlb} to nondiagonal channel \glspl{fim} $\FIMch{n}{j}$ with elements from Table~\ref{Stab:channel-fim} and to 3D using the Jacobian matrices below.
}
For mapping, there is another fundamental Jacobian block that maps Fisher information in $\range{k}{n}{j}$ from \eqref{eq:range}
to \gls{sfv} position 
$\V{p}_{\!\scriptscriptstyle k}^{\text{\tiny sfv}}$ via~\cite{Deutschmann25Asilomar}
\begin{align}
      \frac{\partial {\range{k}{n}{j}}^\trp }{\partial\V{p}_{\!\scriptscriptstyle k}^{\text{\tiny sfv} }  } 
      =
      2 \frac{\posPA{j}{\V{p}_{\!\scriptscriptstyle k}^{\text{\tiny sfv}}}^\trp}{\lVert\V{p}_{\!\scriptscriptstyle k}^{\text{\tiny sfv}}\rVert^2} + 
      2 \frac{{\posPA{j}}^\trp \V{p}_{\!\scriptscriptstyle k}^{\text{\tiny sfv}}}{\lVert\V{p}_{\!\scriptscriptstyle k}^{\text{\tiny sfv}}\rVert^2} \house{k} -
        \eye{3} \eqqcolon \jacobMVAblock{}{} \,.
\end{align}
Partitioning the $(3\Ncomponents \!\times\! \widetilde{\Ncomponents})$ Jacobian matrices
\begin{align}
    \jacobM{\theta^{\text{\tiny ch}}} \!=\!
    \begin{bmatrixs}
        \mathbf{0}  & \jacobM{\theta^{\text{\tiny ch}}\!,1} & \dots & \mathbf{0}\\
        \vdots      & \vdots                                &  \ddots & \vdots \\
        \mathbf{0}  & \mathbf{0}  &  \dots & \jacobM{\theta^{\text{\tiny ch}}\!,{K}}
    \end{bmatrixs}\quad \in \realset{3\Ncomponents }{\widetilde{\Ncomponents}}
\end{align}
mapping to all \glspl{sfv} into $(3 \!\times\! 1)$ vectors $\jacobM{\theta^{\text{\tiny ch}}\!,{k}}$, each mapping to a single \gls{sfv} ${k}$, with $\theta^{\text{\tiny ch}}\!\in\!\{\elevation,\azimuth,\delay,\varphi\}$ denoting the respective local channel parameter,
each \gls{smc} ${\ncomponent}\! \in\! \{1\hdots{\Ncomponents}\}$ impinging at \gls{pa} $j$ 
maps Fisher information to 
\gls{sfv} position $\V{p}_{\!\scriptscriptstyle k}^{\text{\tiny sfv} }$ via
\begingroup             
\allowdisplaybreaks[4]
\begin{align}
    \jacobM{\elevation,{\ncomponent}} &= 
    \frac{\partial\range{k}{n}{j}\!^\trp}{\partial 
    \V{p}_{\!\scriptscriptstyle k}^{\text{\tiny sfv} }
    }
    \frac{\partial\ranget{k}{n}{j}\!^\trp}{\partial\range{k}{n}{j}}
    \frac{\partial\rangep{k}{n}{j}\!^\trp}{\partial\ranget{k}{n}{j}}
    \frac{\partial \elx{k}{n}{j} }{\partial\rangep{k}{n}{j} }  \\
    &=  
    \frac{\jacobMVAblock{}{} ~ \house{k} ~ \rotM{j}}{\lVert \rangep{k}{n}{j} \rVert^2 \, \sqrt{ \rx^2 + \ry^2} } 
    \begin{bmatrix}
        \rx \rz \\
        \ry \rz \\
        -\!\left(\rx^2 + \ry^2\right) 
    \end{bmatrix}\quad \in \realset{3}{1} \nn
    \\
    \jacobM{\azimuth,{\ncomponent}} &= 
    \frac{\partial\range{k}{n}{j}\!^\trp}{\partial\V{p}_{\!\scriptscriptstyle k}^{\text{\tiny sfv}}}
    \frac{\partial\ranget{k}{n}{j}\!^\trp}{\partial\range{k}{n}{j}}
    \frac{\partial\rangep{k}{n}{j}\!^\trp}{\partial\ranget{k}{n}{j}}
     \frac{\partial \azx{k}{n}{j} }{\partial\rangep{k}{n}{j} }  \\
    &= 
    \jacobMVAblock{s}{s'} ~ \house{k} ~ \rotM{j} ~ 
    \frac{1}{\rx^2 + \ry^2}
    \begin{bmatrix}
        -\ry \\
        \rx \\
        0
    \end{bmatrix} \quad \in \realset{3}{1} \nn
    \\
    \jacobM{\delay,{\ncomponent}} &= 
    \frac{\partial\range{k}{n}{j}\!^\trp}{\partial\V{p}_{\!\scriptscriptstyle k}^{\text{\tiny sfv}}}
    \frac{\partial\ranget{k}{n}{j}\!^\trp}{\partial\range{k}{n}{j}}
    \frac{\partial\rangep{k}{n}{j}\!^\trp}{\partial\ranget{k}{n}{j}}
    \frac{\partial \delayx{k}{n}{j} }{\partial\rangep{k}{n}{j} } \nonumber \\
    &= 
    \jacobMVAblock{s}{s'} ~ \house{k} ~ \rotM{j} ~ 
    \frac{\rangep{k}{n}{j}}{\lightspeed ~ \lVert \rangep{k}{n}{j} \rVert} \quad \in \realset{3}{1}     
    \\
    \jacobM{\varphi,{\ncomponent}} &= 
    \frac{\partial\range{k}{n}{j}\!^\trp}{\partial\V{p}_{\!\scriptscriptstyle k}^{\text{\tiny sfv}}}
    \frac{\partial\ranget{k}{n}{j}\!^\trp}{\partial\range{k}{n}{j}}
    \frac{\partial\rangep{k}{n}{j}\!^\trp}{\partial\ranget{k}{n}{j}}
    \frac{\partial \phasex{k}{n}{j} }{\partial\rangep{k}{n}{j} } \nonumber \\
    &= 
    \jacobMVAblock{s}{s'} ~ \house{k} ~ \rotM{j} ~ 
    \frac{-2 \pi ~ \rangep{k}{n}{j}}{\lambda ~ \lVert \rangep{k}{n}{j} \rVert} \quad \in \realset{3}{1}     \,.
\end{align}
\endgroup

\subsubsection{Nuisance Parameter Submatrices}       
The Jacobian matrices mapping from phases in $\etach{n}{j}$ to phases in $\etaglobal{n}$ are
\begin{align*}
    \jacobC{\varphi} = \frac{\partial \phasevec{n}{(j)}\!^\trp}{\partial \phasevec{n}{}} = 
    \begin{cases}
        \unitvector{j} \otimes \eye{\widetilde{\Ncomponents}}  &\in \binset{J\widetilde{\Ncomponents}}{\widetilde{\Ncomponents}} ~~ \text{noncoherent}
        \\
        ~~~~\eye{\widetilde{\Ncomponents}} &\in \binset{\widetilde{\Ncomponents}}{\widetilde{\Ncomponents}} ~~~~ \text{coherent}
    \end{cases}
\end{align*}
and the Jacobian matrices mapping from moduli are
\begin{align}
    \jacobA{a} = \frac{\partial \modulivec{n}{(j)}\!^\trp}{\partial \modulivec{n}{}} = \unitvector{j} \otimes \eye{\widetilde{\Ncomponents}} \quad \in \binset{J\widetilde{\Ncomponents}}{\widetilde{\Ncomponents}}\,,
\end{align}
with $\unitvector{j}\!\in\!\binsetone{J}$ denoting a unit vector where the $j$\textsuperscript{th} entry is one.

\bibliographystyle{IEEEtran}
\balance
\bibliography{IEEEabrv,bibliography,ThisPaper}